\documentclass[preprint,12pt]{elsarticle}
\usepackage{graphics,graphicx}
\usepackage{url}
\usepackage{subfigure}
\usepackage{stmaryrd}
\usepackage{wrapfig}
\usepackage{mathrsfs}
\usepackage{amsfonts,amsmath,amssymb,amsthm}
\usepackage{lineno,hyperref}
\modulolinenumbers[5]
\newtheorem{definition}{Definition}
\newtheorem{theorem}{Theorem}
\newtheorem{corollary}{Corollary}
\newtheorem{remark}{Remark}

\newtheorem{example}{Example}
\graphicspath{{./figure/}}

\begin{document}

\begin{frontmatter}

\title{Curve and surface construction based on the generalized toric-Bernstein basis
functions}

\author{Jing-Gai Li}
\ead{LiJingGaiM@mail.dlut.edu.cn}


\author[]{Chun-Gang Zhu\corref{mycorrespondingauthor}}
\cortext[mycorrespondingauthor]{Corresponding author}
\ead{cgzhu@dlut.edu.cn}

\address{School of Mathematical Sciences, Dalian
University of Technology, Dalian 116024, China.}

\begin{abstract}
The construction of parametric curve and surface plays important role in computer aided geometric design (CAGD), computer aided design (CAD), and geometric modeling. In this paper, we define a new  kind of blending functions associated with a real points set, called generalized toric-Bernstein (GT-Bernstein) basis functions. Then the generalized toric-Bezier (GT-B\'ezier) curves and surfaces  are constructed based on the GT-Bernstein basis functions, which are the projections of the (irrational) toric varieties in fact and the generalizations of the classical rational B\'ezier curves  and surfaces and toric surface patches. Furthermore, we also study the properties of the presented curves and surfaces, including the limiting properties of weights and knots. Some representative  examples verify the properties and results.
\end{abstract}


\end{frontmatter}


\section{Introduction}
\label{sec:1}
Representing curves and surfaces of geometric shapes is an essential but challenging task of Computer Aided Geometric Design (CAGD) and Computed Aided Design (CAD)~(\cite{ref1,ref2}). Curves and surfaces are commonly represented by the parametric method and the implicit method.  To be specifically speaking,  the parametric method is characterized by its advantages of easy plotting, generalization and splicing.  Parametric curves and surfaces have gone experienced the developments of polynomial parametric curve (eg. Ferguson curves), B\'ezier method, B-spline method and NURBS method~(\cite{ref1,ref2}). In this development process, the construction of basis functions(also called blending functions) of a curve and a surface plays a key role. The history of parametric curve and surface is essentially the history of the basis functions, which has gone through the polynomial power basis, Bernstein basis, B-spline basis, and their related rational and extended forms. Therefore, the basis functions are the core of curve and surface construction. A set of basis functions with good properties inevitably makes the parametric curve and surface possessing powerful vitality and application value.

In 1912, S.N. Bernstein~\cite{ref3} constructed a set of polynomial sequences to prove the Weierstrass approximation theorem. The key of Bernstein's proof is the construction of the basis functions, which are called Bernstein basis functions. In 1959, de Casteljau~\cite{ref4} firstly applied the Bernstein basis functions to the curve representation and then P. B\'ezier used them for geometric shape representation,  namely B\'ezier curve. The B\'ezier curve uses the control polygon to represent the curve, which is a major breakthrough in the curve representation method in CAGD and CAD. The B\'ezier method develops rapidly and becomes the core method of geometric shape representation because of its good geometric properties and algorithms. It is worth mentioning that the Bernstein basis functions and their generalizations are widely applied in various fields with its good properties, and thus become an important factor affecting mathematics, applied mathematics and related subjects in the 20th century~\cite{ref5}.

In recent years, the parametric curves and surfaces construction based on different basis functions have been studied by many scholars. Chen and Wang et al.~\cite{ref6,ref7} defined the C-B\'ezier curve and the C-B spline curve (NUAT B-spline curve) by extending the space of mixed algebra and trigonometric polynomial. The C-B\'ezier curve and C-B spline curve introduce shape parameters and adjust the shape of the curve by controlling the changes of control points and parameters. Oru\c{c} and Phillips~\cite{ref8} defined a $q$-B\'ezier curve based on the $q$-Bernstein operator which was constructed by Phillips~\cite{ref01}. They studied the rational $q$-B\'ezier curve and the tensor product $q$-B\'ezier surface, proved the properties of the curve and surface and gave the corresponding subdivision form further. Han et al.~\cite{ref9} constructed a new generalization of B\'ezier curves  and  the corresponding tensor product surfaces over the rectangular domain .These curves and surfaces are based on the Lupa\c{s} $q$-analogue of Bernstein operator. Compared with $q$-B\'ezier curves and surfaces based on Phillips $q$-Bernstein polynomials, these curves and surfaces show more flexibility in choosing the value of $q$ and superiority in shape control of curves and surfaces.  Schaback~\cite{ref10} gave an introduction to certain techniques for the construction of surfaces from scattered data, which emphasis is putting on interpolation methods using compactly supported radial basis functions. Goldman and Simeonov~\cite{ref11} studied the properties of quantum Bernstein bases and quantum B\'ezier curves  by introducing a new variant of the blossom.

In 2002, Krasauskas~\cite{ref12} proposed a new multisided surface modeling method namely toric surface based on the toric ideals and toric varieties defined by a given set of integer lattice points. The basis functions constructing the toric surface are called the toric-Bernstein basis functions (also called the toric-B\'ezier basis functions). The toric surface is degenerated into rational B\'ezier curve if the set of lattice points are constrained to a one-dimensional integer points, and the tensor product and the triangular B\'ezier surfaces are also special forms of the toric surface.  Garc\'ia-Puente et al.~\cite{ref13} explained the geometric significance of the weights of toric surfaces, which is called the toric degeneration property.

Most of the basis functions constructing curves and surfaces defined above are represented in non-negative integer power forms. At present, some researchers have put many efforts on the construction of curves and surfaces based on basis functions with rational or irrational number powers . The multiquadric(MQ) function is a radial basis function (RBF) with the rational number power form which is widely used in numerical analysis and scientific computing \cite{ref10}. In 2015, Zhu et al.~\cite{ref14} extended the Bernstein basis functions and then constructed $\alpha\beta$-Bernstein-like basis with two exponential shape parameters $\alpha$ and $\beta$ with real number degrees.

Garc\'ia-Puente and Sottile~\cite{ref19} showed that tuning a pentagonal toric patch by lattice points $\widetilde{\mathcal {A}}$ (see Fig.~\ref{fig1:subfig:b}) instead of $\mathcal {A}$ (see Fig.~\ref{fig1:subfig:a}) to achieve linear precision, where $\widetilde{\mathcal {A}}$ contains three non-integer points. In 2008, Craciun et al.~\cite{ref15} studied the theory of toric varieties defined by generally real lattice sets, which are applied in algebraic statistics known as toric model~\cite{ref02} and studied the geometric properties of toric surfaces. In 2015, Postinghel et al.~\cite{ref16} presented the degenerations of real irrational toric varieties defined by generally real number set. Li et al. studied the T-B\'ezier curve constructed by the real points preliminary in \cite{li2018}.

\begin{figure}[h!]
  \centering
  \subfigure[Lattice points $\mathcal {A}$]{
   \label{fig1:subfig:a} 
    \includegraphics[width=6cm]{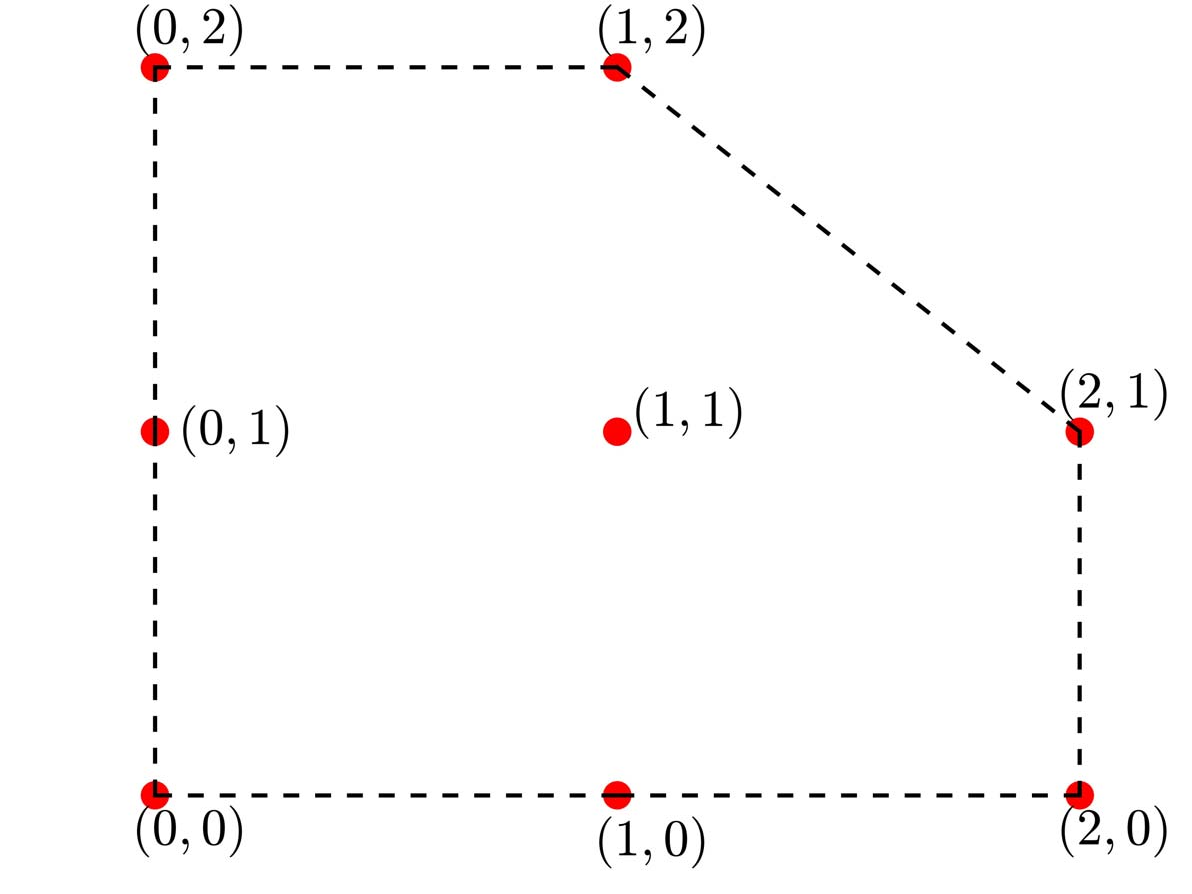}}
   \hspace{1cm}
  \subfigure[Lattice points $\widetilde{\mathcal {A}}$]{
    \label{fig1:subfig:b} 
    \includegraphics[width=6cm]{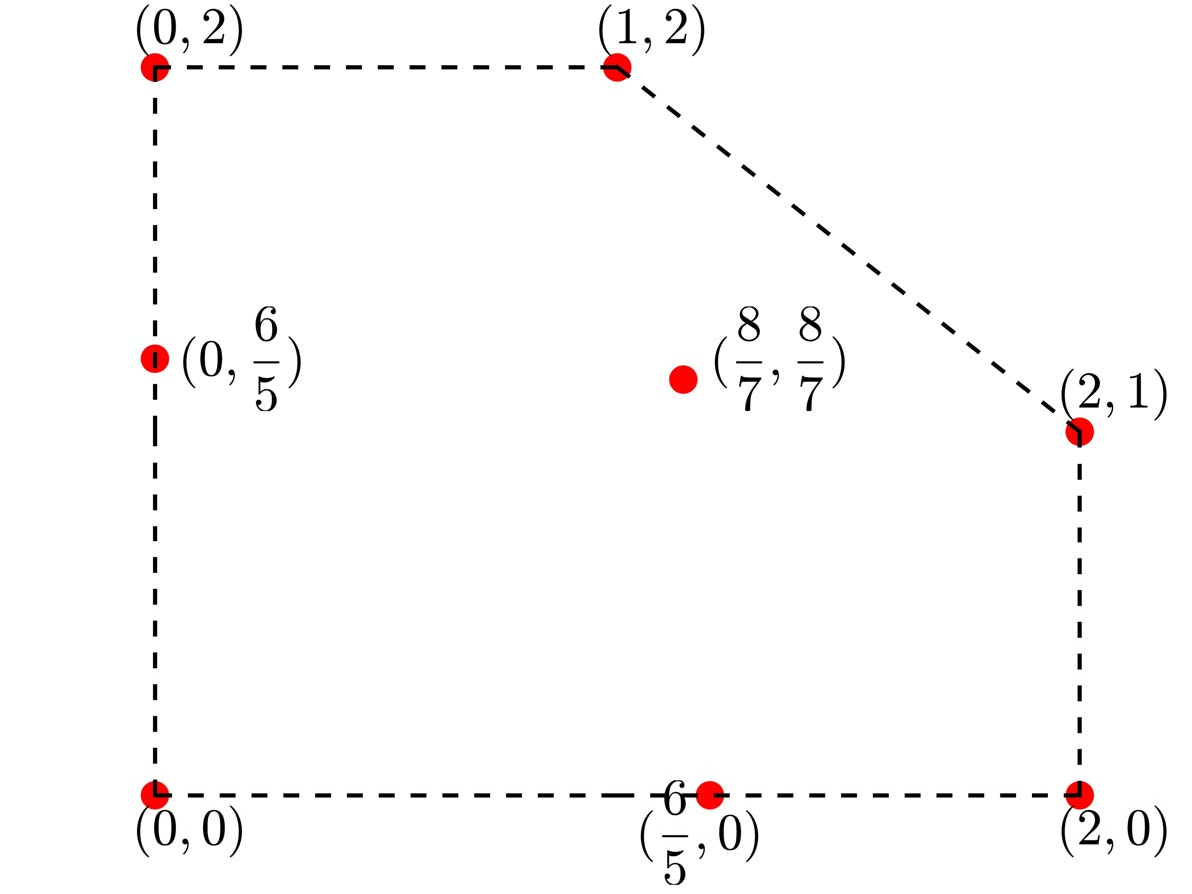}}
  \caption{Lattice points $\mathcal {A}$ and $\widetilde{\mathcal {A}}$}
  \label{fig:1}
\end{figure}

In this paper, inspired by above methods, especially by \cite{ref16}, we present a kind of generalized toric-Bernstein basis functions of real number powers for any given real number set, and then construct a new kind of parametric curve and multisided surface based on the generalized toric-Bernstein basis functions. The properties of presented curves and surfaces are also studied.

The rest of this paper is organized as follows. In Section 2, the generalized toric-Bernstein basis functions are defined and the properties of the basis are discussed. Then, a class of generalized B\'ezier curve is constructed in Section 3, which is the generalization of the classical rational B\'ezier curve. In Section 4, we construct a new kind of generalized toric surface by bivariate generalized toric-Bernstein basis functions. At last, we conclude the whole paper and point out the future work in Section 5.
\section{Generalized Toric-Bernstein Basis Functions}
\label{sec:2}
It is well known that toric Bernstein basis functions depend on the finite set of integer lattice points $\mathcal {A}$ and boundary functions of the convex hull(lattice polygon) $\Delta_{\mathcal {A}}$ of $\mathcal {A}$. When the lattice polygon $\Delta_{\mathcal {A}}$ is a standard triangle or a rectangle, if we take appropriate coefficients, the toric Bernstein basis functions degenerate into the classical Bernstein basis functions after the parameter transformation. So they are the generalizations of the classical Bernstein basis functions. In this section, we generalize the toric Bernstein basis functions to finite set of real points, and give the definitions of generalized toric-Bernstein basis functions in one and two dimensions.

\subsection{Univariate generalized toric-Bernstein basis functions}
\label{subsec:21}
Consider $\mathcal {A}=\{a_{0},a_{1},\cdots,a_{n}\}\subset \mathbb{R}$ with $a_{0}\leq a_{1}\leq \cdots\leq a_{n-1}\leq a_{n}$, and $\Delta_{\mathcal {A}}=[a_{0},a_{n}]$. Obviously, the endpoints of $\Delta_{\mathcal {A}}$ are points $a_{0}$ and $a_{n}$ and we assume $a_{0}<a_{n}$. Set $h_{0}(t)=k_{0}(t-a_{0})$ and $h_{1}(t)=k_{1}(a_{n}-t)$, where $k_{0},k_{1}$ are positive real numbers such that $h_{0}(t)\geq0,h_{1}(t)\geq0,t\in \Delta_{\mathcal {A}}$. Then basis functions indexed by $\mathcal {A}$ can be constructed as follow.

\begin{definition}\label{def:201}
Let $\mathcal {A}=\{a_{0},a_{1},\cdots,a_{n}\}\subset \mathbb{R}$ with $a_{0}\leq a_{1}\leq \cdots\leq a_{n-1}\leq a_{n}$ and $a_{0}<a_{n}$. Then, for any point $a_{i}$ in $\mathcal {A}$, we define the generalized toric-Bernstein (GT-Bernstein) basis functions as
\begin{equation}\label{equ:201}
  \beta_{a_{i}}(t)=c_{a_{i}}h_{0}(t)^{h_{0}(a_{i})}h_{1}(t)^{h_{1}(a_{i})},~~t\in\Delta_{\mathcal {A}},
\end{equation}
where coefficient $c_{a_{i}}>0$ and $a_{i}$ is called knot.
\end{definition}

The rational form of the GT-Bernstein basis $\beta_{a_{i}}(t)$ is
\begin{equation}\label{equ:202}
  T_{a_{i}}(t)=\frac{\omega_{a_{i}}\beta_{a_{i}}(t)}{\sum^{n}_{i=0}\omega_{a_{i}}\beta_{a_{i}}(t)},~~~~t\in\Delta_{\mathcal {A}},
\end{equation}
where $\omega_{a_{i}} > 0$ is called weight.

\begin{remark}\label{rem:201}
The GT-Bernstein basis $\{\beta_{a_{i}}(t)\}$ defined by equation (\ref{equ:201}) depends on the selection of the coefficients $k_{0}$ and $k_{1}$. Since any positive real numbers can be selected, if there is no special explanation, we set $k_{0}=k_{1}=1$. We will show that the curve defined by $\{\beta_{a_{i}}(t)\}$ is independent of $k_{0}$ and $k_{1}$ in Section 3.
\end{remark}

It can be seen that the GT-Bernstein basis can be degenerated into the classical univariate Bernstein basis by simple transformation if $a_{i}=i,c_{a_{i}}=\frac{1}{n^{n}}\binom{n}{i}(i=0,1,\cdots,n)$. For $a_{i}=\frac{i}{n}(i=0,1,\cdots,n)$, let $k_{0}=k_{1}=n$ and select coefficients properly, then the GT-Bernstein basis degenerates into the univariate Bernstein basis too. For $\mathcal {A}=\{a_{0},a_{1},\cdots,a_{n}\}\subset\mathbb{Z}$, the basis functions defined by (\ref{equ:201}) degenerate into the toric-Bernstein basis functions defined in \cite{ref12}. Therefore, the GT-Bernstein basis is the generalization of Bernstein basis and toric-Bernstein basis.

\begin{example}\label{exa:201}
Let $a_{0}=0,a_{1}=\frac{\sqrt{2}}{4},a_{2}=\frac{1}{2},a_{3}=\frac{\sqrt{2}}{2},a_{4}=1$ and $c_{a_{0}}=\frac{1}{2},c_{a_{1}}=1,c_{a_{2}}=\frac{3}{2},c_{a_{3}}=\frac{7}{10},c_{a_{4}}=\frac{9}{10}$. By (\ref{equ:201}), we have
\begin{equation*}
  \beta_{a_{0}}(t)=\frac{1}{2}(1-t),~\beta_{a_{1}}(t)=t^{\frac{\sqrt{2}}{4}}(1-t)^{1-\frac{\sqrt{2}}{4}},
\end{equation*}
\begin{equation*}
  \beta_{a_{2}}(t)=\frac{3}{2}t^{\frac{1}{2}}(1-t)^{\frac{1}{2}},~\beta_{a_{3}}(t)=\frac{7}{10}t^{\frac{\sqrt{2}}{2}}(1-t)^{1-\frac{\sqrt{2}}{2}},~\beta_{a_{4}}(t)=\frac{9}{10}t,
\end{equation*}
and the basis functions $\beta_{a_{i}}(t)$ on $\Delta_{\mathcal {A}}=[0,1]$ are shown in Fig.~\ref{fig:201}.
The changes of basis function $\beta_{a_{2}}(t)$ while coefficient $c_{a_{2}}$ varying as shown in Fig.~\ref{fig:202} (the coefficients of curves from bottom to top are $0.1,0.7,1.5,1.9$ respectively), which shows that the coefficient mainly affects the function value of the basis function at each point. However, the changes of the basis function $\beta_{a_{2}}(t)$ when its corresponding knot changes are shown in Fig.~\ref{fig:203} (the knots corresponding to curves from left to right are $\frac{\sqrt{2}}{5},\frac{1}{2},\frac{\sqrt{5}}{3}$ respectively), which means that the knot mainly affect the positions of the maximum point of the basis function.
\begin{figure}[h!]
\begin{minipage}[t]{0.26\textwidth}
\centering
\includegraphics[width=4.8cm]{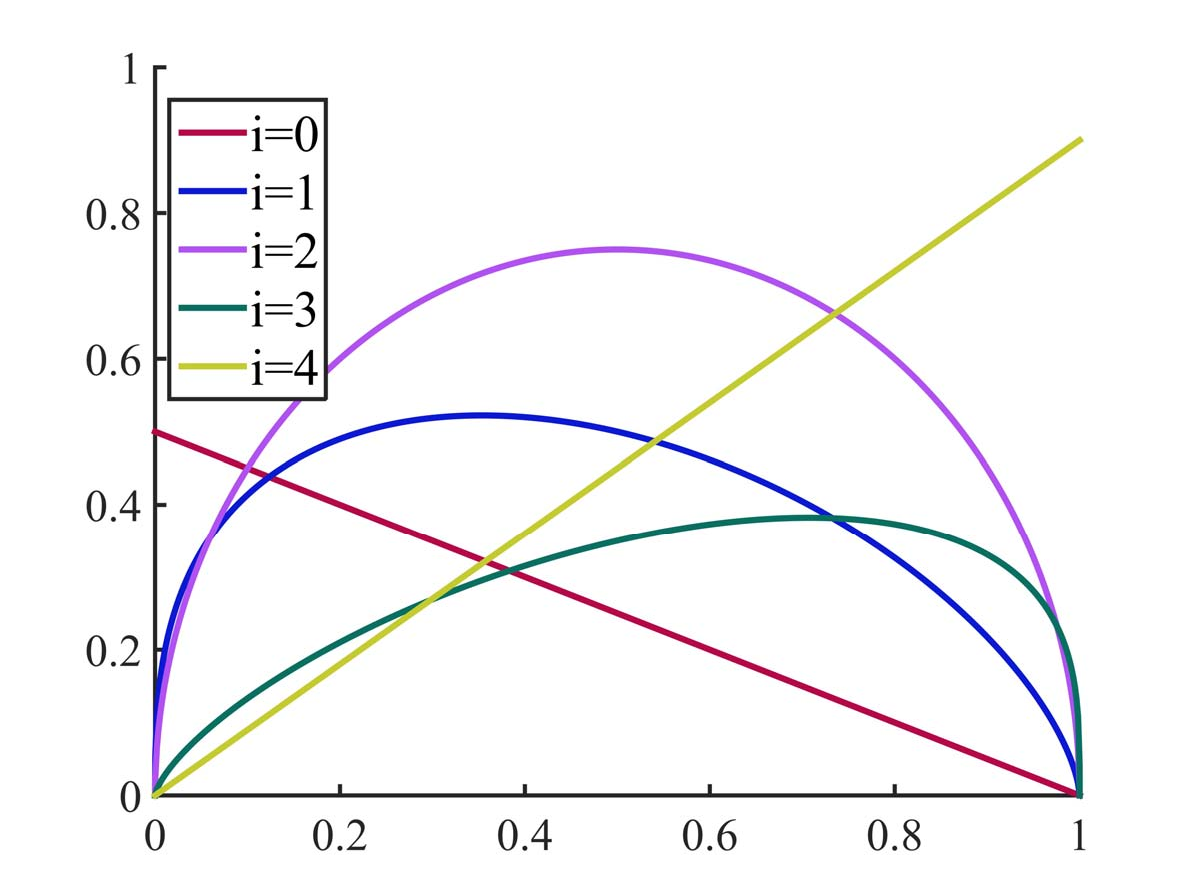}
\caption{GT-Bernstein basis}
\label{fig:201}
\end{minipage}%
\hspace{1.0cm}
\begin{minipage}[t]{0.26\textwidth}
\centering
\includegraphics[width=4.8cm]{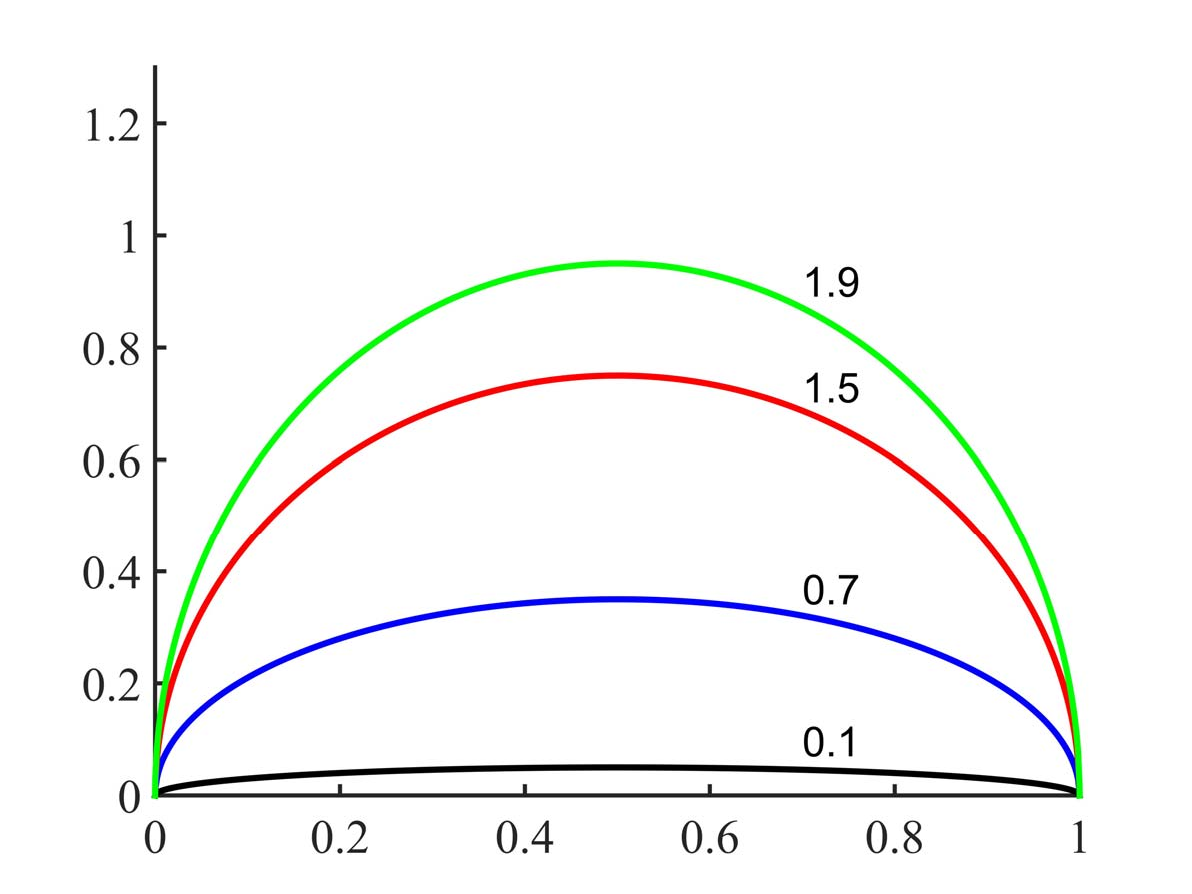}
\caption{Effect of coefficient changing on GT-Bernstein basis}
\label{fig:202}
\end{minipage}%
\hspace{1.0cm}
\begin{minipage}[t]{0.3\textwidth}
\centering
\includegraphics[width=4.8cm]{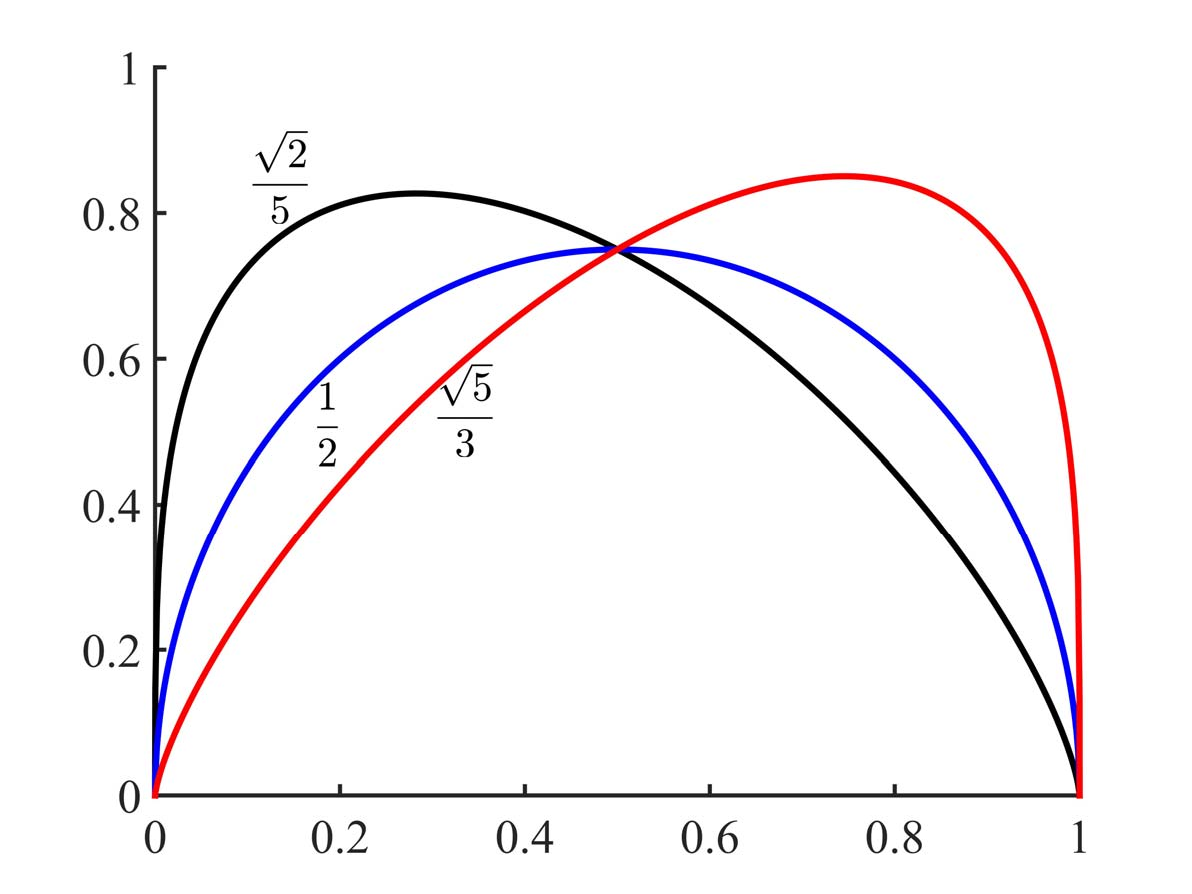}
\caption{Effect of knot changing on GT-Bernstein basis}
\label{fig:203}
\end{minipage}
\end{figure}
\end{example}

From the Definition \ref{def:201} and rational form (\ref{equ:202}), some properties of the basis functions $\{T_{a_{i}}(t)\}$ can be obtained directly as follows.

\begin{theorem}\label{the:ogtbp}
The rational GT-Bernstein basis functions defined in (\ref{equ:202}) have the following properties:
\begin{description}
  \item[(a)] Nonnegativity. $T_{a_{i}}(t)\geq 0, t\in\Delta_{\mathcal {A}}, i=0,1,\cdots,n$.
  \item[(b)] Partition of the unity. $\sum_{i=0}^{n}T_{a_{i}}(t)\equiv 1$.
  \item[(c)] Normalized totally positive (NTP). The rational GT-Bernstein basis $\{T_{a_{i}}(t)\}_{i=0}^{n}$ is a NTP basis. This property is proved recently by Yu et al.~\cite{ref17}.
  \item[(d)] Endpoints property. At the endpoints of $[a_{0},a_{n}]$, we have
      \begin{equation*}
       T_{a_{i}}(a_{0})=
        \left\{
             \begin{array}{ll}
             1,& i=0,\\
             0,& i\neq 0,
             \end{array}
         \right.
         ~~~~~T_{a_{i}}(a_{n})=
        \left\{
             \begin{array}{ll}
             1,& i=n,\\
             0,& i\neq n,
             \end{array}
         \right.
      \end{equation*}
  \item[(e)] Degeneration property.  The GT-Bernstein basis $\{\beta_{a_{i}}(t)\}$ degenerates to the classical Bernstein basis for $\mathcal {A}=\{0,1,\cdots,n\}$ or $\mathcal {A}=\{0,\frac{1}{n},\cdots,1\}$ after proper parameter transformation, and to toric-Bernstein basis for $\mathcal{A}\subset \mathbb{Z}$. Therefore, the rational GT-Bernstein basis degenerates to rational Bernstein basis for $\mathcal {A}=\{0,1,\cdots,n\}$ or $\mathcal {A}=\{0,\frac{1}{n},\cdots,1\}$ after proper parameter transformation.
\end{description}
\end{theorem}
Yu et al.~\cite{ref17} presented the following result for GT-Bernstein basis.

\begin{theorem}\label{the:gtbtp}
Suppose $k_{0}=k_{1}=k$ and set $a_{0}\leq t_{0}< t_{1}< \cdots < t_{n}\leq a_{n}$ to be an any increasing sequence. Then the collocation matrix of $\{\beta_{a_{i}}(t)\}_{i=0}^{n}$ at $t_{0}< t_{1}< \cdots < t_{n}$
\begin{equation}\label{equ:gtbcm}
M\left(\begin{array}{c}
\beta_{a_{0}},\cdots,\beta_{a_{n}}\\
t_{0},\cdots,t_{n}
\end{array}\right)
=(\beta_{a_{j}}(t_{i}))_{j=0,1,\cdots,n}^{i=0,1,\cdots,n}
\end{equation}
is a strictly totally positive matrix.
\end{theorem}

Since the basis $\{\beta_{a_{i}}(t)\}^{n}_{i=0}$ defined by equation (\ref{equ:201}) may do not hold the property of partition of the unity on $\Delta_{\mathcal {A}}$ for arbitrary positive coefficients, we present a method to choose coefficients by Theorem \ref{the:gtbtp}, which makes the basis $\{\beta_{a_{i}}(t)\}$ has partition of the unity on a given increasing sequence $a_{0}\leq t_{0}< t_{1}< \cdots < t_{n}\leq a_{n}$.

Given an increasing sequence $a_{0}\leq t_{0}< t_{1}< \cdots < t_{n}\leq a_{n}$, we have the following system of equations:
\begin{equation}\label{equ:gtbose}
  \sum_{j=0}^{n}\beta_{a_{j}}(t_{i})=c_{a_{0}}h_{0}(t_{i})^{h_{0}(a_{0})}h_{1}(t_{i})^{h_{1}(a_{0})}+\cdots+c_{a_{n}}h_{0}(t_{i})^{h_{0}(a_{n})}h_{1}(t_{i})^{h_{1}(a_{n})}=1,~~i=0,\cdots,n.
\end{equation}
If we write $\mathbf{C}=(c_{a_{0}},\cdots,c_{a_{n}})^{T}$ and $\mathbf{1}=(1,\cdots,1)^{T}$, then we obtain
\begin{equation}\label{equ:gtbosem}
  M\left(\begin{array}{c}
\beta_{a_{0}},\cdots,\beta_{a_{n}}\\
t_{0},\cdots,t_{n}
\end{array}\right)\mathbf{C}=\mathbf{1}.
\end{equation}
It's clear that the basis $\{\beta_{a_{i}}(t)\}^{n}_{i=0}$ satisfies the conditions of Theorem \ref{the:gtbtp}, then the matrix $M$ is a strictly totally positive matrix and system of equations (\ref{equ:gtbosem}) has a unique solution. For the bivariate generalized toric-Bernstein basis in Section \ref{subsec:22}, the method for selection of the coefficients is similar to the univariate case.
\subsection{Bivariate generalized toric-Bernstein basis functions}
\label{subsec:22}
Consider a finite set of real points $\mathcal {A}=\{\mathbf{a}_{0},\mathbf{a}_{1},\cdots,\mathbf{a}_{n}\}\subset\mathbb{R}^{2}$, Let $\Delta_{\mathcal {A}}$ be the convex hull of $\mathcal {A}$. The lines defined by edges $\phi_{i}$ of $\Delta_{\mathcal {A}}$ are $h_{i}(u,v)=\xi_{i}u+\eta_{i}v+\rho_{i}$, where $\langle \xi_{i},\eta_{i} \rangle$ is the normal vector of $\phi_{i}$ towards inside of $\Delta_{\mathcal {A}}$ such that $h_{i}(u,v)\geq0, (u,v) \in \Delta_{\mathcal {A}},i=1,\cdots,r$. We construct the generalized toric-Bernstein basis functions as follows.

\begin{definition}\label{def:202}
Let $\mathcal {A}=\{\mathbf{a}_{0},\mathbf{a}_{1},\cdots,\mathbf{a}_{n}\}\subset\mathbb{R}^{2}$ be a finite collection of real points, and set $\Delta_{\mathcal {A}}$ to be the convex hull of $\mathcal {A}$. Then, for any point $\mathbf{a}_{i}$ in $\mathcal {A}$, we call
\begin{equation}\label{equ:204}
  \beta_{\mathbf{a}_{i}}(u,v)=c_{\mathbf{a}_{i}}h_{1}(u,v)^{h_{1}(\mathbf{a}_{i})} \cdots h_{r}(u,v)^{h_{r}(\mathbf{a}_{i})},~~(u,v)\in\Delta_{\mathcal {A}},
\end{equation}
is the bivariate generalized toric-Bernstein (GT-Bernstein) basis function, where $c_{\mathbf{a}_{i}}>0$ is the coefficient and $\mathbf{a}_{i}$ is called knot.
\end{definition}

The rational form of the GT-Bernstein basic function $\beta_{\mathbf{a}_{i}}(u,v)$ is
\begin{equation}\label{equ:205}
  T_{\mathbf{a}_{i}}(u,v)=\frac{\omega_{\mathbf{a}_{i}}\beta_{\mathbf{a}_{i}}(u,v)}{\sum^{n}_{i=0}\omega_{\mathbf{a}_{i}}\beta_{\mathbf{a}_{i}}(u,v)},~~(u,v)\in\Delta_{\mathcal {A}}.
\end{equation}
where $\omega_{\mathbf{a}_{i}} > 0$ is called weight.

\begin{remark}\label{rem:202}
In (\ref{equ:204}), the basis function depends on the choice of coefficients, and the coefficients can vary from case to case. If there is no special explanation, we set $c_{\mathbf{a}_{i}}\equiv 1$.

For $\mathcal {A}=\{\mathbf{a}_{0},\mathbf{a}_{1},\cdots,\mathbf{a}_{n}\}\subset\mathbb{Z}^{2}$, the basis $\{\beta_{\mathbf{a}_{i}}(u,v)\}$ defined by (\ref{equ:204}) degenerates to the toric-Bernstein basis in \cite{ref12}. In particular, for $\mathcal {A}=\{(i,j)\in \mathbb{Z}^{2}\mid i+j\leq k,i\geq0,j\geq0 \}$, if we choose proper coefficients, then the GT-Bernstein basis degenerates to the bivariate triangular Bernstein basis. Analogously, for $\mathcal {A}=\{(i,j)\in \mathbb{Z}^{2}\mid 0\leq i\leq m,0\leq j\geq n\}$, the GT-Bernstein basis degenerates to the bivariate tensor product Bernstein basis if coefficients chosen properly.
\end{remark}

\begin{example}\label{exa:gtbt}
Let $\widetilde{{\mathcal{A}}}=\{(0,2),(1,2),(0,\frac{6}{5}),(\frac{8}{7},\frac{8}{7}),(2,1),(0,0),(\frac{6}{5},0),(2,0)\}$ (see Fig.~\ref{fig1:subfig:b}). By (\ref{equ:204}), the GT-Bernstein basis defined by $\widetilde{{\mathcal{A}}}$ are given as below.
\begin{equation*}
  \beta_{(0,2)}(u,v)\!\!=\!\!(3-u-v)(2-u)^{2}v^{2},~\beta_{(1,2)}(u,v)\!\!=\!\!(2-u)v^{2}u,~\beta_{(0,\frac{6}{5})}(u,v)\!\!=\!\!(2-v)^{\frac{4}{5}}(3-u-v)^{\frac{9}{5}}(2-u)^{2}v^{\frac{6}{5}},
\end{equation*}
\begin{equation*}
  \beta_{(\frac{8}{7},\frac{8}{7})}\!(\!u,v\!)\!\!=\!\!(2-v)\!^{\frac{6}{7}}\!(3-u-v)\!^{\frac{5}{7}}\!(2-u)^{\frac{6}{7}}\!v^{\frac{8}{7}}\!u^{\frac{8}{7}},~\beta_{(2,1)}\!(\!u,v\!)\!\!=\!\!(2-v)vu^{2},~\beta_{(0,0)}\!(\!u,v\!)\!\!=\!\!(2-v)\!^{2}\!(3-u-v)\!^{3}\!(2-u)^{2},
\end{equation*}
\begin{equation*}
  \beta_{(\frac{6}{5},0)}(u,v)\!=\!(2-v)^{2}(3-u-v)^{\frac{9}{5}}(2-u)^{\frac{4}{5}}u^{\frac{6}{5}},~\beta_{(2,0)}(u,v)\!=\!(2-v)^{2}(3-u-v)u^{2}.
\end{equation*}

Three of the basis functions are shown in Fig.~\ref{fig:gtbt1}. We further set each weight $\omega_{\mathbf{a}_{i}}=1$, then the rational forms of these three basis functions on $\Delta_{\mathcal {A}}$ are shown in Fig.~\ref{fig:gtbt2}.
\begin{figure}[h!]
\centering
  \subfigure[$\beta_{(0,2)}(u,v)$]{
    \label{figgtbt1:subfig:a} 
    \includegraphics[width=4.5cm]{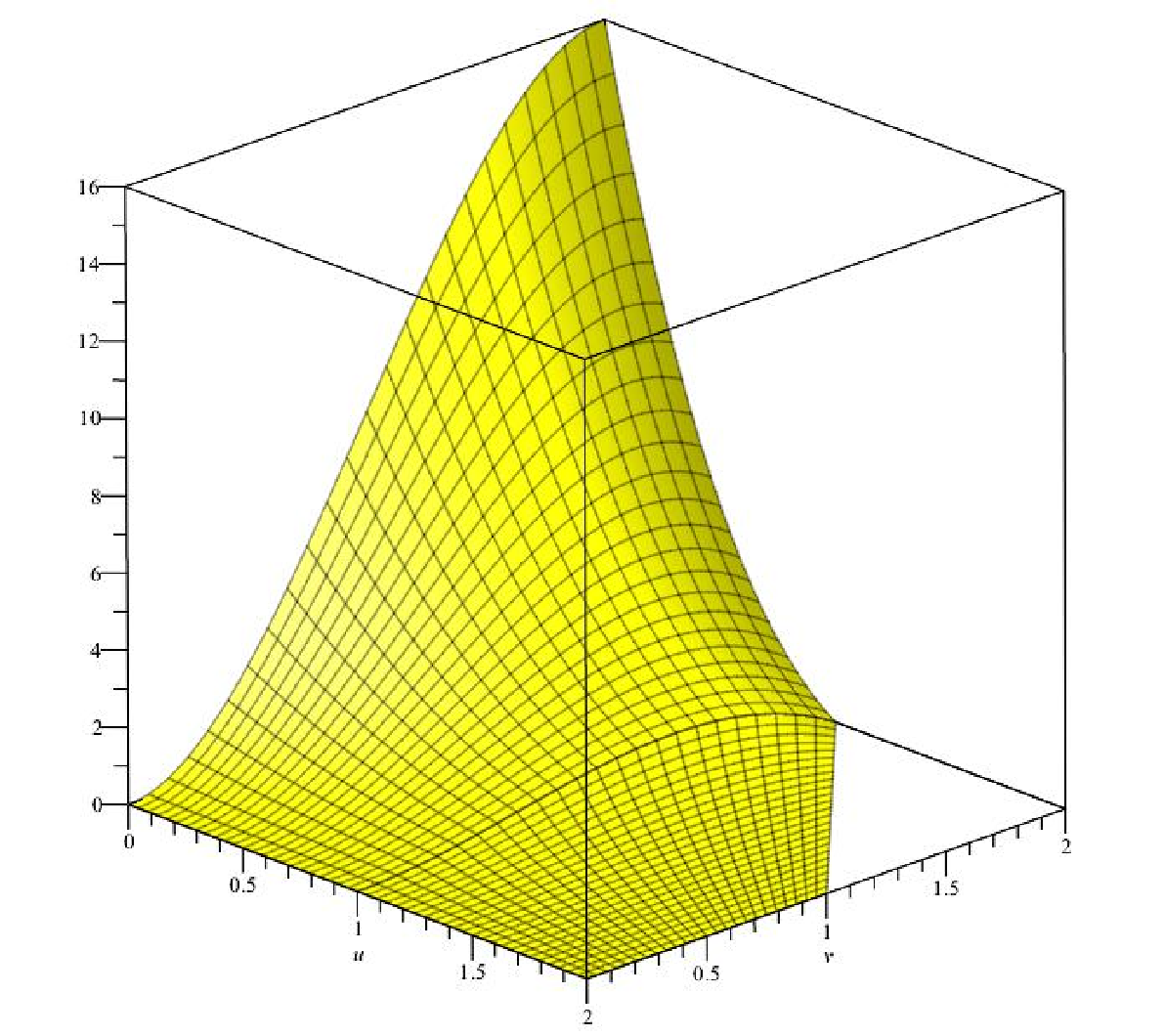}}
 \hspace{1cm}
  \subfigure[$\beta_{(\frac{8}{7},\frac{8}{7})}(u,v)$]{
   \label{figgtbt1:subfig:b} 
    \includegraphics[width=4.5cm]{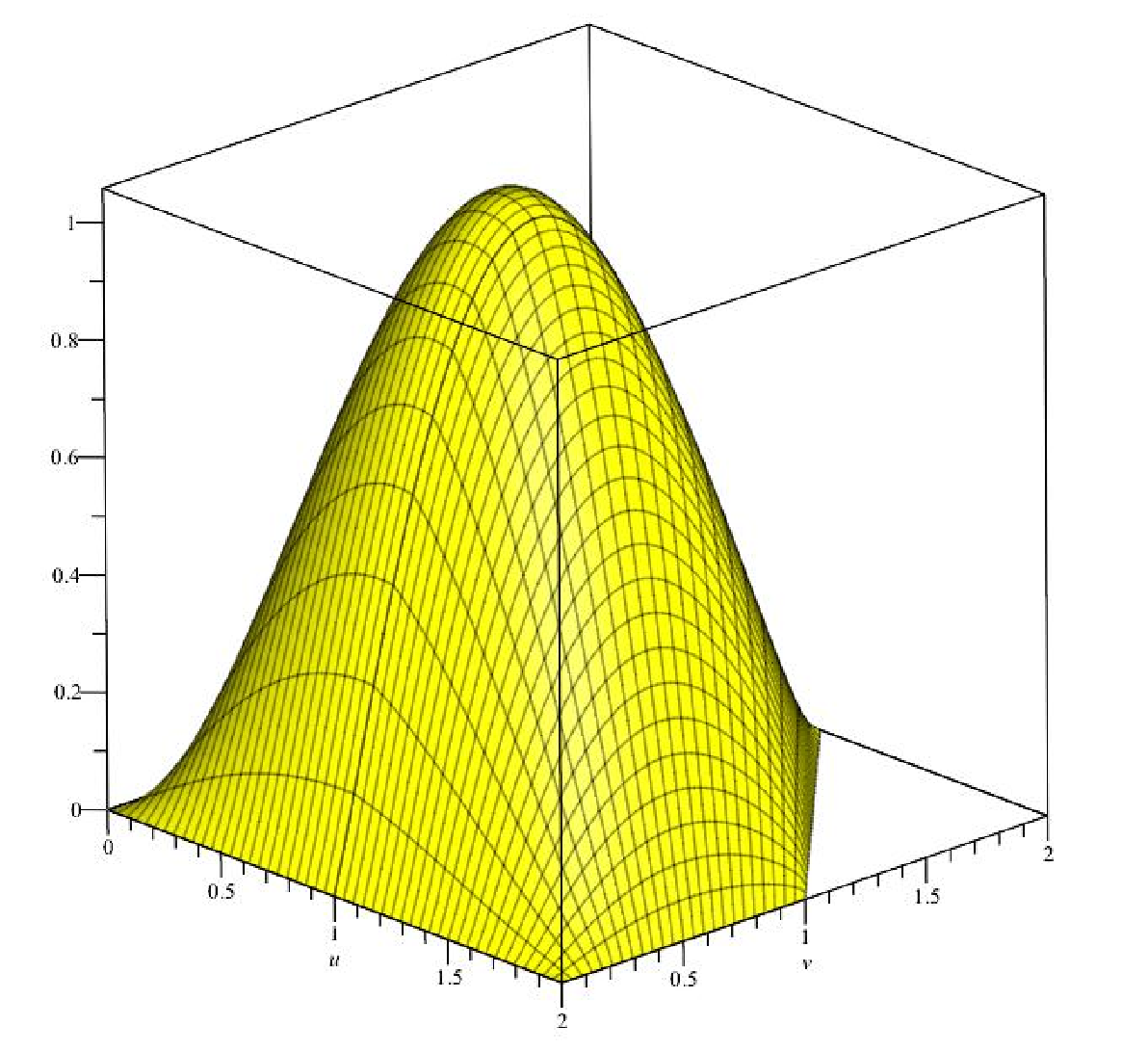}}
  \hspace{1cm}
  \subfigure[$\beta_{(\frac{6}{5},0)}(u,v)$]{
    \label{figgtbt1:subfig:c} 
    \includegraphics[width=4.5cm]{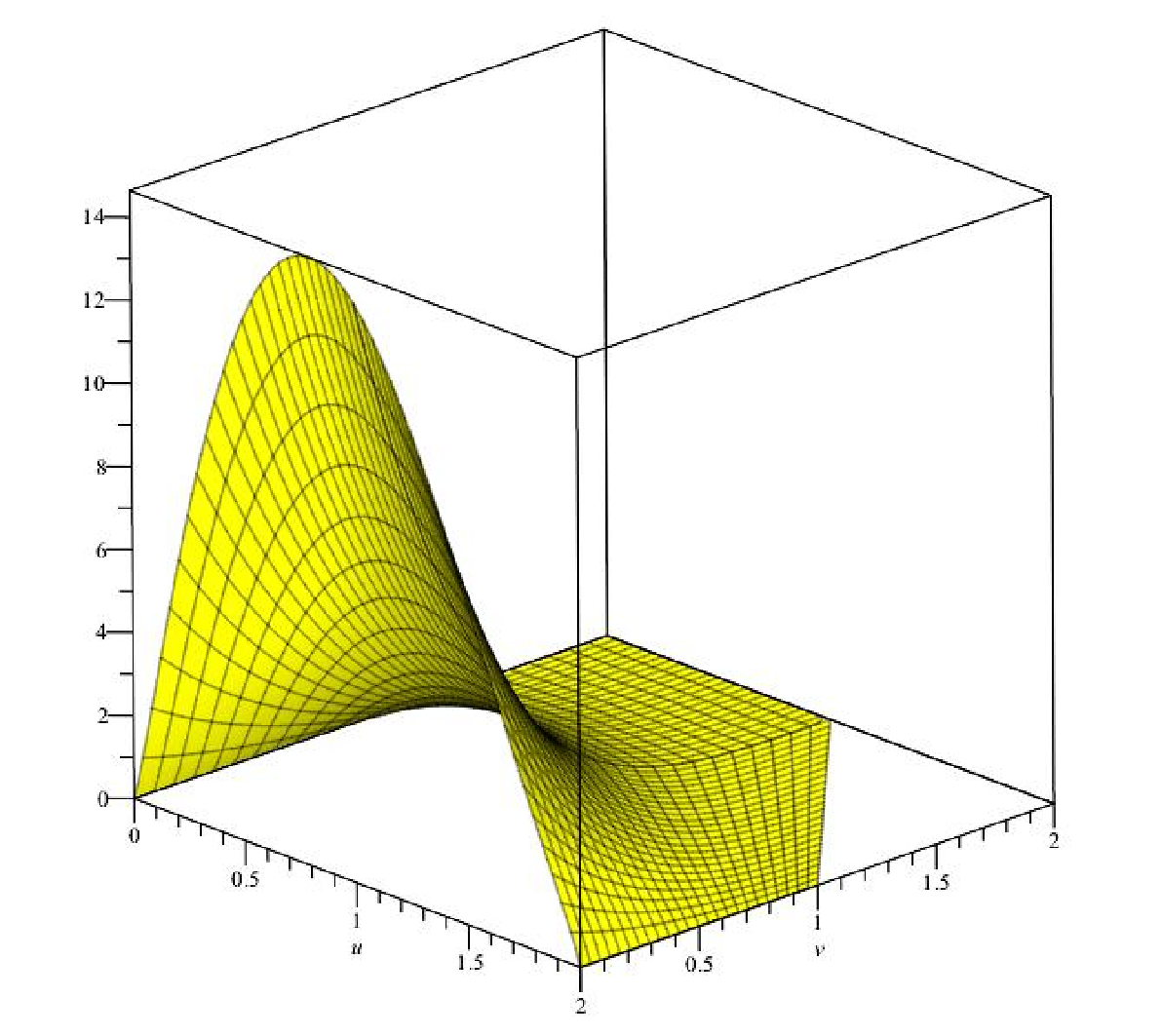}}
  \caption{GT-Bernstein basis functions}
  \label{fig:gtbt1}
\end{figure}

\begin{figure}[h!]
\centering
  \subfigure[$T_{(0,2)}(u,v)$]{
    \includegraphics[width=4.5cm]{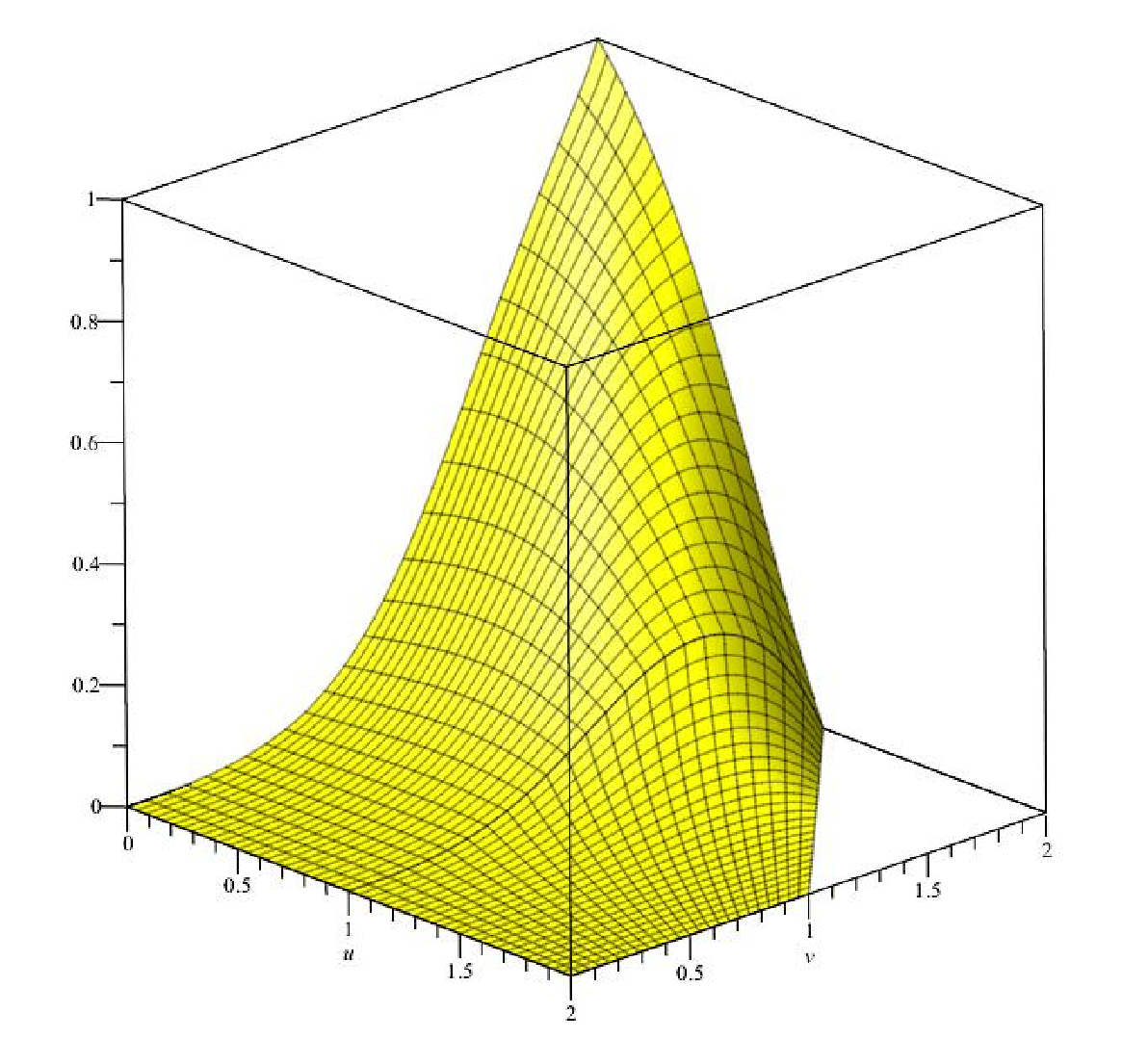}}
  \hspace{1cm}
  \subfigure[$T_{(\frac{8}{7},\frac{8}{7})}(u,v)$]{
    \includegraphics[width=4.5cm]{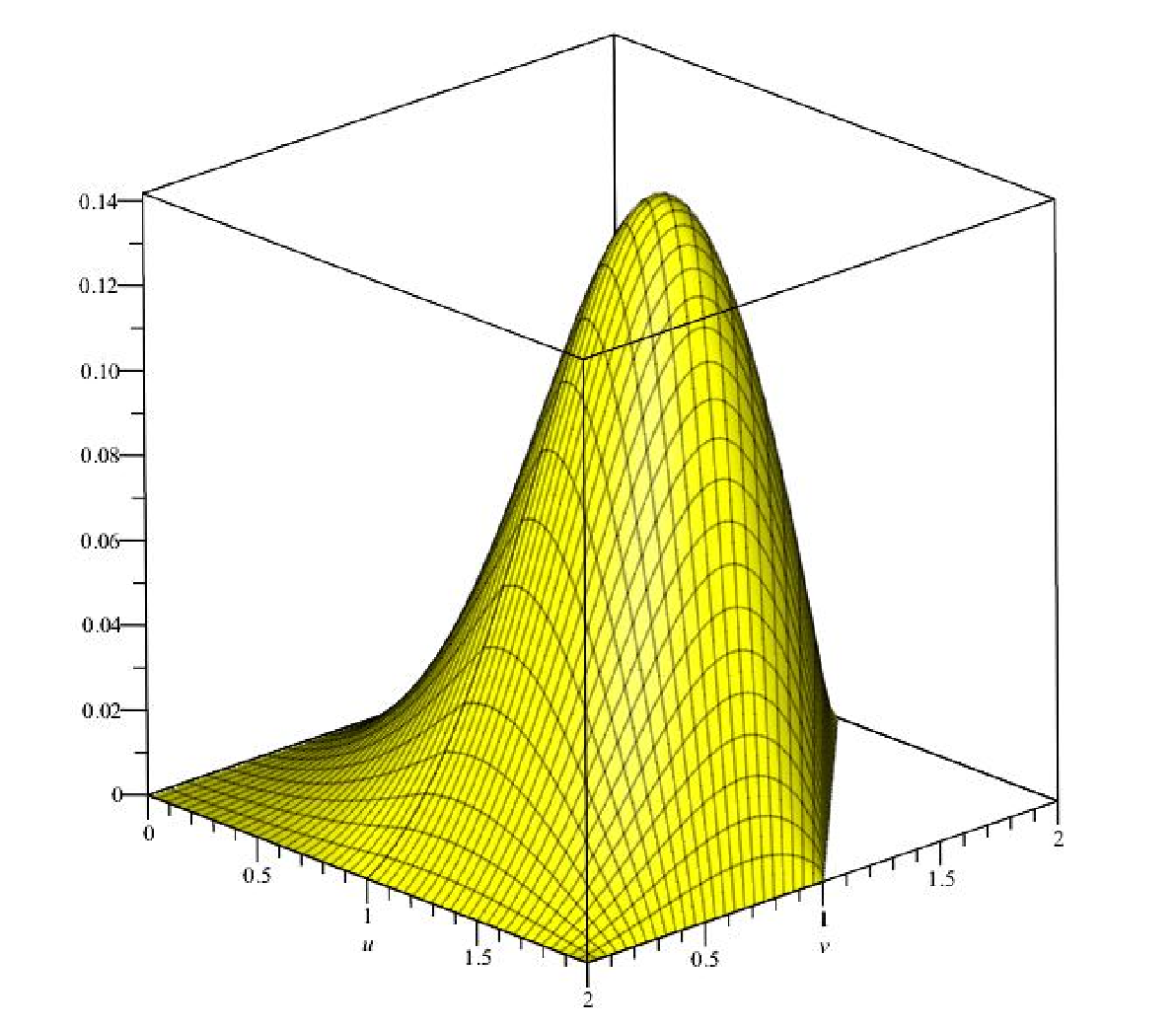}}
   \hspace{1cm}
  \subfigure[$T_{(\frac{6}{5},0)}(u,v)$]{
    \includegraphics[width=4.5cm]{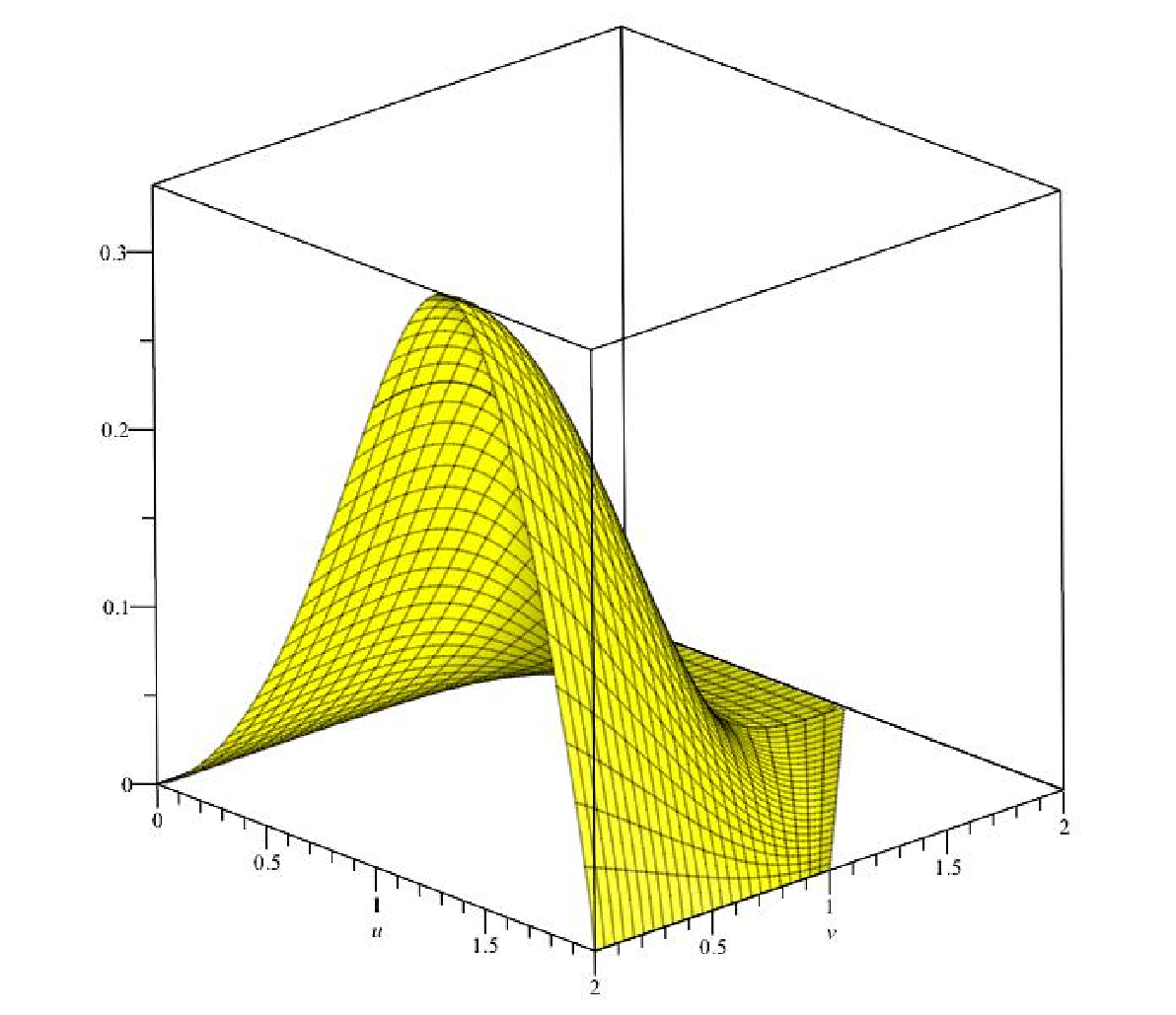}}
  \caption{Rational GT-Bernstein basis functions}
  \label{fig:gtbt2}
\end{figure}
\end{example}

Suppose edges $\phi_{i}(i=1,\cdots,r)$ of the convex hull $\Delta_{\mathcal {A}}$ are ordered counterclockwise
and let $V_{i}$ be vertex of $\Delta_{\mathcal {A}}$ where two edges $\phi_{i}$ and $\phi_{i+1}$ meet, $(i =1,\cdots,r)$.
The indices will be treated in a cyclic fashion: for instance, $\phi_{0}=\phi_{r}$, $\phi_{r+1} = \phi_{1}$ and so on. Denote by $\hat{\phi_{i}}=\phi_{i}\cap \mathcal {A}$ the intersection of $\mathcal {A}$ and $\phi_{i}$. Note that $\{V_{i}\}^{r}_{i=1}$ and $\hat{\phi_{i}}$ are subsets of $\mathcal {A}$ respectively, $i=1,\cdots,r$.

From Definition \ref{def:202} and rational form (\ref{equ:205}), we can obtain the following properties of the basis functions $\{T_{\mathbf{a}_{i}}(u,v)\}$ directly.

\begin{theorem}\label{the:201}
The rational forms of the GT-Bernstein basis functions defined in (\ref{equ:205}) have the following properties:
\begin{description}
  \item[(a)] Nonnegativity. $T_{\mathbf{a}_{i}}(u,v)\!\geq \!0,(u,v)\!\!\in\!\!\Delta_{\mathcal {A}}, i=0,1,\!\cdots\!,n$.
  \item[(b)] Partition of the unity. $\sum_{i=0}^{n}T_{\mathbf{a}_{i}}(u,v)\equiv 1$.
  \item[(c)] Boundary property. When $(u,v)$ is constrained on the edge $\phi_{j}$ of $\Delta_{\mathcal {A}}$, all basis functions $\beta_{\mathbf{a}_{i}}(u,v)$ and $T_{\mathbf{a}_{i}}(u,v)$ with indices $\mathbf{a}_{i}\in \mathcal {A}\setminus \hat{\phi_{j}}$ vanish, that is:
      \begin{equation}\label{equ:206}
        \begin{split}
        \left\{
             \begin{array}{lr}
             \beta_{\mathbf{a}_{i}}(u,v)=0,~\mathbf{a}_{i}\in \mathcal {A}\setminus \hat{\phi_{j}}, &  \\
              & (u,v)\in \phi_{j}.\\
             \beta_{\mathbf{a}_{i}}(u,v)\neq 0,~\mathbf{a}_{i}\in \hat{\phi_{j}}, &
             \end{array}
        \right.\\
        \left\{
             \begin{array}{lr}
             T_{\mathbf{a}_{i}}(u,v)=0,~\mathbf{a}_{i}\in \mathcal {A}\setminus \hat{\phi_{j}}, &  \\
              & (u,v)\in \phi_{j}.\\
             T_{\mathbf{a}_{i}}(u,v)\neq 0,~\mathbf{a}_{i}\in \hat{\phi_{j}}, &
             \end{array}
        \right.
        \end{split}
      \end{equation}
  \item[(d)] Corner points property. At the vertices of $\Delta_{\mathcal {A}}$, we have
     \begin{equation}\label{equ:207}
        \left\{
             \begin{array}{lr}
             T_{\mathbf{a}_{i}}(V_{i})=1, ~\mathbf{a}_{i}=V_{i},  \\
             T_{\mathbf{a}_{i}}(V_{i})= 0, ~\mathbf{a}_{i}\neq V_{i}.
             \end{array}
        \right.
      \end{equation}
  \item[(e)] Degeneration property. For $\mathcal {A}=\{\mathbf{a}_{0},\mathbf{a}_{1},\cdots,\mathbf{a}_{n}\}\subset\mathbb{Z}^{2}$, the basis defined by (\ref{equ:204}) degenerates into toric-Bernstein basis defined in~\cite{ref12}. In particular, the GT-Bernstein basis degenerates to the bivariate triangular Bernstein basis for $\mathcal {A}=\{(i,j)\in \mathbb{Z}^{2}\mid i+j\leq k,i\geq0,j\geq0 \}$, and to the tensor product Bernstein basis for $\mathcal {A}=\{(i,j)\in \mathbb{Z}^{2}\mid 0\leq i\leq m,0\leq j\geq n\}$, if coefficients selected properly.
\end{description}
\end{theorem}

\section{Generalized Toric-B\'ezier Curves}
\label{sec:3}
For given control points and weights, we can use the Bernstein basis functions to construct the classical rational B\'ezier curve. The classical rational B\'ezier curve has many good properties, such as convex hull property, boundary property, and affine invariance. In the same way, the basis functions defined by (\ref{equ:202}) can be used to define a new class of rational curves.

\begin{definition}\label{def:301}
Given real points set $\mathcal {A}=\{a_{0},a_{1},\cdots,a_{n}\}$, control points $\mathcal{B}=\{\mathbf{b}_{a_{i}}\mid a_{i}\in\mathcal {A}\}\subset \mathbb{R}^{3}$, and weights $ \omega=\{\omega_{a_{i}}>0\mid a_{i}\in\mathcal {A}\}$, the rational parametric curve
\begin{equation}\label{equ:301}
   \mathbf{P}_{\mathcal {A},\omega,\mathcal{B}}(t)=\sum_{i=0}^{n}\mathbf{b}_{a_{i}}T_{a_{i}}(t)=\sum_{i=0}^{n}\mathbf{b}_{a_{i}}\frac{\omega_{a_{i}}\beta_{a_{i}}(t)}{\sum^{n}_{i=0}\omega_{a_{i}}\beta_{a_{i}}(t)},~~~~t\in\Delta_{\mathcal {A}}.
\end{equation}
is called the generalized toric-B\'ezier curve (GT-B\'ezier curve for short) of degree $n$.The n-edge polyline polygon is obtained  by sequentially connecting two adjacent control points of $\mathcal{B}$ with a straight line segment, is called control polygon.
\end{definition}

\begin{remark}\label{rem:301}
Although the GT-Bernstein basis defined by equation (\ref{equ:201}) depends on the selection of the coefficients $k_{0}$ and $k_{1}$, the GT-B\'ezier curve is independent on the choice of these two parameters. It can be known from the results in \cite{ref16}, the GT-B\'ezier curve defined by the equation (\ref{equ:301}) is obtained by the projection (the projection is related to the weights and the control points) of the high-dimensional real projective toric variety defined by the $\mathcal {A}=\{a_{0},a_{1},\cdots,a_{n}\}$. Given point set $\mathcal {A}$, for different coefficients $k_{0}$ and $k_{1}$,  after unitizing the corresponding toric variety and eliminating the constant in the projective space, the toric varieties are identical, then the GT-B\'ezier curve defined by point set $\mathcal {A}$ is also the same. For more Details refer to~\cite{ref13, ref16}.

The degree of $n$ of curve in Definition \ref{def:301} is just the number of forms in the curve, one less than the number of knots of $\mathcal {A}$, not exactly the polynomial degree of curve in general sense. If $\mathcal {A}\subset \mathbb{Z}$, then this degree is exactly the polynomial degree of curve .
\end{remark}

\begin{example}\label{exa:301}
Let $\mathcal {A}=\{a_{0}=0,a_{1}=\frac{\sqrt{2}}{4},a_{2}=\frac{1}{2},a_{3}=\frac{\sqrt{2}}{2},a_{4}=1\}$ as show in Example \ref{exa:201} , weights $\omega_{a_{0}}=1,\omega_{a_{1}}=10,\omega_{a_{2}}=20,\omega_{a_{3}}=6,\omega_{a_{4}}=5$ and control points $\mathbf{b}_{a_{0}}=(0,0),\mathbf{b}_{a_{1}}=(0.4,1.3),\mathbf{b}_{a_{2}}=(2,2),\mathbf{b}_{a_{3}}=(3.7,1.5),\mathbf{b}_{a_{4}}=(4,0)$. Suppose $c_{a_{i}}=1(i=0,\cdots,4)$, then the quadratic GT-B\'ezier curve is
\begin{equation*}
   \mathbf{P}_{\mathcal {A},\omega,\mathcal{B}}(t)=\sum_{i=0}^{4}\mathbf{b}_{a_{i}}T_{a_{i}}(t),~t\in[0,1],
\end{equation*}
and the curve is shown in Fig.~\ref{fig:301}.
\begin{figure}[h!]
\begin{center}
\includegraphics[width=7cm]{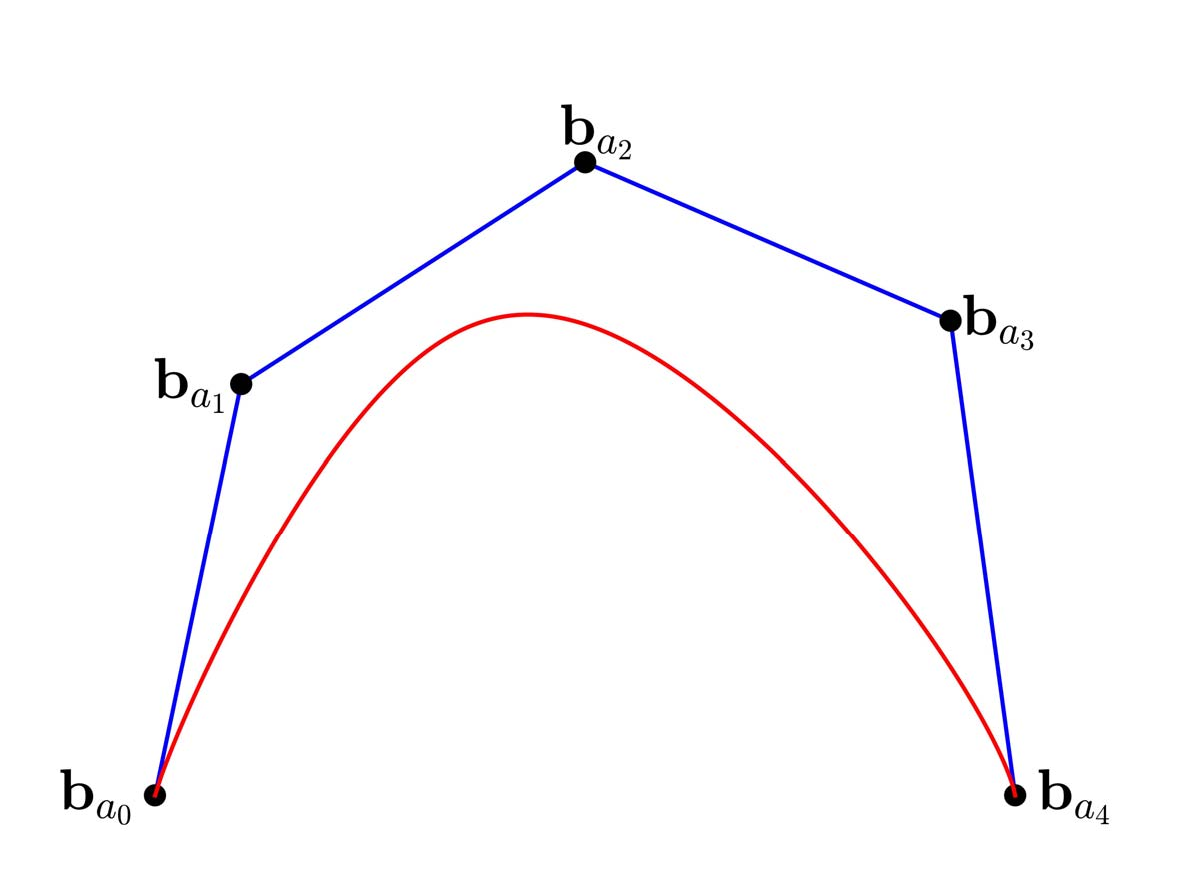}
\caption{Quadratic GT-B\'ezier curve}
\label{fig:301}
\end{center}
\end{figure}
\end{example}

From the properties of the GT-Bernstein basis functions associated with $\mathcal {A}=\{a_{0},a_{1},\cdots,a_{n}\}\subset \mathbb{R}$, some properties of the GT-B\'ezier curve can be obtained as follows:

\begin{description}
  \item[(a)] {\bf{Affine invariance and convex hull property}}. Since the basis (\ref{equ:202}) have the properties of nonnegativity and partition of the unity, these show that the corresponding GT-B\'ezier curve (\ref{equ:301}) has affine invariance and convex hull property.
  \item[(b)] {\bf{Endpoints interpolation property}}. This property follows directly from the corner points property of the basis (\ref{equ:202}), that is $\mathbf{P}_{\mathcal {A},\omega,\mathcal{B}}(a_{0})=\mathbf{b}_{a_{0}},\mathbf{P}_{\mathcal {A},\omega,\mathcal{B}}(a_{n})=\mathbf{b}_{a_{n}}$.
  \item[(c)] {\bf{Progressive iteration approximation (PIA) property}}. The GT-B\'ezier curve has PIA property from the result in \cite{ref20} because its basis $\{T_{a_{i}}(t)\}_{i=0}^{n}$ is a NTP basis.
  \item[(d)] {\bf{Degeneration property}}. If $a_{i}=i$ (or $a_{i}=\frac{i}{n}$), $(i=0,1,\cdots,n)$ and $k_{0}=k_{1}=n$, then the GT-B\'ezier curve (\ref{equ:301}) degenerates into the classical rational B\'ezier curve after reparameterization and coefficients selected properly. For $\mathcal {A}=\{a_{0},a_{1},\cdots,a_{n}\}\subset\mathbb{Z}$, the GT-B\'ezier curve (\ref{equ:301}) is the toric B\'ezier curve defined in~\cite{ref18}, which is exactly the one-dimensional form of the toric surface defined in~\cite{ref12}.
  \item[(e)] {\bf{Endpoints tangent vectors}}.   
%
%
  For
  \begin{equation*}
    k_{0}=\frac{1}{a_{1}-a_{0}}, k_{1}=\frac{1}{a_{n}-a_{n-1}},
  \end{equation*}
  the tangent vectors at the end points of curve GT-B\'ezier curve (\ref{equ:301}) are
    \begin{equation}\label{equ:302}
  \begin{split}
    &\mathbf{P}_{\mathcal {A},\omega,\mathcal{B}}^{\prime}(a_{0})\!\!=\!\!\frac{c_{a_{1}}k_{0}k_{1}^{-\frac{k_{1}}{k_{0}}}(a_{n}\!\!-\!\!a_{0})^{-\frac{k_{1}}{k_{0}}}\omega_{a_{1}}(\mathbf{b}_{a_{1}}\!\!-\!\!\mathbf{b}_{a_{0}})}{c_{a_{0}}\omega_{a_{0}}},\\
    &\mathbf{P}_{\mathcal {A},\omega,\mathcal{B}}^{\prime}(a_{n})\!\!=\!\!\frac{c_{a_{n-1}}k_{1}k_{0}^{-\frac{k_{0}}{k_{1}}}(a_{n}\!\!-\!\!a_{0})^{-\frac{k_{0}}{k_{1}}}\omega_{a_{n\!-\!1}}(\mathbf{b}_{a_{n}}\!\!-\!\!\mathbf{b}_{a_{n-1}})}{c_{a_{n}}\omega_{a_{n}}}.\\
  \end{split}
  \end{equation}

  We can see the tangent vectors at the end points of curve $\mathbf{P}_{\mathcal {A},\omega,\mathcal{B}}(t)$ are parallel to $\overrightarrow{\mathbf{b}_{a_{0}}\mathbf{b}_{a_{1}}}$ and $\overrightarrow{\mathbf{b}_{a_{n-1}}\mathbf{b}_{a_{n}}}$ respectively. And this property can be used to construct $G^{1}$ continuous piecewise GT-B\'ezier curve.
  \item[(f)] {\bf{Multiple knot property}}. When a knot in $\mathcal {A}$ tends to its adjacent knot, the following results describe the limit property of the GT-B\'ezier curve, which also show the resulting GT-B\'ezier curve defined by $\mathcal {A}$ with multiple knots.

  \begin{theorem}\label{the:301}
      Suppose $c_{a_{i}}=1(i=0,\cdots,n)$. When knot $a_{k}~(0\leq k<n)$ approaches to its adjacent knot $a_{k+1}$, the limit of GT-B\'ezier curve $\mathbf{P}_{\mathcal {A},\omega,\mathcal{B}}(t)$ of degree $n$ defined in (\ref{equ:301}) is exactly the GT-B\'ezier curve $\widetilde{\mathbf{P}}_{\widetilde{\mathcal {A}},\widetilde{\omega},\widetilde{\mathcal{B}}}(t)$ of degree $n-1$, defined as
      \begin{equation}\label{equ:303}
      \lim_{a_{k}\!\rightarrow a_{k+1}}\!\!\mathbf{P}_{\mathcal {A},\omega,\mathcal{B}}(t)\!\! = \!\! \widetilde{\mathbf{P}}_{\widetilde{\mathcal {A}},\widetilde{\omega},\widetilde{\mathcal{B}}}(t)\!\!=\!\!\frac{\sum_{i=0}^{k-1}\omega_{a_{i}}\mathbf{b}_{a_{i}}\beta_{a_{i}}(t)\!+\!\widetilde{\omega}_{a_{k+1}}\widetilde{\mathbf{b}}_{a_{k+1}}\beta_{a_{k+1}}(t)\!+\!\sum_{i=k+2}^{n}\omega_{a_{i}}\mathbf{b}_{a_{i}}\beta_{a_{i}}(t)}{\sum_{i=0}^{k-1}\omega_{a_{i}}\beta_{a_{i}}(t)\!+\!\widetilde{\omega}_{a_{k+1}}\beta_{a_{k+1}}(t)\!+\!\sum_{i=k+2}^{n}\omega_{a_{i}}\beta_{a_{i}}(t)},
      \end{equation}
      where $\widetilde{\omega}_{a_{k+1}}\!\!\!=\!\omega_{a_{k}}\!+\omega_{a_{k+1}}$, $\widetilde{\mathbf{b}}_{a_{k+1}}\!\!=\!\!\frac{\omega_{a_{k}}}{\omega_{a_{k}}\!+\!\omega_{a_{k+1}}}\mathbf{b}_{a_{k}}+\frac{\omega_{a_{k+1}}}{\omega_{a_{k}}+\omega_{a_{k+1}}}\mathbf{b}_{a_{k+1}}$, $\widetilde{\mathcal {A}}\!\!=\!\!\{a_{0},\!\cdots\!,a_{k-1},a_{k+1},\!\cdots\!,a_{n}\}$, $\widetilde{\mathcal{B}}\!\!=\!\!\{\mathbf{b}_{a_{0}},\cdots,\mathbf{b}_{a_{k-1}},\widetilde{\mathbf{b}}_{a_{k+1}},\mathbf{b}_{a_{k+2}},\cdots,\mathbf{b}_{a_{n}}\}$ and $\widetilde{\omega}\!\!=\!\!\{\omega_{a_{0}},\cdots,\omega_{a_{k-1}},\widetilde{\omega}_{a_{k+1}},\omega_{a_{k+2}},\cdots,\omega_{a_{n}}\}$.
  \end{theorem}

  \begin{proof}
      When $a_{k}~(0\leq k<n)$ tends to $a_{k+1}$, we have
      \begin{equation*}
        \lim_{a_{k}\rightarrow a_{k+1}}\!\!\beta_{a_{k}}(t)\!=\!\!\lim_{a_{k}\rightarrow a_{k+1}}(t-a_{0})^{a_{k}-a_{0}}(a_{n}-t)^{a_{n}-a_{k}}=(t-a_{0})^{a_{k+1}-a_{0}}(a_{n}-t)^{a_{n}-a_{k+1}}\!\!=\!\beta_{a_{k+1}}(t)
      \end{equation*}
      Thus,
      \begin{eqnarray*}
         \lim_{a_{k}\rightarrow a_{k+1}}\sum_{i=0}^{n}\mathbf{b}_{a_{i}}T_{a_{i}}(t)&\!\!=\!\!& \widetilde{\mathbf{P}}_{\widetilde{\mathcal {A}},\widetilde{\omega},\widetilde{\mathcal{B}}}(t)\\
         &\!\!=\!\!& \lim_{a_{k}\!\rightarrow a_{k\!+\!1}}\sum_{i\!=\!0}^{n}\frac{\mathbf{b}_{a_{i}}\omega_{a_{i}}\beta_{a_{i}}(t)}{\sum^{n}_{i\!=\!0}\omega_{a_{i}}\beta_{a_{i}}(t)}\\
         &\!\!=\!\!& \frac{\sum_{i\neq k,k+1}\!\omega_{a_{i}}\!\mathbf{b}_{a_{i}}\beta_{a_{i}}\!(t)\!+\!\omega_{a_{k}}\!\mathbf{b}_{a_{k}}\beta_{a_{k\!+\!1}}\!(t)\!+\!\omega_{a_{k\!+\!1}}\!\mathbf{b}_{a_{k\!+\!1}}\beta_{a_{k\!+\!1}}\!(t)}{\sum_{i\neq k,k+1}\!\omega_{a_{i}}\beta_{a_{i}}\!(t)\!+\!\omega_{a_{k}}\beta_{a_{k\!+\!1}}\!(t)\!+\!\omega_{a_{k\!+\!1}}\beta_{a_{k\!+\!1}}\!(t)}\\
         &\!\!=\!\!& \frac{\sum_{i\neq k,k+1}\!\omega_{a_{i}}\!\mathbf{b}_{a_{i}}\beta_{a_{i}}\!(t)\!+\!(\omega_{a_{k}}\mathbf{b}_{a_{k}}\!+\!\omega_{a_{k+1}}\!\mathbf{b}_{a_{k+1}})\beta_{a_{k+1}}\!(t)\!}{\sum_{i\neq k,k+1}\!\omega_{a_{i}}\beta_{a_{i}}\!(t)\!+\!(\omega_{a_{k}}\!+\!\omega_{a_{k+1}})\beta_{a_{k+1}}\!(t)\!}.
      \end{eqnarray*}

      Let $\widetilde{\omega}_{a_{k+1}}=\omega_{a_{k}}+\omega_{a_{k+1}}, \widetilde{\mathbf{b}}_{a_{k+1}}=\frac{\omega_{a_{k}}}{\omega_{a_{k}}+\omega_{a_{k+1}}}\mathbf{b}_{a_{k}}+\frac{\omega_{a_{k+1}}}{\omega_{a_{k}}+\omega_{a_{k+1}}}\mathbf{b}_{a_{k+1}}$,
       we can obtain
      \begin{equation*}
        \lim_{a_{k}\rightarrow a_{k+1}}\mathbf{P}_{\mathcal {A},\omega,\mathcal{B}}(t)\! = \! \widetilde{\mathbf{P}}_{\widetilde{\mathcal {A}},\widetilde{\omega},\widetilde{\mathcal{B}}}(t) \!=\!\frac{\sum_{i\neq k,k+1}\!\omega_{a_{i}}\!\mathbf{b}_{a_{i}}\beta_{a_{i}}\!(t)\!+\!\widetilde{\omega}_{a_{k\!+\!1}}\widetilde{\mathbf{b}}_{a_{k\!+\!1}}\beta_{a_{k\!+\!1}}(t)}{\sum_{i\neq k,k+1}\!\omega_{a_{i}}\beta_{a_{i}}\!(t)\!+\!\widetilde{\omega}_{a_{k+1}}\beta_{a_{k+1}}(t)},
      \end{equation*}
  where $\widetilde{\mathcal {A}}=\{a_{0},\cdots,a_{k-1},a_{k+1},\cdots,a_{n}\},\widetilde{\mathcal{B}}=\{\mathbf{b}_{a_{0}},\cdots,\mathbf{b}_{a_{k-1}},\widetilde{\mathbf{b}}_{a_{k+1}},\mathbf{b}_{a_{k+2}},\cdots,\mathbf{b}_{a_{n}}\}$ and $\widetilde{\omega}=\{\omega_{a_{0}},\cdots,\omega_{a_{k-1}},\widetilde{\omega}_{a_{k+1}},\omega_{a_{k+2}},\cdots,\omega_{a_{n}}\}$. This leads to prove the result.
  \end{proof}

  \begin{example}\label{exa:302}
  Consider the curve $\mathbf{P}_{\mathcal {A},\omega,\mathcal{B}}(t)$ defined as in Example \ref{exa:301}. Let knots $\mathcal {A}=\{a_{0}\!=\!0,a_{1}\!=\!\frac{\sqrt{2}}{4},a_{2}\!=\!\frac{1}{2},a_{3}\!=\!\frac{\sqrt{2}}{2},a_{4}\!=\!1\}$, weights $\omega_{a_{0}}\!=\!1,\omega_{a_{1}}\!=\!10,\omega_{a_{2}}\!=\!20,\omega_{a_{3}}\!=\!6,\omega_{a_{4}}\!=\!5$, control points $\mathbf{b}_{a_{0}}\!=\!(0,0),\mathbf{b}_{a_{1}}\!=\!(0.4,1.3),\mathbf{b}_{a_{2}}\!=\!(2,2),\mathbf{b}_{a_{3}}\!=\!(3.7,1.5),\mathbf{b}_{a_{4}}\!=\!(4,0)$ and $c_{a_{i}}\!=\!1~(i\!=\!0,\cdots,4)$. If $a_{1}$ approaches $a_{2}$, then the changes of the GT-B\'ezier curve are shown in Fig.~\ref{fig:302}. We can see that the limit curve $ \lim_{a_{1}\rightarrow a_{2}}\mathbf{P}_{\mathcal {A},\omega,\mathcal{B}}(t)$ coincides with the target curve $\widetilde{\mathbf{P}}_{\widetilde{\mathcal {A}},\widetilde{\omega},\widetilde{\mathcal{B}}}(t)$, which verifies the Theorem \ref{the:301}.

  \begin{figure}[h!]
  \centering
  \subfigure[Initial curve]{
    \includegraphics[width=6cm]{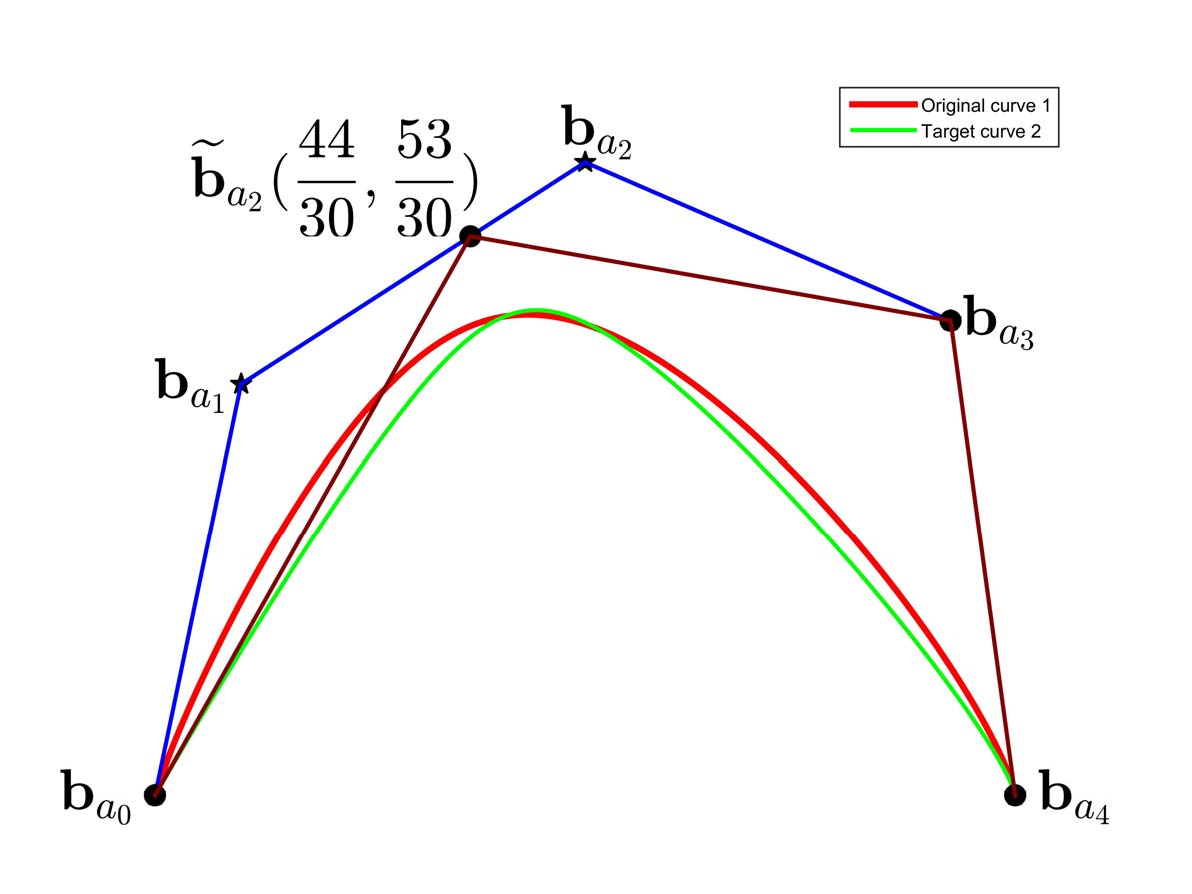}}
  \hspace{1cm}
  \subfigure[Limit curve]{
    \includegraphics[width=6cm]{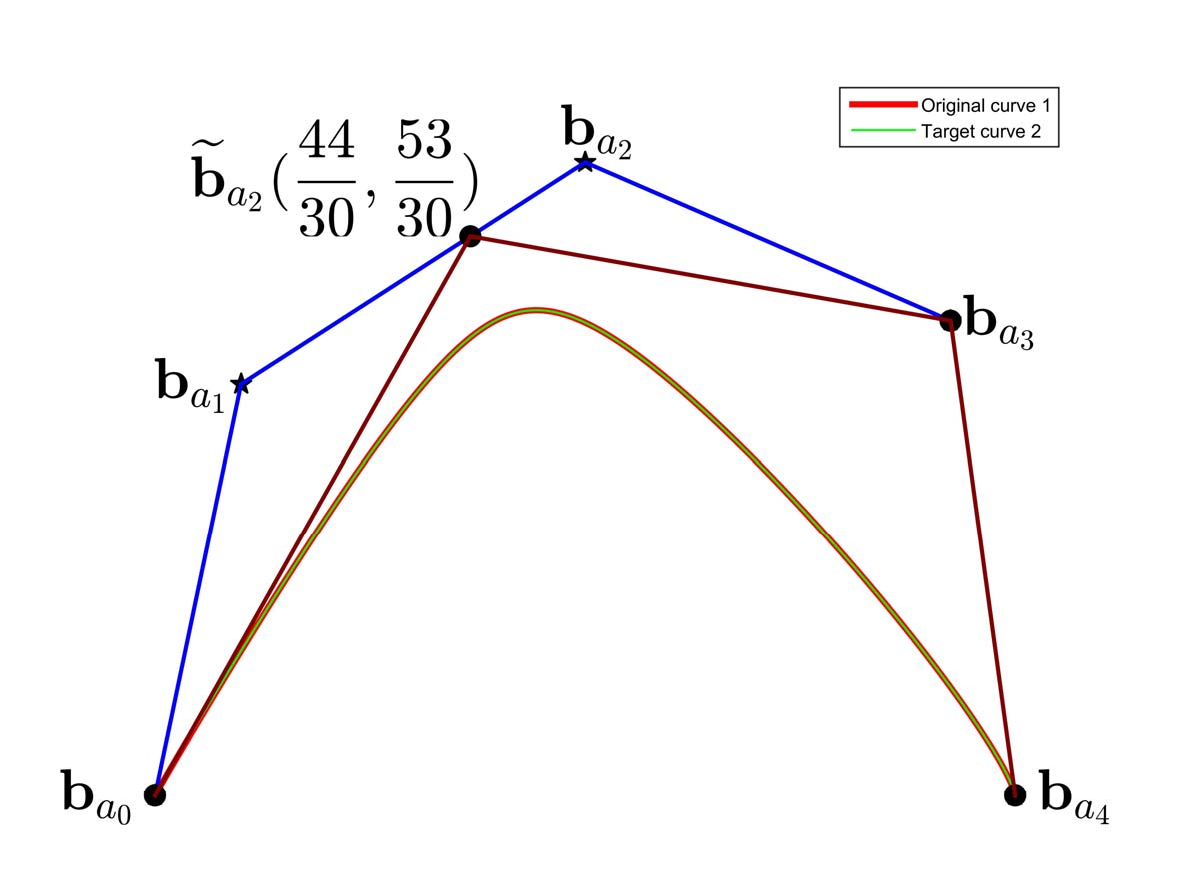}}
  \caption{Limits of the quadratic GT-B\'ezier curve of single knot}
  \label{fig:302}
  \end{figure}
\end{example}

Theorem \ref{the:301} indicates that the GT-B\'ezier curve of degree $n$ degenerates into the GT-B\'ezier curve of degree $n-1$ with knots $\widetilde{\mathcal {A}}\!=\!\{a_{0},\!\cdots\!,a_{k-1},a_{k+1},\!\cdots\!,a_{n}\}$, control points $\widetilde{\mathcal{B}}\!=\!\{\mathbf{b}_{a_{0}},\!\cdots\!,\mathbf{b}_{a_{k-1}},\widetilde{\mathbf{b}}_{a_{k+1}},\mathbf{b}_{a_{k+2}},\!\cdots\!,\mathbf{b}_{a_{n}}\}$ and weights
  $\widetilde{\omega}\!=\!\{\omega_{a_{0}},\!\cdots\!,\omega_{a_{k-1}},\widetilde{\omega}_{a_{k+1}},\omega_{a_{k+2}},\!\cdots\!,\omega_{a_{n}}\}$ when $a_{k}=a_{k+1}$. The following corollary generalizes Theorem \ref{the:301}, and gives the limit of GT-B\'ezier curve with multiple knots. The proof of the corollary is similar to Theorem \ref{the:301} and will be omitted here.

  \begin{corollary}\label{cor:301}
  Suppose $c_{a_{i}}=1~(i=0,\cdots,n)$. When knots $a_{q},a_{q+1},\cdots,a_{q+k-2}(0\leq q<n,1<k\leq n+1-q)$ approaches to the knot $a_{q+k-1}$, the limit of GT-B\'ezier curve $\mathbf{P}_{\mathcal {A},\omega,\mathcal{B}}(t)$ of degree $n$ defined in (\ref{equ:301}) is exactly the GT-B\'ezier curve of degree $n\!-\!k\!+\!1$ as
  \begin{eqnarray*}
    \lim_{a_{q},\!\cdots\!,a_{q\!+\!k\!-\!2}\!\rightarrow\! a_{q\!+\!k\!-\!1}}\!\!\!\!\!\!\!\mathbf{P}_{\mathcal {A},\omega,\mathcal{B}}(t)&\!\! = \!\!& \widetilde{\mathbf{P}}_{\widetilde{\mathcal {A}},\widetilde{\omega},\widetilde{\mathcal{B}}}(t)\\
     &\!\!=\!\!&  \frac{\sum_{i=0}^{q\!-\!1}\omega_{a_{i}}\!\mathbf{b}_{a_{i}}\beta_{a_{i}}\!(t)\!+\!\widetilde{\omega}_{a_{q\!+\!k\!-\!1}}\!\widetilde{\mathbf{b}}_{a_{q\!+\!k\!-\!1}}\beta_{a_{q\!+\!k\!-\!1}}\!(t)\!+\!\sum_{i=q\!+\!k}^{n}\omega_{a_{i}}\mathbf{b}_{a_{i}}\beta_{a_{i}}\!(t)\!}{\sum_{i=0}^{q\!-\!1}\omega_{a_{i}}\beta_{a_{i}}\!(t)\!+\!\widetilde{\omega}_{a_{q\!+\!k\!-\!1}}\beta_{a_{q\!+\!k\!-\!1}}\!(t)\!\!+\!\!\sum_{i=q\!+\!k}^{n}\omega_{a_{i}}\beta_{a_{i}}\!(t)\!},
  \end{eqnarray*}
  where $\widetilde{\omega}_{a_{q+k-1}}\!\!=\!\!\omega_{a_{q}}+\omega_{a_{q+1}}+\cdots+\omega_{a_{q+k-1}}, \widetilde{\mathbf{b}}_{a_{q+k-1}}\!\!=\!\!\frac{\omega_{a_{q}}}{\widetilde{\omega}_{a_{q+k-1}}}\mathbf{b}_{a_{q}}\!+\cdots+\frac{\omega_{a_{q+k-1}}}{\widetilde{\omega}_{a_{q+k-1}}}\mathbf{b}_{a_{q+k-1}}$, $\widetilde{\mathcal {A}}=\{a_{0},\cdots,a_{q-1},a_{q+k-1},\cdots,a_{n}\}$, $\widetilde{\mathcal{B}}\!\!=\!\!\{\mathbf{b}_{a_{0}},\cdots,\mathbf{b}_{a_{q-1}},\widetilde{\mathbf{b}}_{a_{q+k-1}},\mathbf{b}_{a_{q+k}},\cdots,\mathbf{b}_{a_{n}}\}$ and $\widetilde{\omega}\!\!=\!\!\{\omega_{a_{0}},\cdots,\omega_{a_{q-1}},\widetilde{\omega}_{a_{q+k-1}},\omega_{a_{q+k}},\cdots,\omega_{a_{n}}\}$.
  \end{corollary}

\begin{example}\label{exa:303}
   Consider the curve $\mathbf{P}_{\mathcal {A},\omega,\mathcal{B}}(t)$ defined as in Example \ref{exa:301}. If $a_{1}\rightarrow a_{2}$ and $a_{3}\rightarrow a_{2}$, then the changes of the GT-B\'ezier curve are shown in Fig.~\ref{fig:303}. The limit curve is constructed by knots $\widetilde{\mathcal {A}}=\{a_{0}=0,a_{2}=\frac{1}{2},a_{4}=1\}$, control points  $\widetilde{\mathcal{B}}=\{\mathbf{b}_{a_{0}}=(0,0),\widetilde{\mathbf{b}}_{a_{2}}=(\frac{66.2}{36},\frac{62}{36}),\mathbf{b}_{a_{4}}=(4,0)\}$ and weights
  $\widetilde{\omega}=\{\omega_{a_{0}}=1,\widetilde{\omega}_{a_{2}}=36,\omega_{a_{4}}=5\}$. We can see that the limit curve coincides with the target curve together, which verifies the result of Corollary \ref{cor:301}.

  \begin{figure}[h!]
  \centering
  \subfigure[Initial curve]{
    \includegraphics[width=6cm]{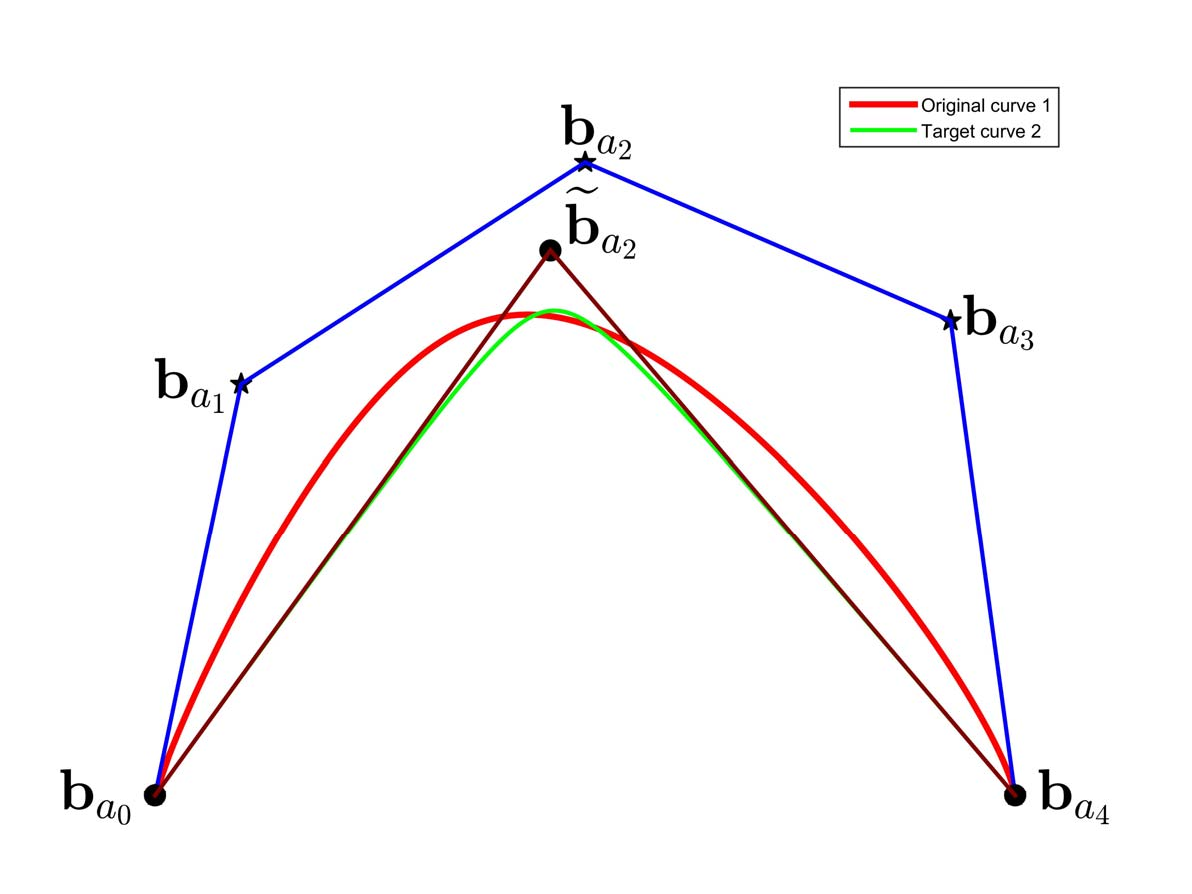}}
   \hspace{1cm}
  \subfigure[Limit curve]{
    \includegraphics[width=6cm]{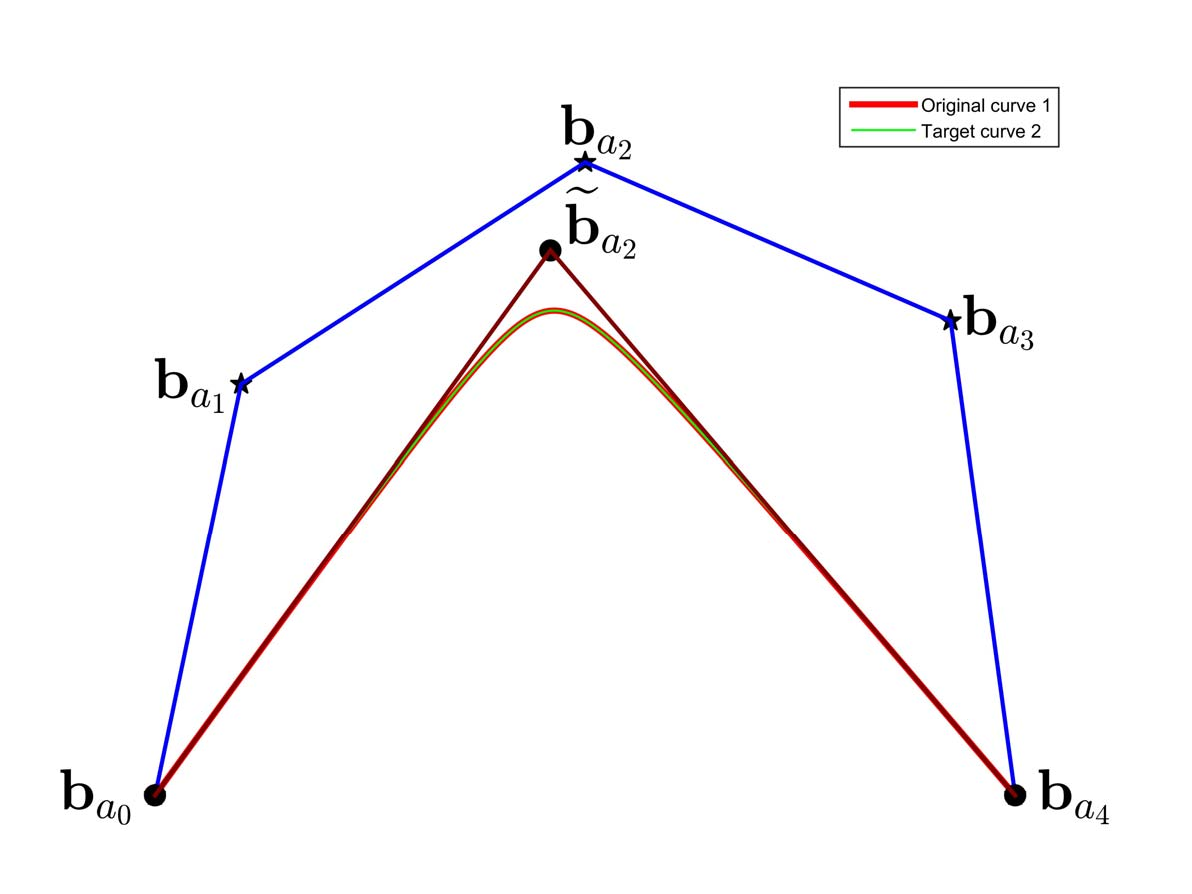}}
  \caption{Limits of the quadratic GT-B\'ezier curve with multiple knots }
  \label{fig:303}
  \end{figure}
\end{example}

\item[(g)] {\bf{Toric degeneration property}}. For each $t\in\Delta_{\mathcal {A}}$, we have the limiting property of GT-B\'ezier curve while a single weight of curve tends to infinity, that is
  \begin{equation*}
  \lim_{\omega_{a_{i}}\rightarrow +\infty}\mathbf{P}_{\mathcal {A},\omega,\mathcal{B}}(t)=
    \left\{
             \begin{array}{ll}
             \mathbf{b}_{a_{0}}&  t=a_{0},\\
             \mathbf{b}_{a_{i}}&  t\in(a_{0},a_{n}),\\
             \mathbf{b}_{a_{n}}&  t=a_{n}.
             \end{array}
        \right.
  \end{equation*}
  And this property can be derived from weight property of rational B\'ezier directly. Fig.~\ref{fig:304} shows the limit curve of GT-B\'ezier curve defined in Example \ref{exa:301} with $\omega_{a_{1}}\rightarrow +\infty$.
\begin{figure}[h!]
  \begin{center}
  \includegraphics[width=7cm]{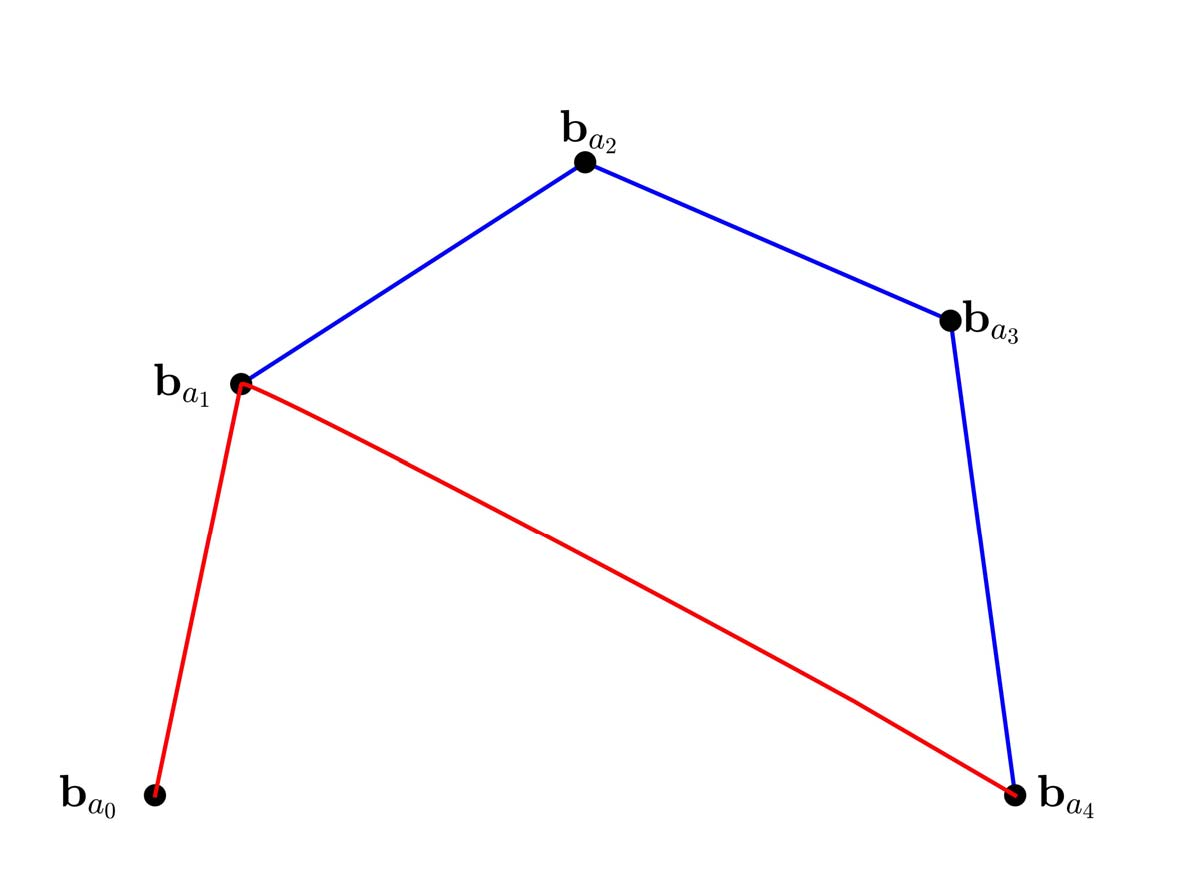}
  \caption{Limit of GT-B\'ezier curve with $\omega_{a_{1}}\rightarrow +\infty$.}
  \label{fig:304}
  \end{center}
\end{figure}

Next, we consider the property of GT-B\'ezier curve if all the weights tend to infinity.

Let $\lambda :\mathcal{A}\rightarrow \mathbb{R}$ be a lifting function to lift the points $a_{i}$ of $\mathcal{A}$ to $(a_{i},\lambda(a_{i}))\in \mathbb{R}^{2}$. We denote $P_{\lambda}=conv\{(a_{i},\lambda(a_{i}))\mid a_{i}\in \mathcal{A}\}$ the convex hull of the lifted points. Each edge of the convex hull $P_{\lambda}$ has a normal vector pointing to the outer side. We call it the upper edges of $P_{\lambda}$ if the last coordinate of the normal vector is positive. If we project these upper edges back vertically into $\mathbb{R}$, they can cover $\Delta_{\mathcal{A}}$ and form a regular subdivision $\Gamma_{\lambda}$ of $\Delta_{\mathcal{A}}$ induced by $\lambda$ \cite{ref13}.

  We group together the points of $\mathcal{A}$ that are in the same subset of the $\Gamma_{\lambda}$ and on the same upper edge of the $P_{\lambda}$. Then we get a decomposition of $\mathcal{A}$, which is called regular decomposition $\mathcal{S}_{\lambda}$ of $\mathcal{A}$ induced by $\lambda$. For each subset $\mathcal{F}$ of $\mathcal{S}_{\lambda}$, we can use the weights $\omega|_{\mathcal{F}}=\{\omega_{a_{i}}\mid a_{i}\in\mathcal{F}\}$ and the control points $\mathcal{B}|_{\mathcal{F}}=\{\mathbf{b}_{a_{i}}\mid a_{i}\in\mathcal{F}\}$ to define a new GT-B\'ezier curve $\mathbf{P}_{\mathcal{F},\omega|_{\mathcal{F}},\mathcal{B}|_{\mathcal{F}}}$ on $\Delta_{\mathcal{F}}=conv\{a_{i}\in\mathcal{F}\mid a_{i}\in\mathcal{A}\}$ by Definition \ref{def:301}. The union of these curves
  \begin{equation*}
    \mathbf{P}_{\mathcal {A},\omega,\mathcal{B}}(\mathcal{S}_{\lambda})=\bigcup_{\mathcal{F}\in \mathcal{S}_{\lambda}}\mathbf{P}_{\mathcal{F},\omega|_{\mathcal{F}},\mathcal{B}|_{\mathcal{F}}}
  \end{equation*}
   is called the regular control curve of $\mathbf{P}_{\mathcal {A},\omega,\mathcal{B}}$ induced by regular decomposition $\mathcal{S}_{\lambda}$.

  We can use lifting function $\lambda$ to get a set of weights with a parameter $x$, $\omega_{\lambda}(x):=\{x^{\lambda(a_{i})}\omega_{a_{i}}\mid a_{i}\in\mathcal{A}\}$. These weights are used to define the map 
  \begin{equation}\label{equ:304}
    \mathbf{P}_{\mathcal {A},\omega_{\lambda}(x),\mathcal{B}}(t)=\frac{\sum_{i=0}^{n}x^{\lambda(a_{i})}\omega_{a_{i}}\mathbf{b}_{a_{i}}\beta_{a_{i}}(t)}{\sum^{n}_{i=0}x^{{\lambda(a_{i})}}\omega_{a_{i}}\beta_{a_{i}}(t)},~~~~t\in\Delta_{\mathcal {A}}.
  \end{equation}
  The image of $\Delta_{\mathcal{A}}$ under this map is a GT-B\'ezier curve with a parameter $x$, denoted as $P_{\mathcal {A},\omega_{\lambda}(x),\mathcal{B}}$.
  We have the following result.

  \begin{theorem}\label{the:302}
  The limit of the GT-B\'ezier curve $\mathbf{P}_{\mathcal {A},\omega_{\lambda}(x),\mathcal{B}}$ as $x\rightarrow \infty$ is the regular control curve induced by regular decomposition $\mathcal{S}_{\lambda}$, that is
  \begin{equation*}
    \lim_{x\rightarrow \infty}\mathbf{P}_{\mathcal {A},\omega_{\lambda}(x),\mathcal{B}}=\mathbf{P}_{\mathcal {A},\omega,\mathcal{B}}(\mathcal{S}_{\lambda}).
  \end{equation*}
  \end{theorem}

  \begin{proof}
  According to the theory of real irrational toric varieties in~\cite{ref16}, the GT-B\'ezier curve $\mathbf{P}_{\mathcal {A},\omega,\mathcal{B}}$ is obtained by the projection of the high-dimensional real projective toric variety formed by $\mathcal {A}$. Then $\mathbf{P}_{\mathcal {A},\omega,\mathcal{B}}$ is projection after the composition of a sequence of mappings
   \begin{equation*}
     \mathcal {A}\stackrel{\{\beta_{a_{i}}\mid a_{i}\in\mathcal {A}\}}{\longrightarrow}X_{\mathcal {A}}\stackrel{\omega}{\longrightarrow}X_{\mathcal {A},\omega}\stackrel{\mathcal{B}}{\longrightarrow}\mathbf{P}_{\mathcal {A},\omega,\mathcal{B}}.
   \end{equation*}

   For $\mathcal {A}$ and weights $\omega_{\lambda}(x)$ with parameter $x$, we can get a family of translated toric varieties $X_{\mathcal {A},\omega_{\lambda}(x)}$
   \begin{equation*}
     \mathcal {A}\stackrel{\{\beta_{a_{i}}\mid a_{i}\in\mathcal {A}\}}{\longrightarrow}X_{\mathcal {A}}\stackrel{\omega_{\lambda}(x)}{\longrightarrow}X_{\mathcal {A},\omega_{\lambda}(x)}.
   \end{equation*}

  When $x\rightarrow \infty$, $X_{\mathcal {A},\omega_{\lambda}(x)}$ limits to a union of irrational toric varieties in the Hausdorff distance, which are defined by the all of subset of $\mathcal{S}_{\lambda}$. That is
  \begin{equation*}
    \lim_{x\rightarrow \infty}X_{\mathcal {A},\omega_{\lambda}(x)}=\bigcup_{\mathcal{F}\in \mathcal{S}_{\lambda}}X_{\mathcal{F},\omega|_{\mathcal{F}}}.
  \end{equation*}

  Then add control points $\mathcal{B}$, we have
  \begin{equation*}
    \bigcup_{\mathcal{F}\in \mathcal{S}_{\lambda}}X_{\mathcal{F},\omega|_{\mathcal{F}}}\stackrel{\mathcal{B}}{\longrightarrow}\bigcup_{\mathcal{F}\in \mathcal{S}_{\lambda}}\mathbf{P}_{\mathcal{F},\omega|_{\mathcal{F}},\mathcal{B}|_{\mathcal{F}}}=\mathbf{P}_{\mathcal {A},\omega,\mathcal{B}}(\mathcal{S}_{\lambda})
  \end{equation*}
  So the result holds.
  \end{proof}

  Theorem~\ref{the:302} shows that regular control curves are exactly the limits of the GT-B\'ezier curve when all the weights tend to infinity. Obviously the control polygon is the regular control curve of GT-B\'ezier curve. This property is also called toric degeneration of GT-B\'ezier curves.

\begin{example}\label{exa:304}
  Let $\mathcal {A}=\{0,\frac{\sqrt{2}}{4},\frac{1}{2},\frac{\sqrt{2}}{2},1\}$, and the lifted values of $\mathcal {A}$ by a lifting function $\lambda$ be $(2,1,5,9-4\sqrt{2},1)$. This induces a regular decomposition of $\mathcal {A}$ as
  \begin{equation*}
    \left\{\{0,\frac{1}{2}\}~,~\{\frac{1}{2},\frac{\sqrt{2}}{2},1\}\right\}.
  \end{equation*}
  The lifted point $\frac{\sqrt{2}}{4}$ doesn't lie on any upper edge of the lifting polygon $P_{\lambda}$,then it doesn't lie on any subset of the decomposition.

  Fig.~\ref{fig305:subfig:a} shows $\mathcal {A}$, the lifted values of $\mathcal {A}$ by $\lambda$, and the corresponding regular decomposition. Fig.~\ref{fig305:subfig:b},\ref{fig305:subfig:c},\ref{fig305:subfig:d} show the toric degeneration of this GT-B\'ezier curve for $x=1.3$, $x=2$, and $x=3$. The GT-B\'ezier curve approaches its regular control curve as the parameter $x$ becomes larger.

  If $\lambda^{'}$ takes the values of $\mathcal {A}$ as $\{0,2.5,3,2.5,0\}$, then this induces a regular decomposition of $\mathcal {A}$ as
  \begin{equation*}
    \left\{\{0,\frac{\sqrt{2}}{4}\}~,~\{\frac{\sqrt{2}}{4},\frac{1}{2}\}~,~\{\frac{1}{2},\frac{\sqrt{2}}{2}\}~,~\{\frac{\sqrt{2}}{2},1\}\right\}.
  \end{equation*}
  The corresponding regular decomposition is shown in Fig.~\ref{fig306:subfig:a} and the regular control curve is exactly the control polygon of the curve.

 Moreover $\lambda^{''}$ takes the values of $\mathcal {A}$ as $\{1,3,1,0,1\}$, then the regular decomposition of $\mathcal {A}$ is $\left\{\!\{\!0,\!\!\frac{\sqrt{2}}{4}\},\!\{\!\frac{\sqrt{2}}{4}\!,\!1\}\!\right\}$ (see Fig.~\ref{fig306:subfig:b}) and the regular control curve is as shown in  Fig.~\ref{fig:304}.
\begin{figure}[h!]
\centering
  \subfigure[$\mathcal{A}$ and $\lambda$]{
    \label{fig305:subfig:a} 
    \includegraphics[width=5.1cm]{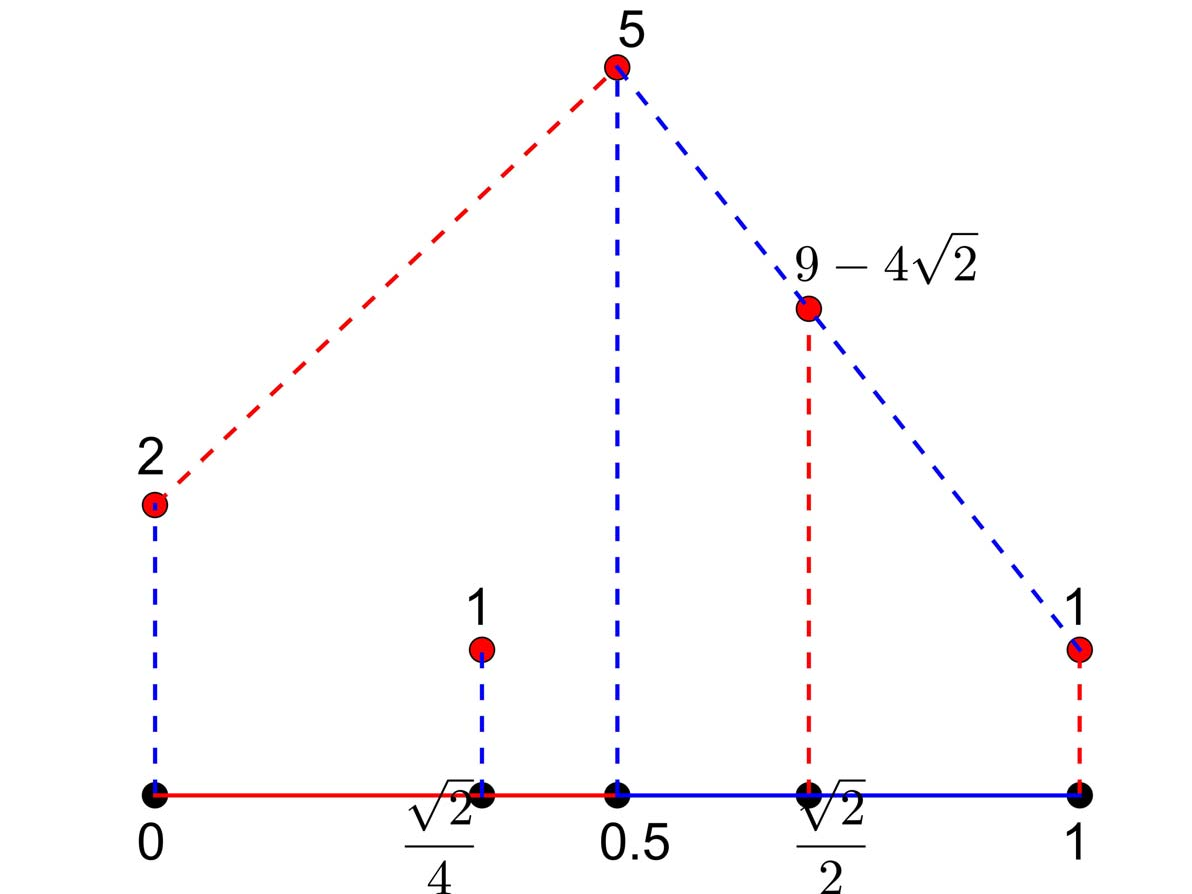}}
  \hspace{0.2cm}
  \subfigure[$x=1.3$]{
    \label{fig305:subfig:b} 
    \includegraphics[width=5.1cm]{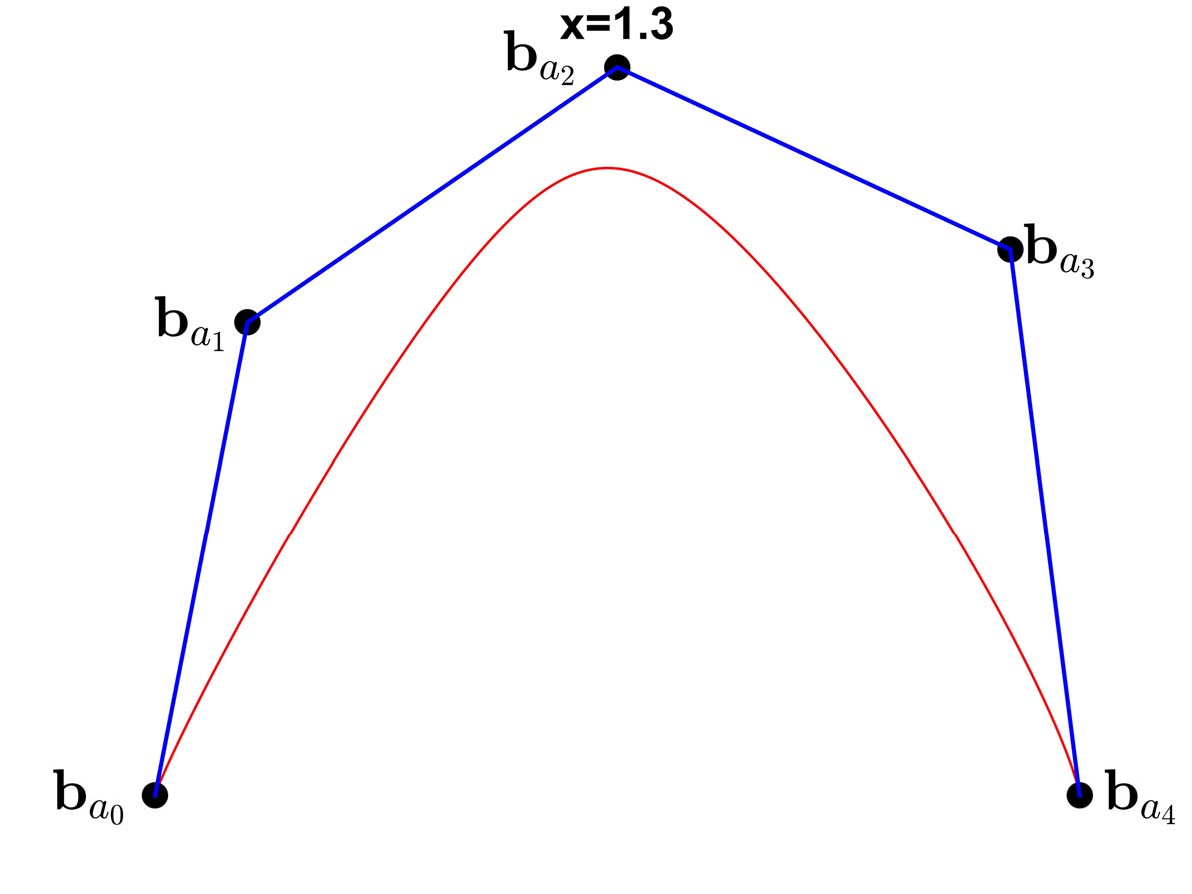}}
  \hspace{0.2cm}
  \subfigure[$x=2$]{
    \label{fig305:subfig:c} 
    \includegraphics[width=5.1cm]{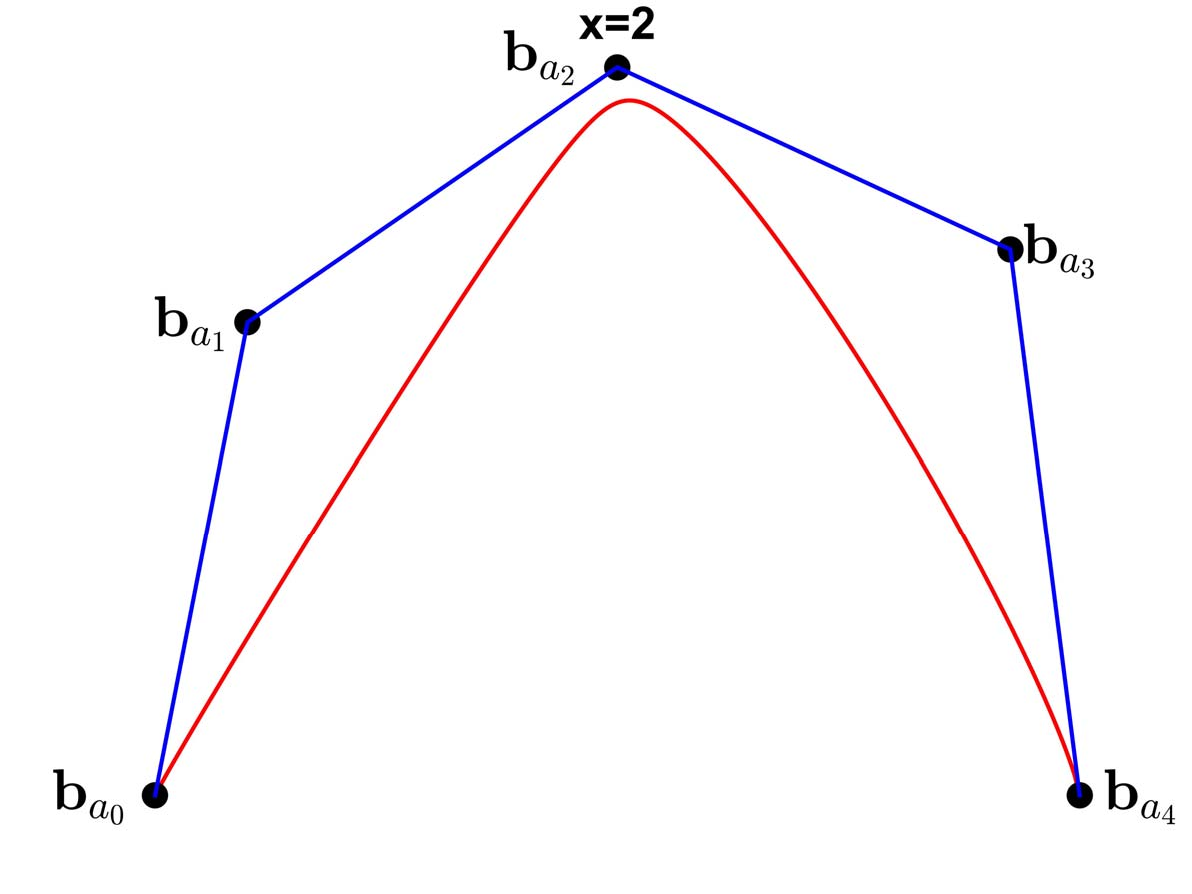}}\\
\centering
  \subfigure[$x=3$]{
    \label{fig305:subfig:d} 
    \includegraphics[width=5.1cm]{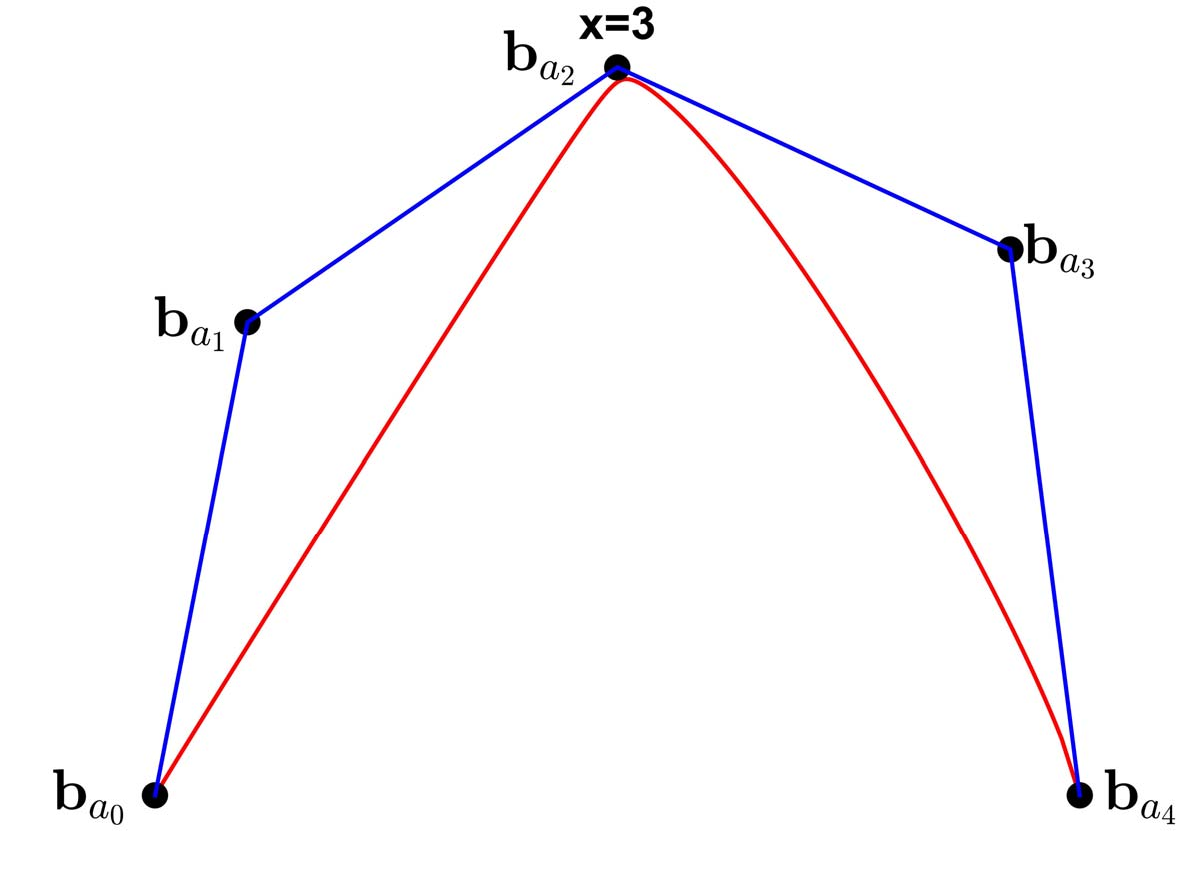}}
  \hspace{0.5cm}
  \subfigure[Regular control curve]{
    \label{fig305:subfig:e} 
    \includegraphics[width=5.1cm]{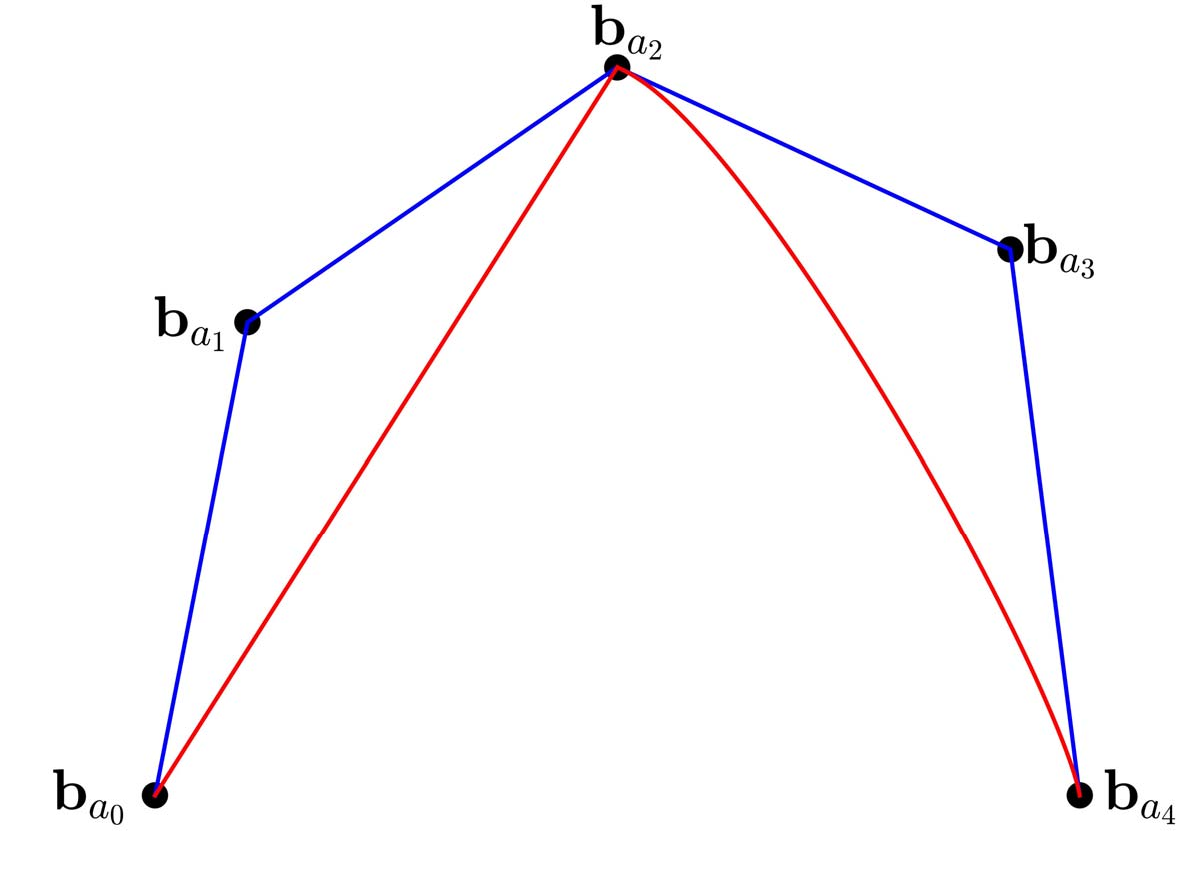}}
  \caption{Toric degeneration of GT-B\'ezier curve}
  \label{fig:305}
\end{figure}
\begin{figure}[h!]
  \centering
  \subfigure[$\mathcal{A}$ and $\lambda^{'}$]{
    \label{fig306:subfig:a} 
    \includegraphics[width=5.1cm]{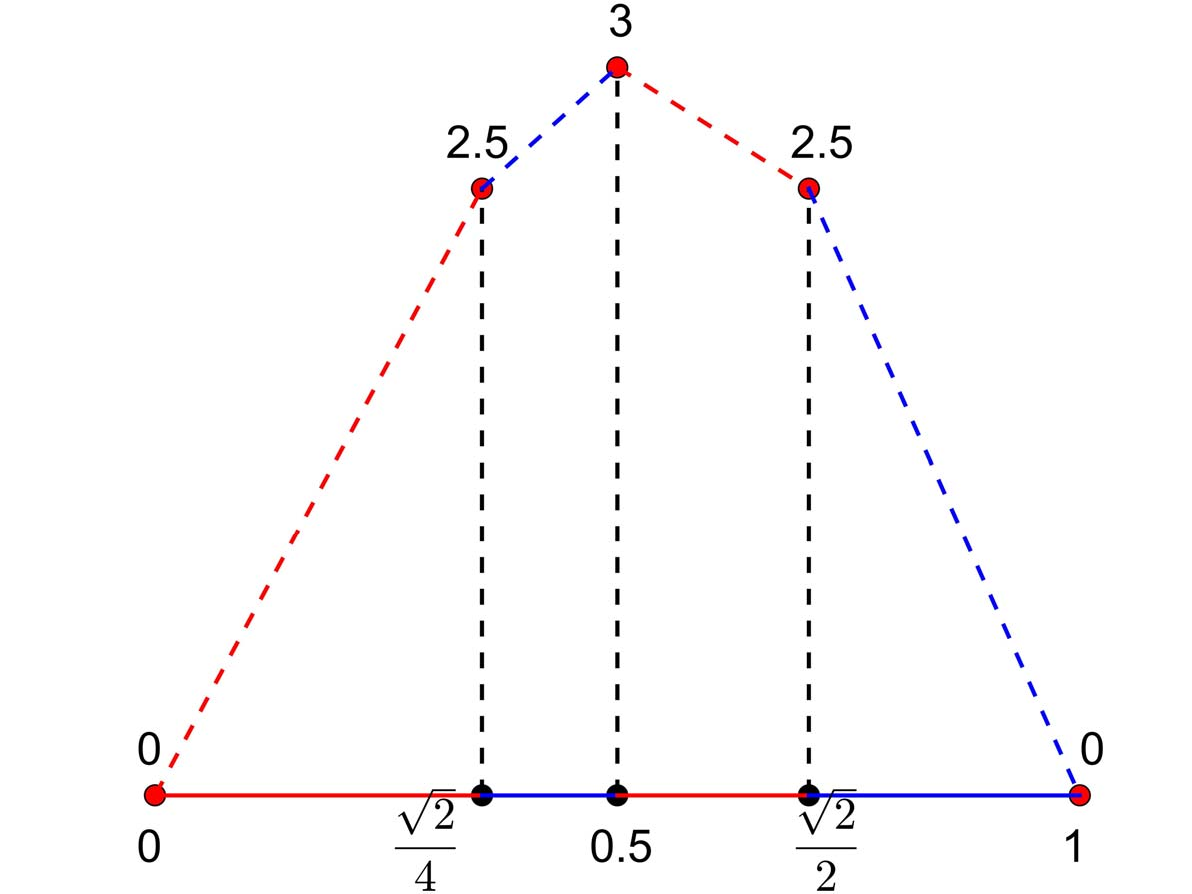}}
  \hspace{0.5cm}
  \subfigure[$\mathcal{A}$ and $\lambda^{''}$]{
    \label{fig306:subfig:b} 
    \includegraphics[width=5.1cm]{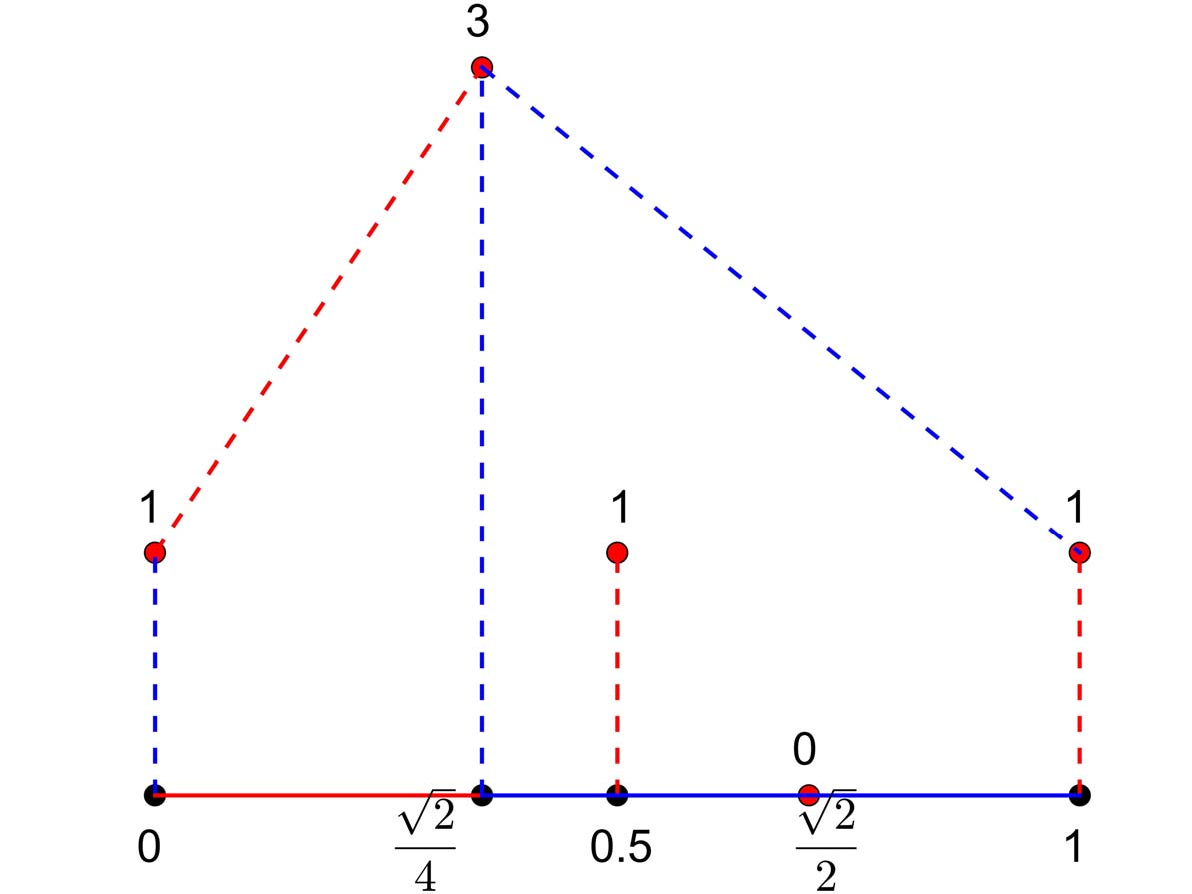}}
  \caption{Regular decompositions of $\mathcal{A}$}
  \label{fig:306}
\end{figure}
\end{example}

\item[(h)] {\bf{Variation diminishing (VD) property}}.
Let $d_{i}=a_{i}-a_{0}~(i=1,\cdots,n)$ for $\mathcal {A}=\{a_{0},a_{1},\cdots,a_{n}\}\subset \mathbb{R}$ with $a_{0}\leq a_{1}\leq \cdots\leq a_{n-1}\leq a_{n}$. If $d_{i}~(i=1,\cdots,n)$ are rational numbers, then $d_{i}$ can be expressed as $d_{i}=\frac{p_{i}}{q_{i}}(p_{i},q_{i}\in \mathbb{N})$. Let $q$ be the least common multiple of $q_{1},q_{2},\cdots,q_{n}$, namely, $q=[q_{1},q_{2},\cdots,q_{n}]$, then $q d_{i}\in \mathbb{N}\setminus \{0\}$. At this point, we have the following theorem.

\begin{theorem}\label{the:303}
If $d_{i}=a_{i}-a_{0}\in \mathbb{Q}~(i=1,2,\cdots,n)$, then the planar GT-B\'ezier curve $\mathbf{P}_{\mathcal {A},\omega,\mathcal{B}}(t)$ is variation diminishing, which means that the number of intersections of any straight line with the GT-B\'ezier curve $\mathbf{P}_{\mathcal {A},\omega,\mathcal{B}}(t)$ is no more than the number of intersections of the line with its control polygon.
\end{theorem}

\begin{proof}
In order to prove this theorem, we need to use the Cartesian notation rule, which presents the upper bound of the number of the positive roots of the polynomial. For any polynomial $f(t)=m_{0}+m_{1}t+\cdots+m_{n}t^{n}$, if we write $Z_{t>0}[f(t)]$ to denote the number of positive roots of $f(t)$ and denote $V[m_{0},m_{1},\cdots,m_{n}]$ as the number of strict sign changes of polynomial coefficients, then
\begin{equation*}
  Z_{t>0}[m_{0}+m_{1}t+\cdots+m_{n}t^{n}]\leq V[m_{0},m_{1},\cdots,m_{n}].
\end{equation*}
Let $L$ denote any straight line, $C$ denote the planar GT-B\'ezier curve defined by $\mathcal {A}$, and write $I(C,L)$ to denote the number of times $L$ crosses $C$. Establish the Cartesian coordinate system with $L$ as the abscissa axis. Because GT-B\'ezier curve is geometric invariant, we can let $(x_{a_{i}},y_{a_{i}})(i=0,1,\cdots,n)$ represent the new coordinates of the control points. Let $P$ denote the control polygon and $I(P,L)$ denote the number of times $L$ crosses $P$. We only need to prove that $I(C,L)\leq I(P,L)$.

We set a parameter transformation as $u=\frac{t-a_{0}}{a_{n}-t}, t\in(a_{0},a_{n})$, so that $u\in(0,+\infty)$. Then by the Cartesian notation rule
\begin{eqnarray*}
  I(C,L) \!\!\!\!&=& \!\!\!\!Z_{a_{0}\leq t \leq a_{n}}\left[\sum_{i=0}^{n}y_{a_{i}}T_{a_{i}}(t)\right] \!=\! Z_{a_{0}\leq t \leq a_{n}}\left[\sum_{i=0}^{n}\frac{y_{a_{i}}\omega_{a_{i}}c_{a_{i}}(t-a_{0})^{t-a_{i}}(a_{n}-t)^{a_{n}-a_{i}}}{\sum_{i=0}^{n}\omega_{a_{i}}c_{a_{i}}(t-a_{0})^{t-a_{i}}(a_{n}-t)^{a_{n}-a_{i}}}\right]\\  &=&\!\!\!\!Z_{0< u < +\infty}\left[\sum_{i=0}^{n}\frac{y_{a_{i}}\omega_{a_{i}}c_{a_{i}}u^{a_{i}-a_{0}}}{\sum_{i=0}^{n}\omega_{a_{i}}c_{a_{i}}u^{a_{i}-a_{0}}}\right] \!=\! Z_{0< u < +\infty}\left[\sum_{i=0}^{n}\frac{y_{a_{i}}\omega_{a_{i}}c_{a_{i}}u^{q d_{i}}}{\sum_{i=0}^{n}\omega_{a_{i}}c_{a_{i}}u^{q d_{i}}}\right] \\
   &=& \!\!\!\!Z_{0< u < +\infty}\left[\sum_{i=0}^{n}\frac{y_{a_{i}}\omega_{a_{i}}c_{a_{i}}u^{q d_{i}}}{\sum_{i=0}^{n}\omega_{a_{i}}c_{a_{i}}u^{q d_{i}}}\right] \!=\! Z_{0< u < +\infty}\left[\sum_{i=0}^{n}y_{a_{i}}u^{q d_{i}}\right] \\
   &\leq&\!\!\!\!V\left[y_{a_{0}},y_{a_{1}},\cdots,y_{a_{n}}\right]\!=\! I(P,L),
\end{eqnarray*}
and this leads to end the proof.
\end{proof}

From Theorem~\ref{the:303}, we have the following property.

\item[(i)] {\bf{Convexity-preserving property}}. Suppose $d_{i}=a_{i}-a_{0}\in \mathbb{Q}~(i=1,\cdots,n)$, then the planar GT-B\'ezier curve is convex if its control polygon is convex.
\end{description}

\section{Generalized Toric-B\'ezier Surfaces}

\begin{definition}\label{def:401}
Let $\mathcal {A}=\{\mathbf{a}_{0},\mathbf{a}_{1},\cdots,\mathbf{a}_{n}\}\subset\mathbb{R}^{2}$ be a finite set of real points. Given positive weights $\omega=\{\omega_{a_{i}}\mid a_{i}\in\mathcal {A}\}$ and control points $\mathcal {B}=\{\mathbf{b}_{\mathbf{a}_{i}}\mid\mathbf{a}_{i}\in\mathcal {A}\}$, the generalized toric-B\'ezier surface (GT-B\'ezier surface for short) is defined as
\begin{equation}\label{equ:401}
   \mathbf{P}_{\mathcal {A},\omega,\mathcal{B}}(u,v)\!\!=\!\!\sum_{i\!=\!0}^{n}\!\mathbf{b}_{\mathbf{a}_{i}}T_{\mathbf{a}_{i}}(u,v) \!\!=\!\!\sum_{i\!=\!0}^{n}\!\mathbf{b}_{\mathbf{a}_{i}}\frac{\omega_{\mathbf{a}_{i}}\beta_{\mathbf{a}_{i}}(u,v)}{\sum^{n}_{i\!=\!0}\!\omega_{\mathbf{a}_{i}}\beta_{\mathbf{a}_{i}}(u,v)},~~~(u,v)\in\Delta_{\mathcal {A}}.
\end{equation}
\end{definition}

\begin{example}\label{exa:401}
Let $\mathcal {A}\!=\!\{(0,2),(1,2),(0,1),(1,1),(2,1),(0,0),(1,0),(2,0)\}$ be the integer points in the pentagon as shown in Fig.~\ref{fig1:subfig:a}, and set control points $\mathcal {B}\!=\!\{(0,2,0),(1,2,4),(0,\frac{6}{5},2),(\frac{8}{7},\frac{8}{7},5),$\\
$(2,1,2),(0,0,0),(\frac{6}{5},0,2),(2,0,0)\}$ and weights $\omega=\{2,2,5,7,2,3,5,2\}$. Suppose $c_{\mathbf{a}_{i}}=1~(i=0,\cdots,7)$, then we can define a toric surface as shown in Fig.~\ref{fig401:subfig:a}. This toric surface does not have linear precision, but we can tune it to achieve linear precision. We set $\widetilde{{\mathcal{A}}}=\{(0,2),(1,2),(0,\frac{6}{5}),(\frac{8}{7},\!\frac{8}{7}),\!(2,\!1),\!(0,\!0),\!(\frac{6}{5},\!0),\!(2,\!0)\}$ by moving the non-extreme points of $\mathcal {A}$ within the pentagon (Fig.~\ref{fig1:subfig:b}). The GT-B\'ezier surface constructed by $\widetilde{{\mathcal{A}}}$, $\omega$ and $\mathcal {B}$ has linear precision, as shown in Fig.~\ref{fig401:subfig:b}. The theoretical proof can be found in~\cite{ref19}.

\begin{figure}[h!]
  \centering
  \subfigure[Toric surface]{
    \label{fig401:subfig:a} 
    \includegraphics[width=6cm]{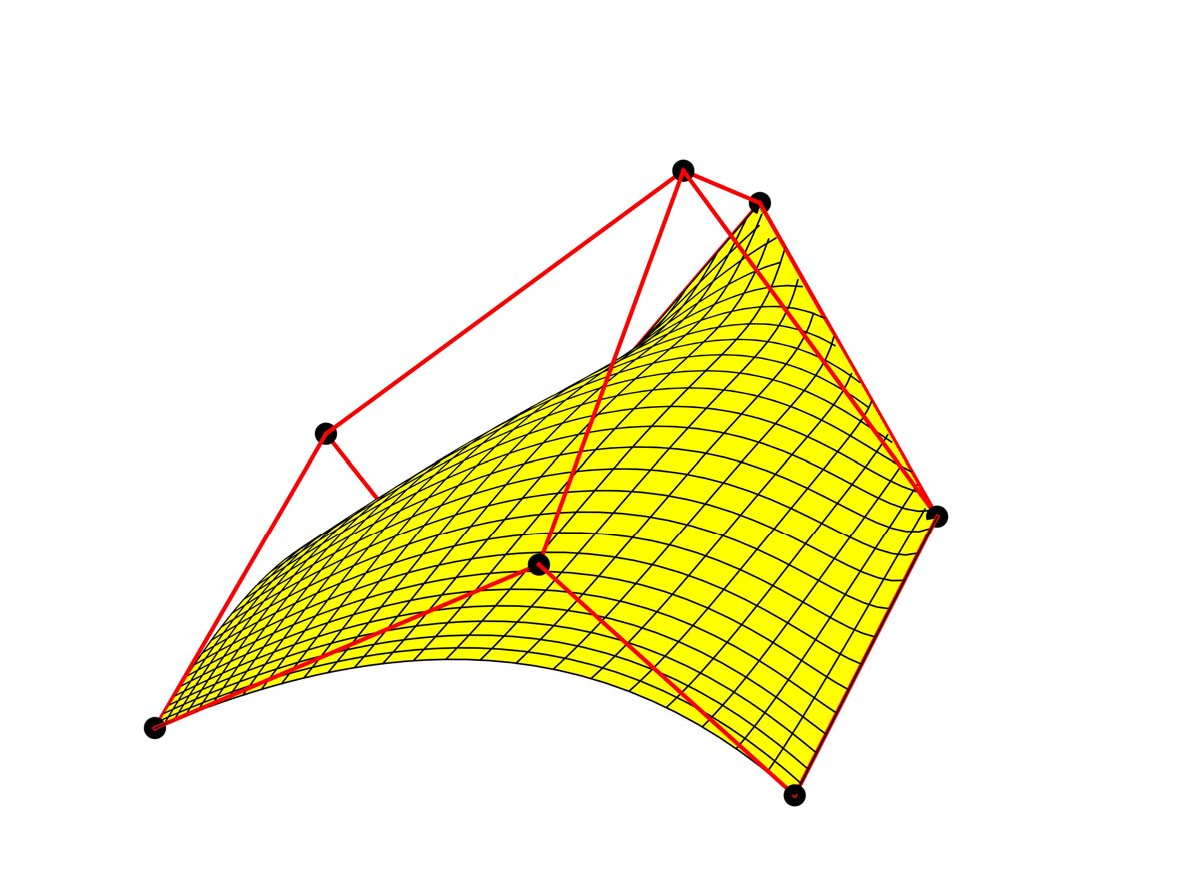}}
  \hspace{0.6cm}
  \subfigure[GT-B\'ezier Surface]{
    \label{fig401:subfig:b} 
    \includegraphics[width=6cm]{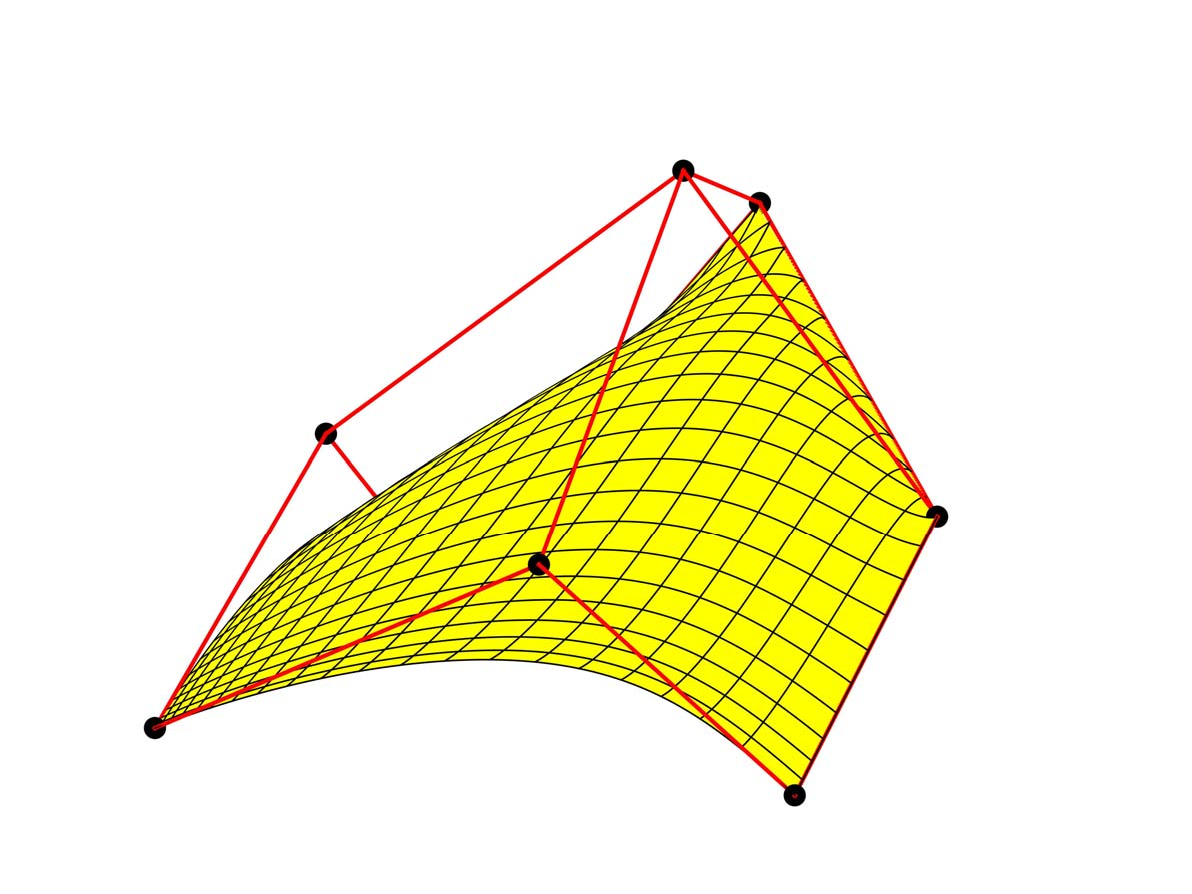}}
  \caption{Toric Surface and GT-B\'ezier Surface}
  \label{fig:401}
\end{figure}

\end{example}

From the properties of the GT-Bernstein basis functions, we have the following properties of the GT-B\'ezier surface.

\begin{description}
  \item[(a)] {\bf{Affine invariance and convex hull property}}. Since the basis functions (\ref{equ:205}) possess of nonnegativity and partition of unity, the corresponding GT-B\'ezier surface (\ref{equ:401}) has affine invariance and convex hull property.
  \item[(b)] {\bf{Degeneration property}}. When $\mathcal {A}\!\!=\!\!\{\!\mathbf{a}_{0},\!\mathbf{a}_{1}\!,\!\cdots\!,\!\mathbf{a}_{n}\!\}\!\!\subset\!\!\mathbb{Z}^{2}$, the GT-B\'ezier surface associated of $\mathcal {A}$ degenerates to the toric surface defined in~\cite{ref12} by the property of basis(\ref{equ:205}). In particular, the rational B\'ezier triangle defined by $\mathcal {A}=\{(i,j)\in \mathbb{Z}^{2}\mid i+j\leq k,i\geq0,j\geq0 \}$, and the rational tensor product B\'ezier surface defined by $\mathcal {A}=\{(i,j)\in \mathbb{Z}^{2}\mid 0\leq i\leq m,0\leq j\geq n\}$ are special cases of the GT-B\'ezier surface.
  \item[(c)] {\bf{Corner points interpolation property}}. This property follows directly from the property at the corner points property of the basis (\ref{equ:205}), that is $\mathbf{P}_{\mathcal {A},\omega,\mathcal{B}}(V_{i})=\mathbf{b}_{V_{i}},i=1,\cdots,r$, where $V_{i}\in \mathcal {A}$ are the vertices of $\Delta_{\mathcal {A}}$.
  \item[(d)] {\bf{Isoparametric curves property}}. The isoparametric curves $\mathbf{P}_{\mathcal {A},\omega,\mathcal{B}}(u^{\ast},v)$ and $\mathbf{P}_{\mathcal {A},\omega,\mathcal{B}}(u,v^{\ast})$ of a GT-B\'ezier surface are respectively the GT-B\'ezier curves.

      \begin{theorem}\label{the:bop}
      Each boundary of the GT-B\'ezier surface is a GT-B\'ezier curve $\mathbf{P}_{\hat{\phi_{i}},\omega|_{\hat{\phi_{i}}},\mathcal{B}|_{\hat{\phi_{i}}}}$, which defined by control points $\mathbf{b}_{\mathbf{a}_{i}}$ and weights $\omega_{\mathbf{a}_{i}}$ by $\mathbf{a}_{i}\in \hat{\phi_{i}}$ of corresponding edges $\phi_{i}\subset\Delta_{\mathcal {A}}$, where $i=1,\cdots,r$.
      \end{theorem}

      \begin{proof}
      Consider the restriction  $\mathbf{P}_{\hat{\phi},\omega|_{\hat{\phi}},\mathcal{B}|_{\hat{\phi}}}$ of the GT-B\'ezier surface at the fixed edge $\phi=\phi_{i}$ of $\Delta_{\mathcal {A}}$. Denote $V_{0}=(u_{0},v_{0})=V_{i-1}$, $V_{1}=(u_{1},v_{1})=V_{i}$, and $h_{i}(u,v)=h(u,v)=\xi u+\eta v+\rho$ is the equation of $\phi$ for simplicity. Let the angle between the edge $\phi$ and the $u$ axis be $\alpha$. Then $\tan\alpha=-\frac{\xi}{\eta}$, and
      \begin{equation*}
        \cos\alpha=\frac{\eta}{\sqrt{\xi^{2}+\eta^{2}}}.
      \end{equation*}
      Let $\sigma=|V_{0}V_{1}|=\sqrt{(u_{1}-u_{0})^{2}+(v_{1}-v_{0})^{2}}$.

      All basis functions $\beta_{\mathbf{a}_{i}}(u,v)$ with indices $\mathbf{a}_{i}\in\mathcal {A}\setminus\hat{\phi}$ vanishes if $(u,v)\in \phi$, hence $\mathbf{P}_{\hat{\phi},\omega|_{\hat{\phi}},\mathcal{B}|_{\hat{\phi}}}$ depends only on weights and control points indexed by $\mathbf{a}_{i}\in\hat{\phi}$. If $\mathbf{a}_{j}=(u_{j},v_{j})\in\hat{\phi}$, then $h(\mathbf{a}_{j})=\xi u_{j}+\eta v_{j}+\rho=0$, $v_{j}=-\frac{\rho+\xi u_{j}}{\eta}$. Let $l_{j}=|\mathbf{a}_{j}V_{0}|=\sqrt{(u_{j}-u_{0})^{2}+(v_{j}-v_{0})^{2}}$. By geometric relationship, we have
      \begin{equation*}
        u_{j}=u_{0}+l_{j}\cos\alpha
      \end{equation*}
      For the edge equation $h_{k}(u,v)=0$ for the edge $\phi_{k}$ of $\Delta_{\mathcal {A}}$, we evaluate $h_{k}(u,v)$ at point $\mathbf{a}_{j}$,
      \begin{eqnarray*}
       h_{k}(\mathbf{a}_{j})&\!\!=\!\!&  \xi_{k} u_{j}+\eta_{k} v_{j}+\rho_{k} \\
       &\!\! =\!\!&  \frac{\eta\xi_{k}-\xi\eta_{k}}{\eta}u_{j}+\rho_{k}-\frac{\eta_{k}}{\eta}\rho \\
       &\!\!=\!\!&\rho_{k}-\frac{\eta_{k}}{\eta}\rho+\frac{\eta\xi_{k}-\xi\eta_{k}}{\eta}u_{0}+\frac{(\eta\xi_{k}-\xi\eta_{k})\cos\alpha}{\eta}l_{j}.
      \end{eqnarray*}
      Thus the basis defined on the edge $\phi$ can be expressed as
      \begin{equation*}
          \beta_{\mathbf{a}_{i}}(u,v) =c_{\mathbf{a}_{i}}\prod_{k=1}^{r} h_{k}(u,v)^{(\rho_{k}-\frac{\eta_{k}}{\eta}\rho+\frac{\eta\xi_{k}-\xi\eta_{k}}{\eta}u_{0})}(h_{1}(u,v)^{\frac{(\eta\xi_{1}-\xi\eta_{1})\cos\alpha}{\eta}}\cdots h_{r}(u,v)^{\frac{(\eta\xi_{r}-\xi\eta_{r})\cos\alpha}{\eta}})^{l_{j}}.
      \end{equation*}
      Here the first $r$ factors $h_{k}(u,v)^{(\rho_{k}-\frac{\eta_{k}}{\eta}\rho+\frac{\eta\xi_{k}-\xi\eta_{k}}{\eta}u_{0})}$ do not depend on $j$ and can be canceled in the definition of GT-B\'ezier surface.

      When $(u,v)\in \phi$, then $h(u,v)=0$, $v=-\frac{\rho+\xi u}{\eta}$. So when $(u,v)\in \phi$, $h_{k}(u,v)$ is univariate function of $u$, written $h_{k}(u)$. If we set new variables
      \begin{equation*}
        s\!\!=\!\!h_{1}(u)^{\frac{(\eta\xi_{1}-\xi\eta_{1})\cos\alpha}{\eta}}\cdots h_{r}(u)^{\frac{(\eta\xi_{r}\!-\!\xi\eta_{r})\cos\alpha}{\eta}},~~~t=\frac{\sigma s}{1+s},
      \end{equation*}
      we obtain
         \begin{equation*}
           \mathbf{P}_{\hat{\phi},\omega|_{\hat{\phi}},\mathcal{B}|_{\hat{\phi}}}(u) = \frac{\sum_{\mathbf{a}_{j}\in\hat{\phi}}\omega_{\mathbf{a}_{j}}\mathbf{b}_{\mathbf{a}_{j}}c_{\mathbf{a}_{j}}s^{l_{j}}}{\sum_{\mathbf{a}_{j}\in\hat{\phi}}\omega_{\mathbf{a}_{j}}c_{\mathbf{a}_{j}}s^{l_{j}}} =\frac{\sum_{\mathbf{a}_{j}\in\hat{\phi}}\omega_{\mathbf{a}_{j}}\mathbf{b}_{\mathbf{a}_{j}}c_{\mathbf{a}_{j}}t^{l_{j}}(\sigma-t)^{\sigma-l_{j}}}{\sum_{\mathbf{a}_{j}\in\hat{\phi}}\omega_{\mathbf{a}_{j}}c_{\mathbf{a}_{j}}t^{l_{j}}(\sigma-t)^{\sigma-l_{j}}}.
         \end{equation*}
      We choose a natural parameter $\tau$ on the edge, $u=u_{0}+\tau\sigma\cos\alpha(0<\tau<1)$, to prove that this reparametrization is 1-1, and calculate derivatives
      \begin{eqnarray*}
        \frac{\mathrm{d}s}{\mathrm{d}\tau} &\!\!=\!\!& \frac{\mathrm{d}}{\mathrm{d}\tau}(h_{1}(u)^{\frac{(\eta\xi_{1}-\xi\eta_{1})\cos\alpha}{\eta}})\!\!\cdots h_{r}(u)^{\frac{(\eta\xi_{r}-\xi\eta_{r})\cos\alpha}{\eta}}\!+\!\\
         &&\cdots\!+\!h_{1}(u)^{\frac{(\eta\xi_{1}-\xi\eta_{1})\cos\alpha}{\eta}}\!\!\cdots \! \frac{\mathrm{d}}{\mathrm{d}\tau}(h_{r}(u)^{\frac{(\eta\xi_{r}-\xi\eta_{r})\cos\alpha}{\eta}}\!)\\
         &\!\!=\!\!& \sigma\!\prod_{k=1}^{r}\!h_{k}(u)^{\frac{(\eta\xi_{k}-\xi\eta_{k})\cos\alpha}{\eta}}\!\sum_{j=1}^{r}\frac{(\frac{\eta\xi_{j}-\xi\eta_{j}}{\eta}\cos\alpha)^{2}}{h_{j}(u)}\!>\!0,~~~0<\tau<1,
      \end{eqnarray*}
      and
      \begin{equation*}
        \frac{\mathrm{d}t}{\mathrm{d}\tau}=\frac{\mathrm{d}}{\mathrm{d}\tau}(\frac{\sigma s}{1+s})=\frac{\sigma}{(1+s)^{2}}\frac{\mathrm{d}s}{\mathrm{d}\tau}>0.
      \end{equation*}
      Hence the reparametrization $\tau\longmapsto t$ is monotonic. Also it is easy to check that it preserves endpoints. Therefore it is 1-1 and ends the proof.
      \end{proof}

   \item[(e)] {\bf{Multiple knot property}}. When a knot of $\mathcal {A}$ tends to its adjacent knot, the following theorem describes the limit property of the GT-B\'ezier surface, and demonstrates the construction of GT-B\'ezier surface by $\mathcal {A}$ with multiple knots.

      \begin{theorem}\label{the:401}
      Suppose $c_{\mathbf{a}_{i}}=1(i=0,\cdots,n)$. When the knot $\mathbf{a}_{k}~(0\leq k<n)$ approaches to $\mathbf{a}_{q}~(0\leq q<n, and~q\neq k )$ along line $\mathbf{a}_{k}\mathbf{a}_{q}$ with the convex hull $\Delta_{\mathcal{A}}$ unchanging, the limit of GT-B\'ezier surface $\mathbf{P}_{\mathcal {A},\omega,\mathcal{B}}(u,v)$ defined in (\ref{equ:401}) is exactly the GT-B\'ezier surface $\widetilde{\mathbf{P}}_{\widetilde{\mathcal {A}},\widetilde{\omega},\widetilde{\mathcal{B}}}(u,v)$, defined as
      \begin{equation}
         \lim_{\mathbf{a}_{k}\rightarrow \mathbf{a}_{q}}\mathbf{P}_{\mathcal {A},\omega,\mathcal{B}}(u,v)\!\! = \!\! \widetilde{\mathbf{P}}_{\widetilde{\mathcal {A}},\widetilde{\omega},\widetilde{\mathcal{B}}}(u,v)\!\!=\!\!\frac{\sum_{i\neq k,q}\omega_{\mathbf{a}_{i}}\mathbf{b}_{\mathbf{a}_{i}}\beta_{\mathbf{a}_{i}}(u,v)+\widetilde{\omega}_{\mathbf{a}_{q}}\widetilde{\mathbf{b}}_{\mathbf{a}_{q}}\beta_{\mathbf{a}_{q}}(u,v)}{\sum_{i\neq k,q}\omega_{\mathbf{a}_{i}}\beta_{\mathbf{a}_{i}}(u,v)+\widetilde{\omega}_{\mathbf{a}_{q}}\beta_{\mathbf{a}_{q}}(u,v)},
      \end{equation}
      where $\widetilde{\omega}_{\mathbf{a}_{q}}\!=\!\omega_{\mathbf{a}_{k}}\!+\omega_{\mathbf{a}_{q}},\widetilde{\mathbf{b}}_{\mathbf{a}_{q}}\!\!=\!\frac{\omega_{\mathbf{a}_{k}}}{\omega_{\mathbf{a}_{k}}\!\!+\omega_{\mathbf{a}_{q}}}\mathbf{b}_{\mathbf{a}_{k}}\!\!+\frac{\omega_{\mathbf{a}_{q}}}{\omega_{\mathbf{a}_{k}}+\omega_{\mathbf{a}_{q}}}\mathbf{b}_{\mathbf{a}_{q}}$, $\widetilde{\mathcal {A}}=\{\mathbf{a}_{0},\cdots,\mathbf{a}_{k-1},\mathbf{a}_{k+1},\cdots,\mathbf{a}_{n}\}, \widetilde{\mathcal{B}}=\{\mathbf{b}_{\mathbf{a}_{0}},\cdots,\mathbf{b}_{\mathbf{a}_{k-1}},\mathbf{b}_{\mathbf{a}_{k+1}},\cdots,\widetilde{\mathbf{b}}_{\mathbf{a}_{q}},\cdots,\mathbf{b}_{\mathbf{a}_{n}}\}$ and ~$\widetilde{\omega}\!=\!\{\omega_{\mathbf{a}_{0}},\cdots,\omega_{\mathbf{a}_{k-1}},\omega_{\mathbf{a}_{k+1}},\cdots,\widetilde{\omega}_{\mathbf{a}_{q}},\cdots,\omega_{\mathbf{a}_{n}}\}$.
      \end{theorem}

  \begin{proof}
      When $\mathbf{a}_{k}$ tends to $\mathbf{a}_{q}$, we have
      \begin{eqnarray*}
        \lim_{\mathbf{a}_{k}\rightarrow \mathbf{a}_{q}}\beta_{\mathbf{a}_{k}}(u,v) &\!\!=\!\!& \!\lim_{\mathbf{a}_{k}\rightarrow \mathbf{a}_{q}}h_{1}(u,v)^{h_{1}(\mathbf{a}_{k})} \cdots h_{r}(u,v)^{h_{r}(\mathbf{a}_{k})} \\
         &\!\!=\!\!& h_{1}(u,v)^{h_{1}(\mathbf{a}_{q})} \cdots h_{r}(u,v)^{h_{r}(\mathbf{a}_{q})} \\
         &\!\!=\!\!& \beta_{\mathbf{a}_{q}}(u,v).
      \end{eqnarray*}
      Thus,
      \begin{eqnarray*}
         \lim_{\mathbf{a}_{k}\rightarrow  \mathbf{a}_{q}}\sum_{i=0}^{n}\mathbf{b}_{\mathbf{a}_{i}}T_{\mathbf{a}_{i}}(u,v) &\!\!= \!\!& \widetilde{\mathbf{P}}_{\widetilde{\mathcal {A}},\widetilde{\omega},\widetilde{\mathcal{B}}}(u,v)\\
         &\!\!\!=\!\!\!& \lim_{\mathbf{a}_{k}\rightarrow  \mathbf{a}_{q}}\sum_{i=0}^{n}\frac{\mathbf{b}_{\mathbf{a}_{i}}\omega_{\mathbf{a}_{i}}\beta_{\mathbf{a}_{i}}(u,v)}{\sum^{n}_{i=0}\omega_{\mathbf{a}_{i}}\beta_{\mathbf{a}_{i}}(u,v)}\\
         &\!\!=\!\!& \frac{\sum_{i\neq k,q}\omega_{\mathbf{a}_{i}}\mathbf{b}_{\mathbf{a}_{i}}\beta_{\mathbf{a}_{i}}\!(u,\!v)\!\!+\!\!\omega_{\mathbf{a}_{k}}\mathbf{b}_{\mathbf{a}_{k}}\beta_{\mathbf{a}_{q}}\!(u,\!v)\!\!+\!\!\omega_{\mathbf{a}_{q}}\!\mathbf{b}_{\mathbf{a}_{q}}\beta_{\mathbf{a}_{q}}\!(u,\!v)}{\sum_{i\neq k,q}\omega_{\mathbf{a}_{i}}\beta_{\mathbf{a}_{i}}(u,v)\!\!+\!\!\omega_{\mathbf{a}_{k}}\beta_{\mathbf{a}_{q}}(u,v)\!\!+\!\!\omega_{\mathbf{a}_{q}}\beta_{\mathbf{a}_{q}}(u,v)}\\
         &\!\!=\!\!& \frac{\sum_{i\neq k,q}\omega_{\mathbf{a}_{i}}\mathbf{b}_{\mathbf{a}_{i}}\beta_{\mathbf{a}_{i}}(u,v)\!+\!(\omega_{\mathbf{a}_{k}}\mathbf{b}_{\mathbf{a}_{k}}\!+\!\omega_{\mathbf{a}_{q}}b_{\mathbf{a}_{q}})\beta_{\mathbf{a}_{q}}(u,v)}{\sum_{i\neq k,q}\omega_{\mathbf{a}_{i}}\beta_{\mathbf{a}_{i}}(u,v)\!+\!(\omega_{\mathbf{a}_{k}}+\omega_{\mathbf{a}_{q}})\beta_{\mathbf{a}_{q}}(u,v)}.
      \end{eqnarray*}
      Let $\widetilde{\omega}_{\mathbf{a}_{q}}=\omega_{\mathbf{a}_{k}}+\omega_{\mathbf{a}_{q}}$, $\widetilde{\mathbf{b}}_{\mathbf{a}_{q}}=\frac{\omega_{\mathbf{a}_{k}}}{\omega_{\mathbf{a}_{k}}+\omega_{\mathbf{a}_{q}}}\mathbf{b}_{\mathbf{a}_{k}}+\frac{\omega_{\mathbf{a}_{q}}}{\omega_{\mathbf{a}_{k}}+\omega_{\mathbf{a}_{q}}}\mathbf{b}_{\mathbf{a}_{q}}$,
      we can obtain
    \begin{equation*}
       \lim_{\mathbf{a}_{k}\rightarrow \mathbf{a}_{q}}\mathbf{P}_{\mathcal {A},\omega,\mathcal{B}}(u,v)\!\! = \!\! \widetilde{\mathbf{P}}_{\widetilde{\mathcal {A}},\widetilde{\omega},\widetilde{\mathcal{B}}}(u,v) \!\!=\!\!\frac{\sum_{i\neq k,q}\omega_{\mathbf{a}_{i}}\mathbf{b}_{\mathbf{a}_{i}}\beta_{\mathbf{a}_{i}}(u,v)+\widetilde{\omega}_{\mathbf{a}_{q}}\widetilde{\mathbf{b}}_{\mathbf{a}_{q}}\beta_{\mathbf{a}_{q}}(u,v)}{\sum_{i\neq k,q}\omega_{\mathbf{a}_{i}}\beta_{\mathbf{a}_{i}}(u,v)+\widetilde{\omega}_{\mathbf{a}_{q}}\beta_{\mathbf{a}_{q}}(u,v)},
    \end{equation*}
  where $\widetilde{\mathcal {A}}=\{\mathbf{a}_{0},\cdots,\mathbf{a}_{k-1},\mathbf{a}_{k+1},\cdots,\mathbf{a}_{n}\}$, $\widetilde{\mathcal{B}}=\{\mathbf{b}_{\mathbf{a}_{0}},\cdots,\mathbf{b}_{\mathbf{a}_{k-1}},\mathbf{b}_{\mathbf{a}_{k+1}},\cdots,\widetilde{\mathbf{b}}_{\mathbf{a}_{q}},\cdots,\mathbf{b}_{\mathbf{a}_{n}}\}$ and ~$\widetilde{\omega}=\{\omega_{\mathbf{a}_{0}},\cdots,\omega_{\mathbf{a}_{k-1}},\omega_{\mathbf{a}_{k+1}},\cdots,\widetilde{\omega}_{\mathbf{a}_{q}},\cdots,\omega_{\mathbf{a}_{n}}\}$.
  \end{proof}

  \begin{example}\label{exa:402}
  Consider the GT-B\'ezier surface defined in Example \ref{exa:401}. Let $\mathcal{A}\!=\!\{(0,2),(1,2)$\\
  $,(0,\frac{6}{5}),(\frac{8}{7},\frac{8}{7}),(2,1),(0,0),(\frac{6}{5},0),(2,0)\}$, control points $\mathcal {B}\!=\!\{(0,2,0),(1,2,4),(0,\frac{6}{5},2),(\frac{8}{7},\frac{8}{7},5)$\\
  $,(2,1,2),(0,0,0),(\frac{6}{5},0,2),(2,0,0)\}$, weights $\omega=\{2,2,5,7,2,3,5,2\}$ and $c_{\mathbf{a}_{i}}=1~(i=0,\cdots,7)$. If $\mathbf{a}_{3}=(\frac{8}{7},\frac{8}{7})$ approaches $\mathbf{a}_{1}=(1,2)$, then the changes of the GT-B\'ezier surface are shown in Fig.~\ref{fig:403}.

  Since the shape of the convex hull $\Delta_{\mathcal{A}}$ and control points $\mathcal {B}$ are unchanging during the process of $\mathbf{a}_{3}$ tending to $\mathbf{a}_{1}$, the original curved surface is stretched like an elastic film by the boundary property of the GT-B\'ezier surface. Until $\mathbf{a}_{3}=\mathbf{a}_{1}$, the resulting surface is defined by $\widetilde{\mathcal{A}}=\{(0,2),(1,2),(0,\frac{6}{5}),(2,1),(0,0),(\frac{6}{5},0),(2,0)\!\}$, control points $\widetilde{\mathcal {B}}\!\!=\!\!\{(0,2,0),(\frac{10}{9},\frac{12}{9},\frac{43}{9}),\!(0,\frac{6}{5},2),\!(2,1,2),\!(0,0,0),\!(\frac{6}{5},0,2),\!(2,0,0)\!\}$, weights $\widetilde{\omega}=\{2,9,5,2,3,5,2\}$.

  \begin{figure}[h!]
  \centering
  \subfigure[$\mathbf{a}_{3}=(\frac{8}{7},\frac{8}{7})$]{
    \includegraphics[width=3.7cm]{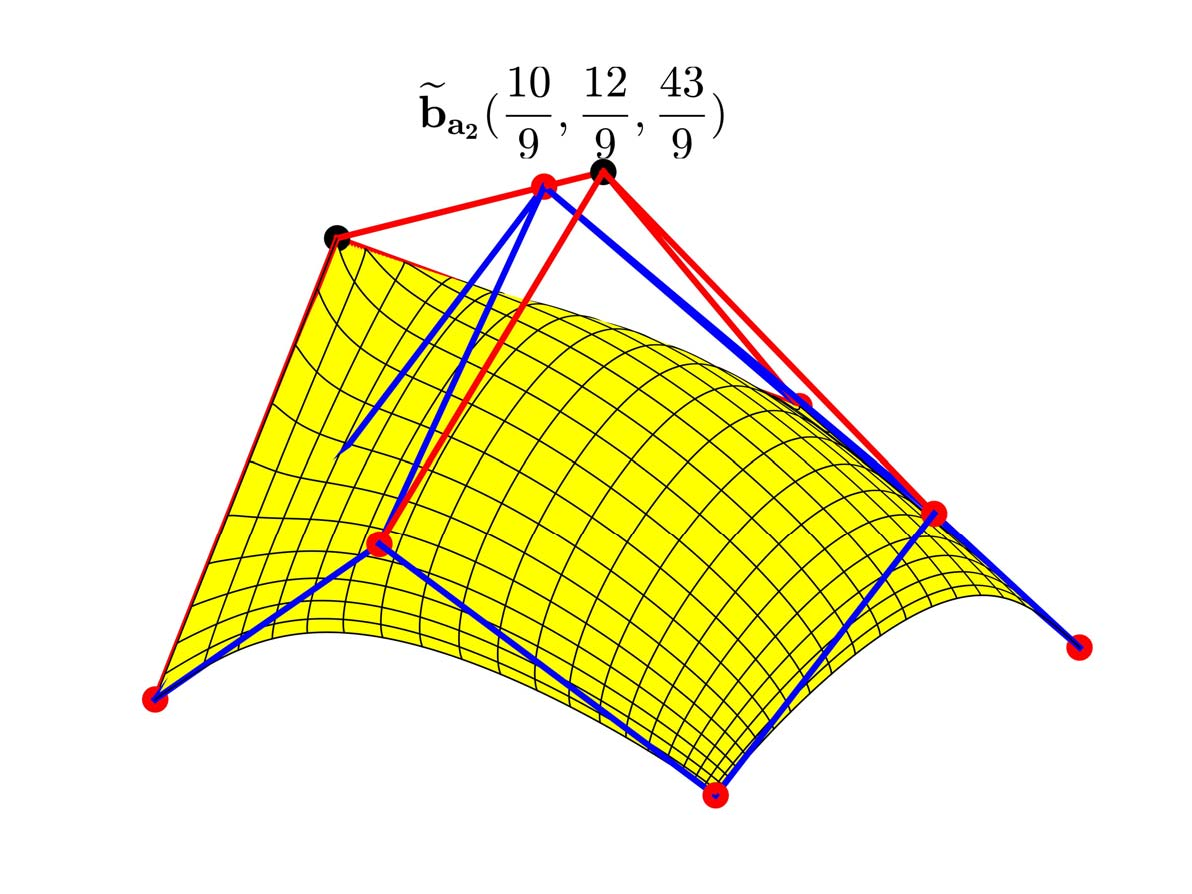}}
  \hspace{0.1cm}
  \subfigure[$\mathbf{a}_{3}=(\frac{38}{35},\frac{52}{35})$]{
    \includegraphics[width=3.7cm]{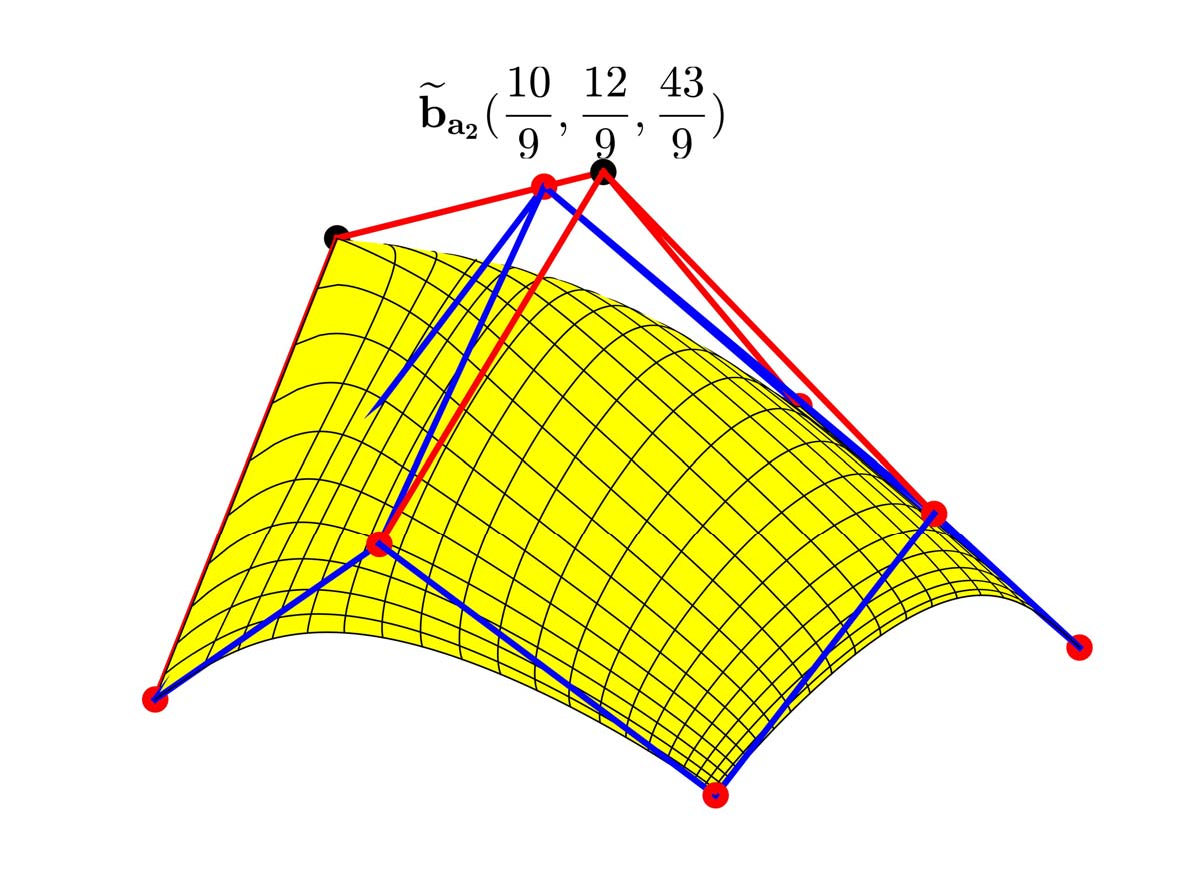}}
 \hspace{0.1cm}
  \subfigure[$\mathbf{a}_{3}=(\frac{36}{35},\frac{64}{35})$]{
    \includegraphics[width=3.7cm]{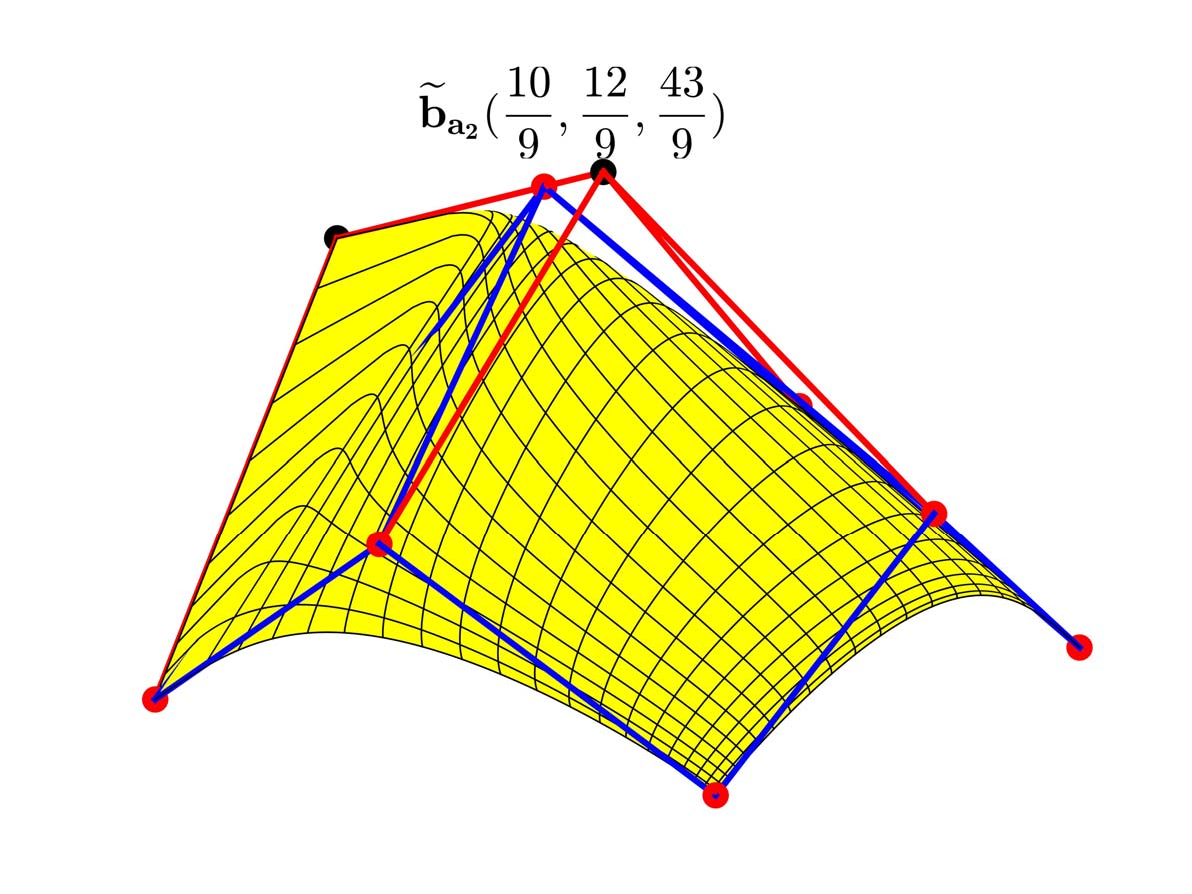}}
  \hspace{0.1cm}
  \subfigure[Limit surface]{
    \includegraphics[width=3.7cm]{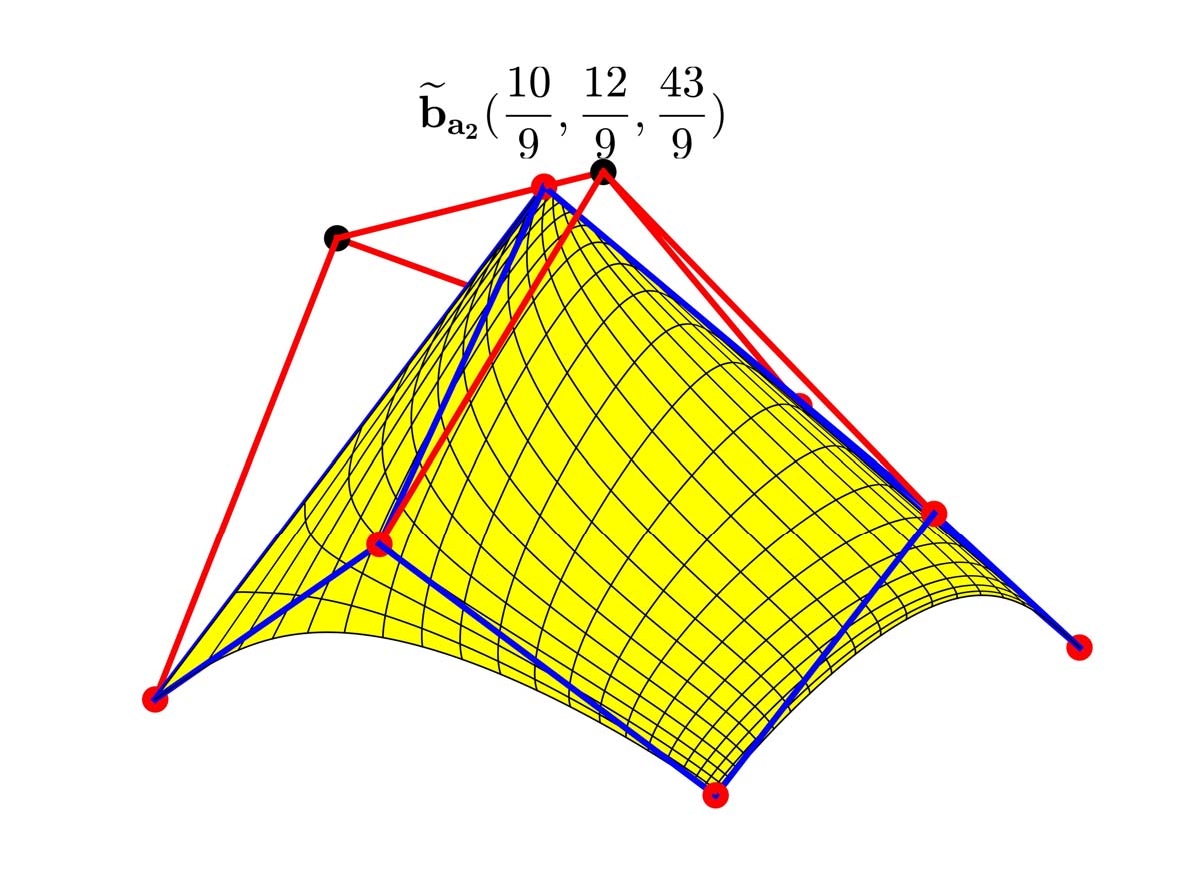}}
  \caption{Limit of GT-B\'ezier surface with $\mathbf{a}_{3}\rightarrow  \mathbf{a}_{1}$}
  \label{fig:403}
\end{figure}
\end{example}

Theorem~\ref{the:401} shows that the limit surface of the GT-B\'ezier surface is defined by the $\widetilde{\mathcal {A}}=\{\mathbf{a}_{0},\cdots,\mathbf{a}_{k-1},\mathbf{a}_{k+1},\cdots,\mathbf{a}_{n}\}$, $\widetilde{\mathcal{B}}=\{\mathbf{b}_{\mathbf{a_{0}}},\cdots,\mathbf{b}_{\mathbf{a}_{k-1}},\mathbf{b}_{\mathbf{a}_{k+1}},\cdots,\widetilde{\mathbf{b}}_{\mathbf{a}_{q}},\cdots,\mathbf{b}_{\mathbf{a}_{n}}\}$ and ~$\widetilde{\omega}=\{\omega_{\mathbf{a_{0}}},\cdots,\omega_{\mathbf{a}_{k-1}},\omega_{\mathbf{a}_{k+1}},\cdots,\widetilde{\omega}_{\mathbf{a}_{q}},\cdots,\omega_{\mathbf{a}_{n}}\}$  when knot $\mathbf{a}_{k}~(0\leq k<n)$ approaches to $\mathbf{a}_{q}~(0\leq q<n, and~q\neq k )$  without the convex hull $\Delta_{\mathcal{A}}$ changing. When multiple knots appear without changing the convex hull $\Delta_{\mathcal{A}}$, we only need to treat it by Theorem~\ref{the:401} repeatedly.

 \item[(f)]  {\bf{Toric degeneration property}}. Similarly, let $\lambda :\mathcal{A}\rightarrow \mathbb{R}$ be a lifting function to lift the points $\mathbf{a}_{i}$ of $\mathcal{A}$ to $(\mathbf{a}_{i},\lambda(\mathbf{a}_{i}))\in \mathbb{R}^{3}$. We denote $P_{\lambda}=conv\{(\mathbf{a}_{i},\lambda(\mathbf{a}_{i}))\mid \mathbf{a}_{i}\in \mathcal{A}\}$ the convex hull of the lifted points. Each face of the convex hull $P_{\lambda}$ has a normal vector pointing to the outer side. We call it the upper face of $P_{\lambda}$ if the last coordinate of the normal vector is positive. If we project these upper faces back vertically into $\mathbb{R}^{2}$, they can cover $\Delta_{\mathcal{A}}$ and form a regular subdivision $\Gamma_{\lambda}$ of $\Delta_{\mathcal{A}}$ induced by $\lambda$ (see \cite{ref13}).

  We group together the points of $\mathcal{A}$ that are in the same subset of the $\Gamma_{\lambda}$ and on the same upper face of the $P_{\lambda}$. Then we get a decomposition of $\mathcal{A}$, which is called regular decomposition $\mathcal{S}_{\lambda}$ of $\mathcal{A}$ induced by $\lambda$. For each subset $\mathcal{F}$ of $\mathcal{S}_{\lambda}$, we can use the weights $\omega|_{\mathcal{F}}=\{\omega_{\mathbf{a}_{i}}\mid \mathbf{a}_{i}\in\mathcal{F}\}$ and the control points $\mathcal{B}|_{\mathcal{F}}=\{\mathbf{b}_{\mathbf{a}_{i}}\mid \mathbf{a}_{i}\in\mathcal{F}\}$ to define a new GT-B\'ezier surface $\mathbf{P}_{\mathcal{F},\omega|_{\mathcal{F}},\mathcal{B}|_{\mathcal{F}}}$ on $\Delta_{\mathcal{F}}=conv\{\mathbf{a}_{i}\in\mathcal{F}\}$ by Definition \ref{def:401}. The union of these patches
  \begin{equation*}
    \mathbf{P}_{\mathcal {A},\omega,\mathcal{B}}(\mathcal{S}_{\lambda})=\bigcup_{\mathcal{F}\in \mathcal{S}_{\lambda}}\mathbf{P}_{\mathcal{F},\omega|_{\mathcal{F}},\mathcal{B}|_{\mathcal{F}}}
  \end{equation*}
   is called the regular control surface of $\mathbf{P}_{\mathcal {A},\omega,\mathcal{B}}$ induced by regular decomposition $\mathcal{S}_{\lambda}$.

   We can use lifting function $\lambda$ to get a set of weights with a parameter $x$, $\omega_{\lambda}(x):=\{x^{\lambda(\mathbf{a}_{i})}\omega_{\mathbf{a}_{i}}\mid \mathbf{a}_{i}\in\mathcal{A}\}$. These weights are used to define the map
  \begin{equation}\label{equ:403}
    \mathbf{P}_{\mathcal {A},\omega_{\lambda}(x),\mathcal{B}}(u,v) \!\!=\!\! \frac{\sum_{i=0}^{n}x^{\lambda(\mathbf{a}_{i})}\omega_{\mathbf{a}_{i}}\mathbf{b}_{\mathbf{a}_{i}}\beta_{\mathbf{a}_{i}}(u,v)}{\sum^{n}_{i=0}x^{{\lambda(\mathbf{a}_{i})}}\omega_{\mathbf{a}_{i}}\beta_{\mathbf{a}_{i}}(u,v)},~~~~~(u,v)\in \Delta_{\mathcal{A}} .
  \end{equation}
  The image of $\Delta_{\mathcal{A}}$ under this map is a GT-B\'ezier surface with a parameter $x$, denoted as $\mathbf{P}_{\mathcal {A},\omega_{\lambda}(x),\mathcal{B}}$.
  We have the following result.

  \begin{theorem}\label{the:402}
  The limit of the GT-B\'ezier surface $\mathbf{P}_{\!\mathcal {A},\omega_{\lambda}(x),\mathcal{B}}$ as $x\rightarrow \infty$ is the regular control surface induced by the regular decomposition $\mathcal{S}_{\lambda}$, that is
  \begin{equation*}
    \lim_{x\rightarrow \infty}\mathbf{P}_{\mathcal {A},\omega_{\lambda}(x),\mathcal{B}}=\mathbf{P}_{\mathcal {A},\omega,\mathcal{B}}(\mathcal{S}_{\lambda}).
  \end{equation*}
  \end{theorem}

  \begin{proof}
  The proof of the theorem is similar to Theorem \ref{the:302} and will be omitted here.
  \end{proof}

  Theorem~\ref{the:402} describes the conclusion that the limit surface of the GT-B\'ezier surface is its regular control surface, and explains the geometric meaning of the limit surface of the GT-B\'ezier surface when all the weights tend to infinity. And this property is called toric degeneration of GT-B\'ezier surfaces.

  \begin{example}\label{exa:404}
  Given point set $\widetilde{\mathcal{A}}$ is shown in Fig.~\ref{fig1:subfig:b}, and the lifted values of $\widetilde{\mathcal {A}}$ by $\lambda$ are shown in Figure Fig.~\ref{fig404:subfig:b}. The upper hull and the subdivision of $\Delta_{\widetilde{\mathcal{A}}}$ by $\lambda$ are shown in Fig.~\ref{fig404:subfig:c}, and the regular decomposition $\mathcal{S}_{\lambda}$ is shown in Fig.~\ref{fig404:subfig:d}.

  Let control points $\mathcal {B}\!\!=\!\!\{\!(0,\!2,\!0),\!(1,\!2,\!4),\!(0,\frac{6}{5},\!2),\!(\frac{8}{7},\!\frac{8}{7},\!18),(2,1,2),\!(0,0,0),(\frac{6}{5},0,2),(2,0,0)\!\}$ and weights $\omega=\{2,2,5,7,2,3,5,2\}$ corresponding to $\widetilde{\mathcal{A}}$. The toric degeneration process of this GT-B\'ezier surface is shown in Fig.~\ref{fig:405}. This figure also shows the GT-B\'ezier surfaces for the parameters $x=5$, $x=100$ and $x=600$ respectively. As the parameter $x$ becomes larger, the GT-B\'ezier surface approaches its regular control surface in Fig.~\ref{fig405:subfig:d} (consists of surface patches defined by three triangles and two quadrilaterals).
  \begin{figure}[h!]
  \centering
  \subfigure[$\lambda$]{
    \label{fig404:subfig:b} 
    \includegraphics[width=3.7cm]{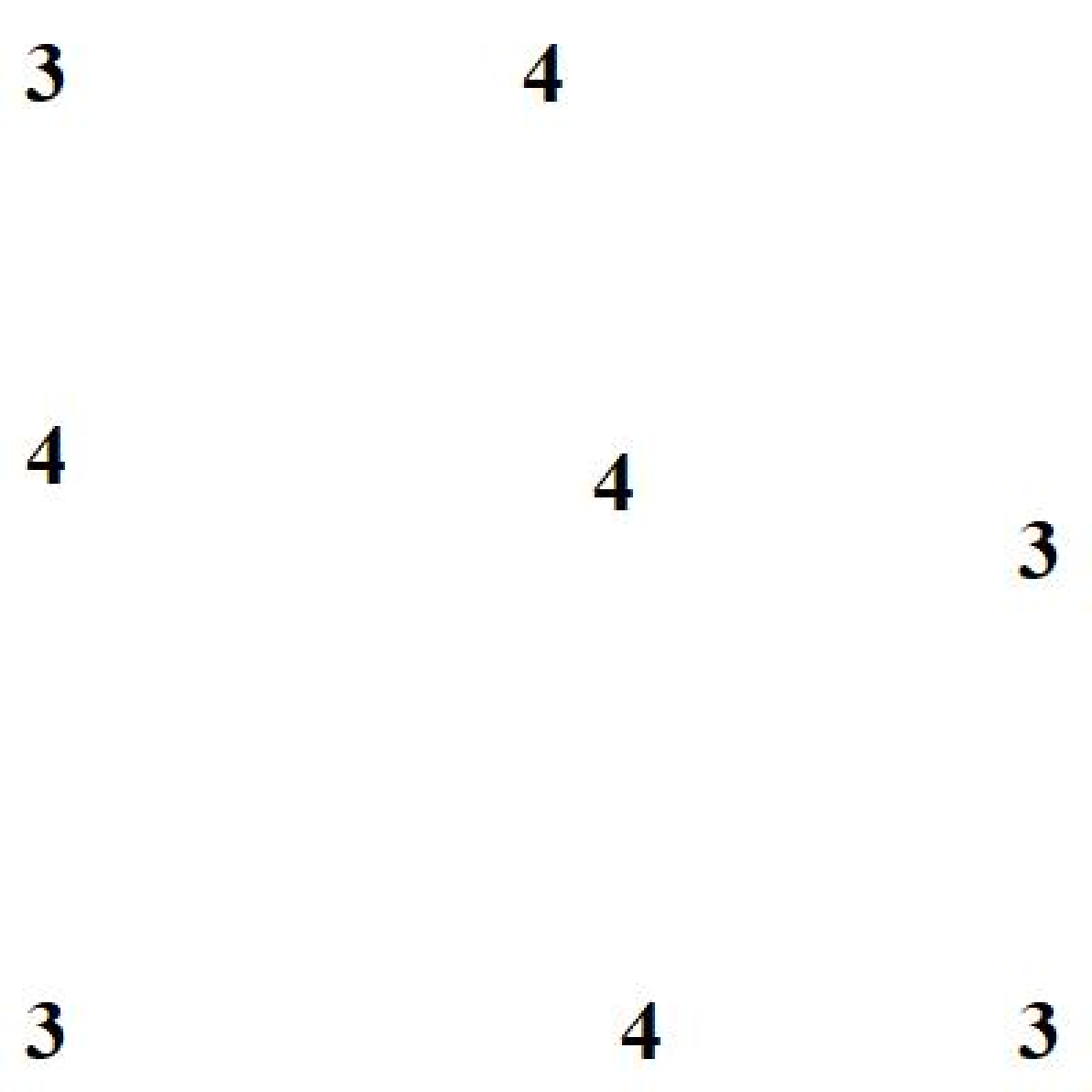}}
   \hspace{0.5cm}
  \subfigure[Upper hull and projection]{
    \label{fig404:subfig:c} 
    \includegraphics[width=4.2cm]{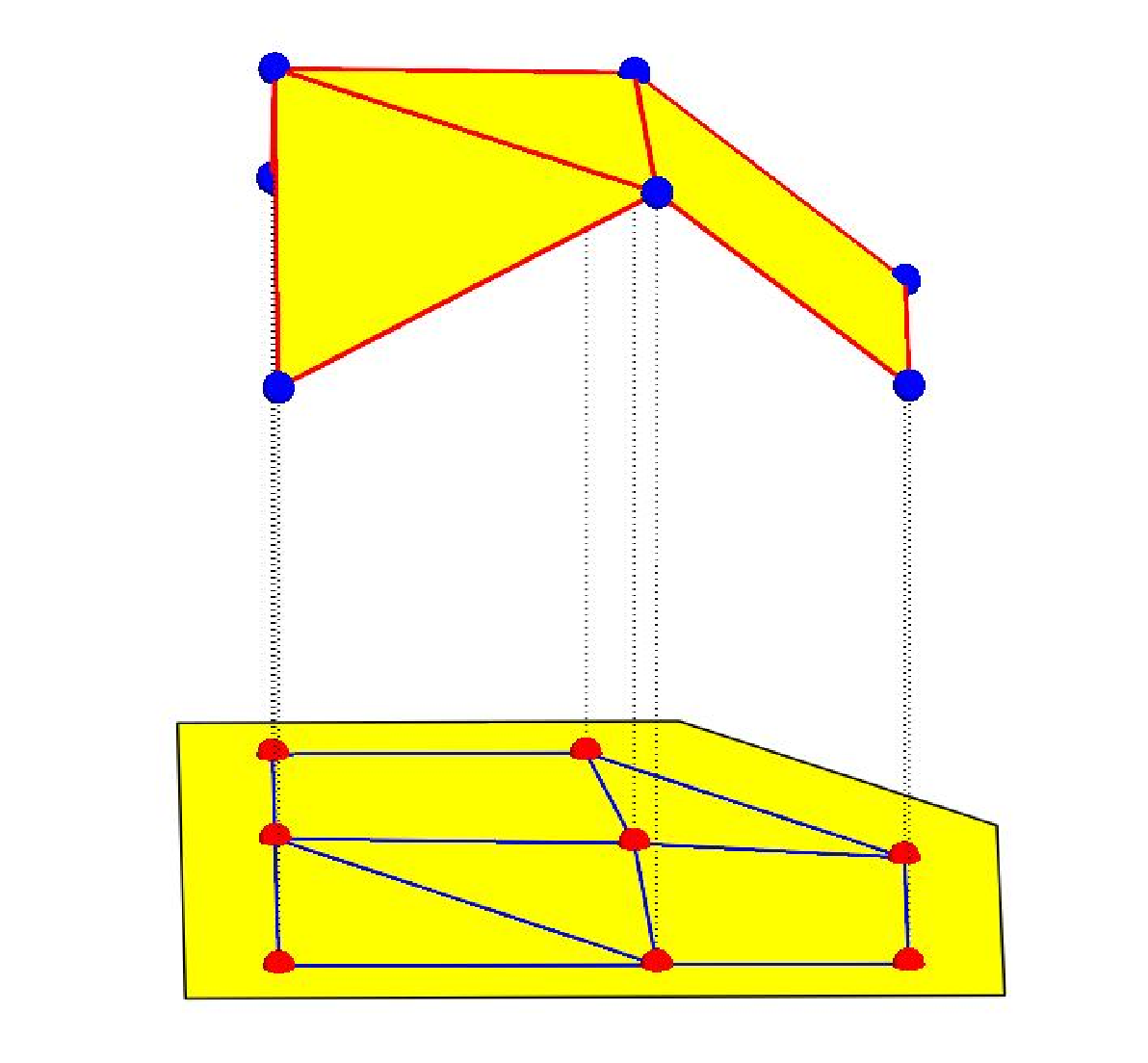}}
   \hspace{0.5cm}
  \subfigure[Decomposition]{
    \label{fig404:subfig:d} 
    \includegraphics[width=3.7cm]{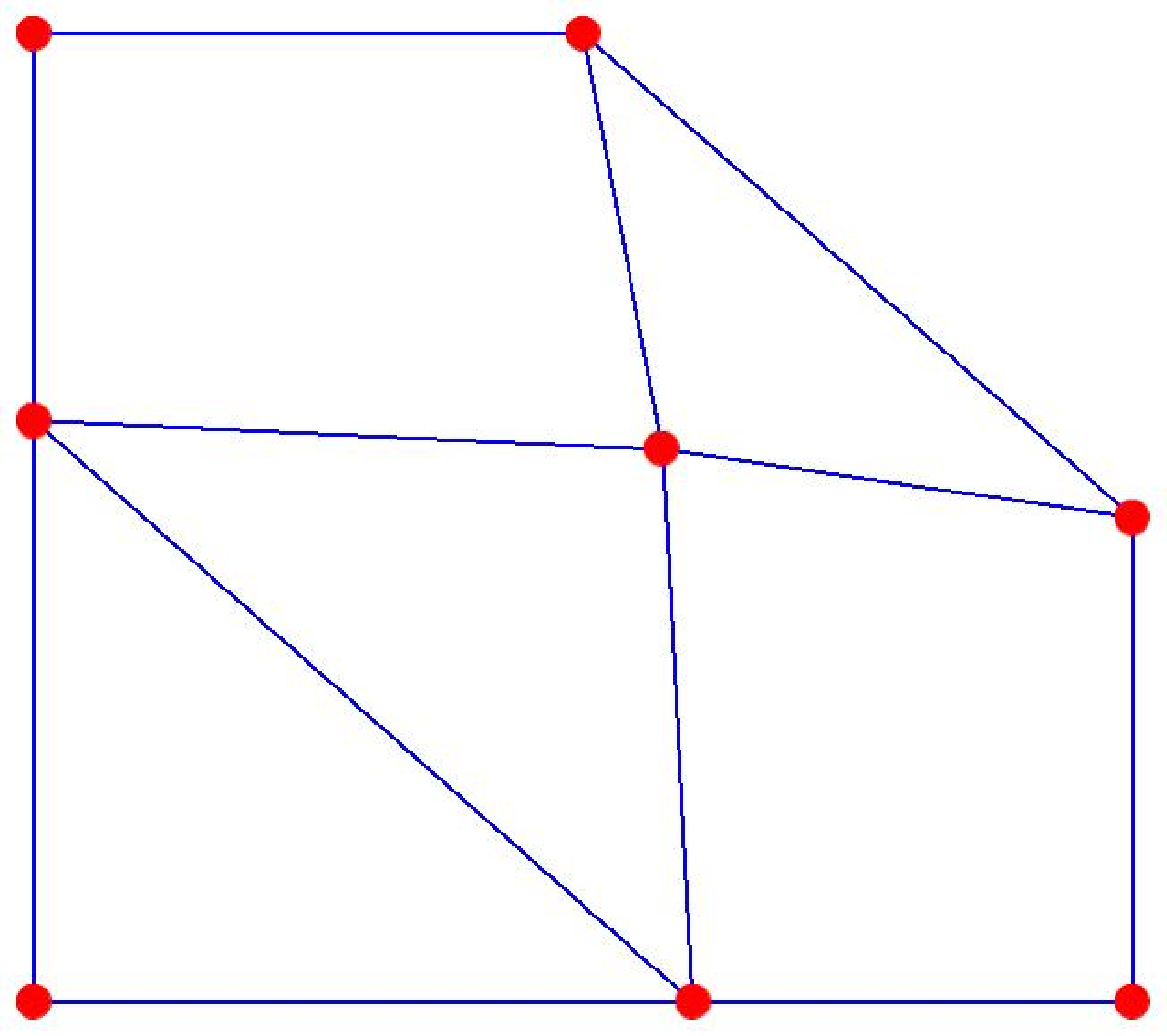}}
  \caption{Regular decomposition }
  \label{fig:404}
  \end{figure}
  \begin{figure}[h!]
  \centering
  \subfigure[$x=5$]{
    \label{fig405:subfig:a} 
    \includegraphics[width=3.7cm]{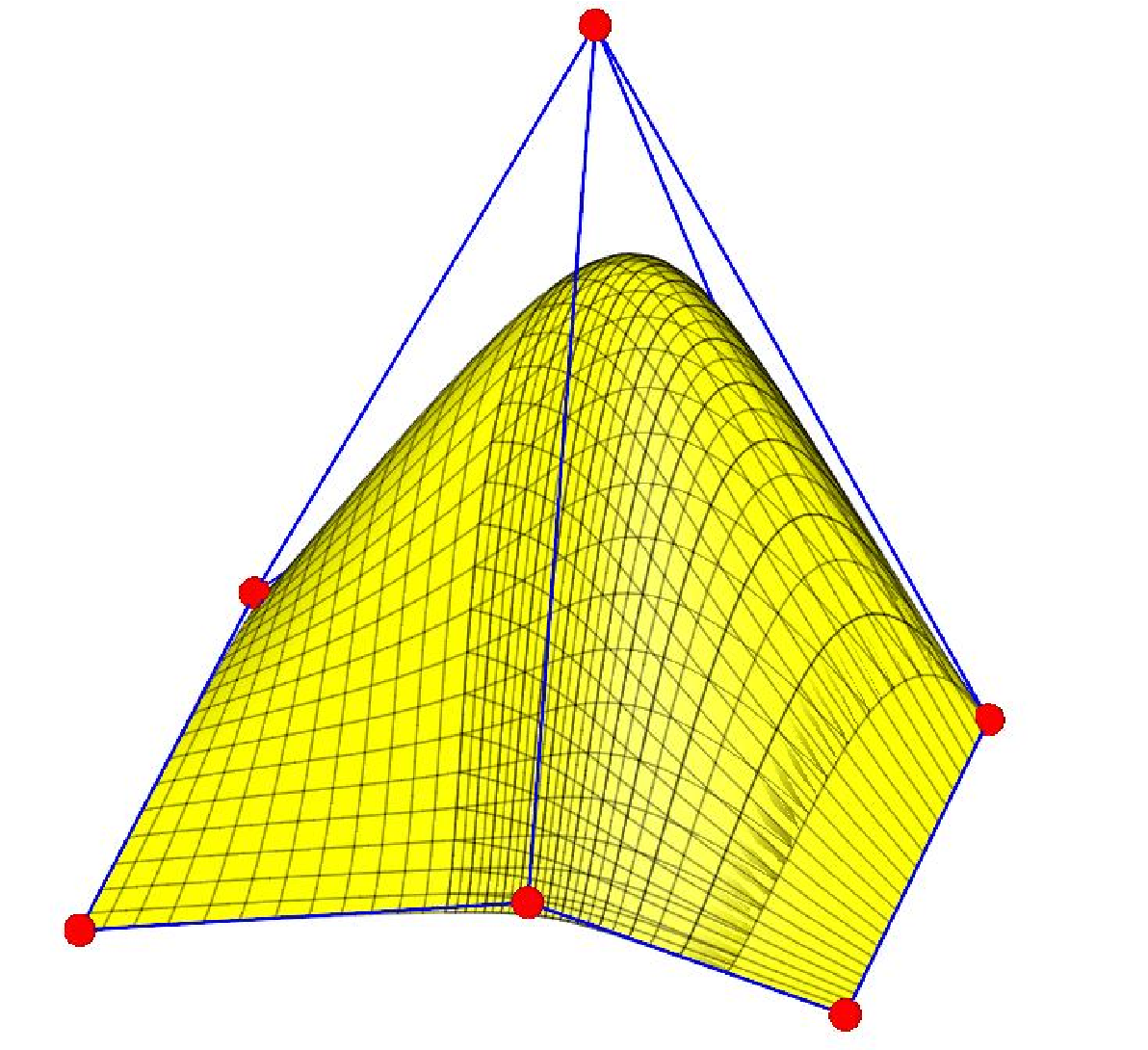}}
  \hspace{0.1cm}
  \subfigure[$x=100$]{
    \label{fig405:subfig:b} 
    \includegraphics[width=3.7cm]{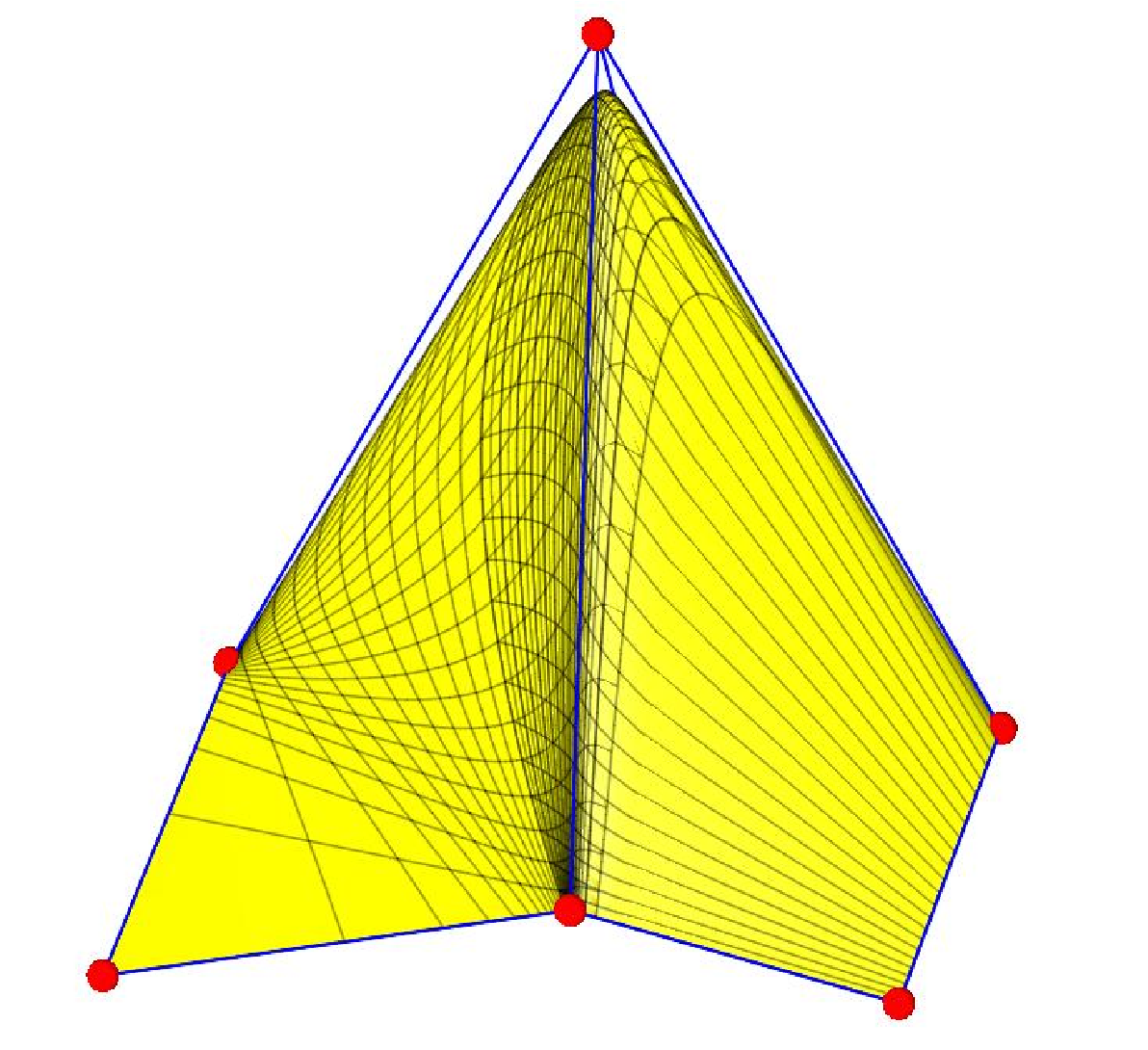}}
   \hspace{0.1cm}
  \subfigure[$x=600$]{
    \label{fig405:subfig:c} 
    \includegraphics[width=3.7cm]{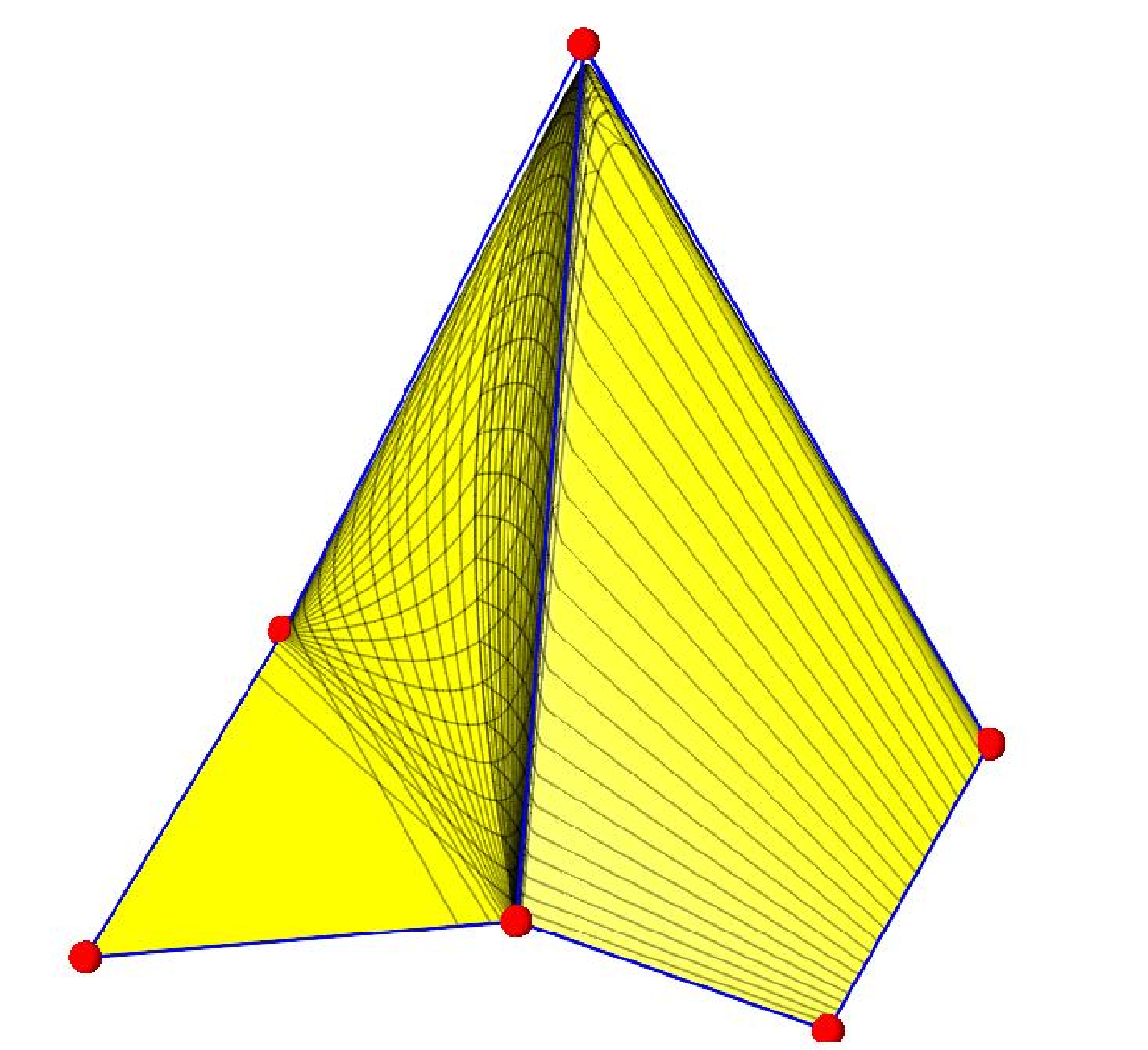}}
\hspace{0.1cm}
  \subfigure[Regular control surface]{
    \label{fig405:subfig:d} 
    \includegraphics[width=3.7cm]{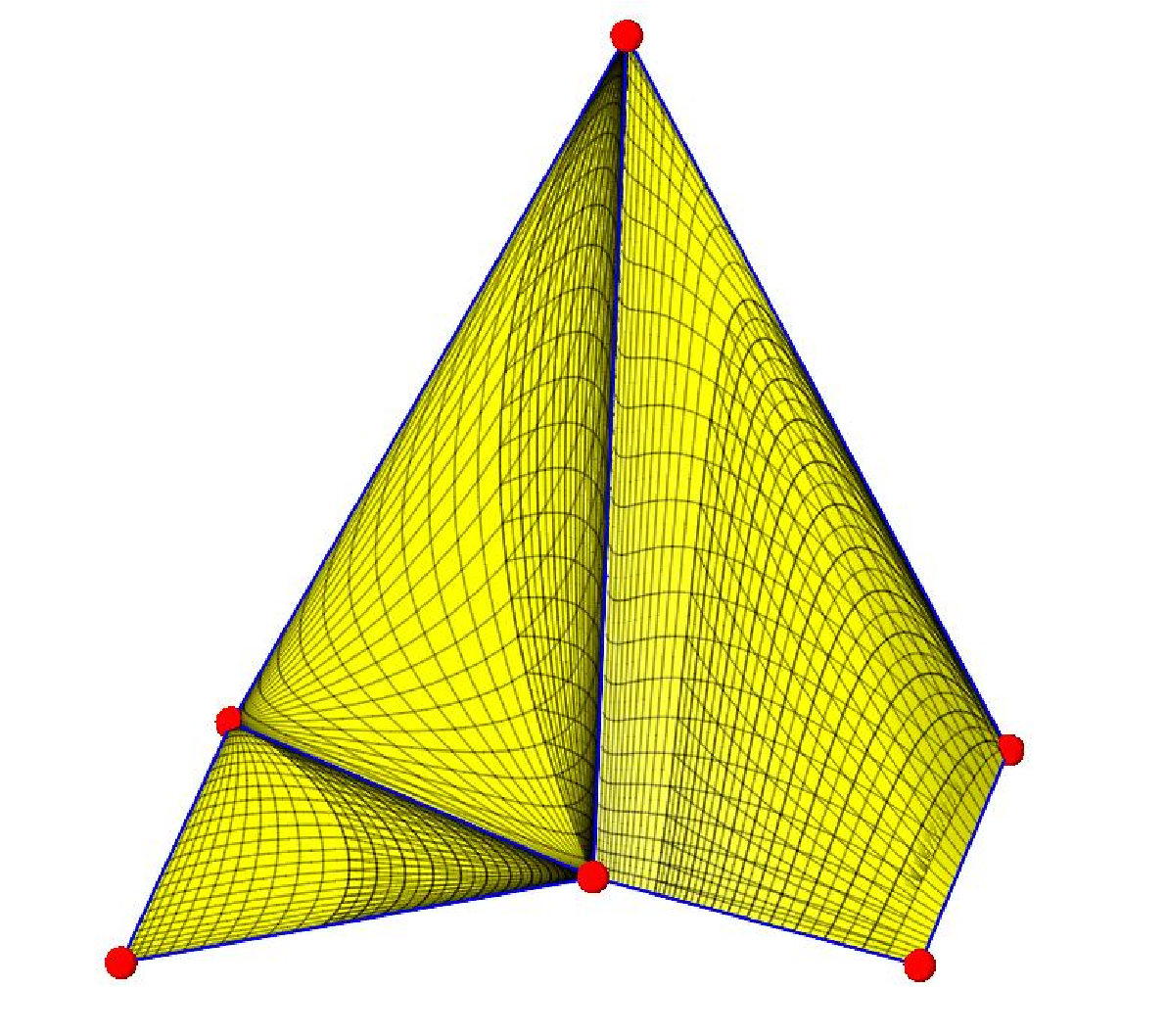}}
  \caption{Toric degeneration of GT-B\'ezier surface}
  \label{fig:405}
  \end{figure}
  \end{example}
 \end{description}
\section{Conclusion and future work}
\label{sec:4}
In this paper, we define a new kind of blending functions, called GT-Bernstein Basis Functions associated with a real number set. And then, we define a new kind of parametric curve and multisided surface based on the GT-Bernstein basis functions, which are the generalizations of the classical rational B\'ezier curves and surfaces, and toric surface patches. We indicate that the GT-B\'ezier curve and surface we presented partially preserve the properties of rational B\'ezier curves and surfaces. Finally, we also present the limiting properties of weights and knots.

Our further work will be devoted to elevation algorithm and de Casteljau algorithm of GT-B\'ezier curves and surfaces. In addition, the basis defined by the general real number knots limits the application range of the curve and surface. At present, only the MQ radial basis functions are investigated deeply in theories and applications, and their degrees are only in rational form. In this paper, we present the definition and study the properties of curves and surfaces theoretically only. How to apply the curves and surfaces for related subjects is our work in future too.
\section*{Acknowledgements}
This work is partly supported by the National Natural Science Foundation of China (Nos. 11671068, 11801053).

\section*{References}
\bibliographystyle{elsarticle-num}

\end{document}